\documentclass[3p]{elsarticle}
\usepackage[titletoc]{appendix}
\usepackage[normalem]{ulem}
\usepackage{amsmath,amssymb,amsthm}
\usepackage{mathtools}
\usepackage{proof}
\usepackage{array}
\usepackage{hyperref}
\usepackage{multicol}

\theoremstyle{plain}
\newtheorem{lemma}{Lemma}
\newtheorem{theorem}{Theorem}
\newtheorem{proposition}{Proposition}
\newtheorem{corollary}{Corollary}

\theoremstyle{definition}
\newtheorem{definition}{Definition}
\newtheorem{construction}{Construction}
\newtheorem{problem}{Problem}

\theoremstyle{remark}
\newtheorem{example}{Example}
\newtheorem{remark}{Remark}

\usepackage{tikz}
\usepackage{tikz-cd}
\usetikzlibrary{patterns,arrows, topaths, calc, positioning}
\tikzset{>=stealth}
\tikzset{>=stealth}
\tikzstyle{vertex} = [circle, minimum size = 1.5mm, inner sep = 0mm, draw={black}, fill]
\tikzstyle{hyperedge} = [rectangle, minimum width = 5mm, minimum height = 5mm, draw, inner sep = 0mm]
\tikzstyle{hyperedgewide} = [rectangle, minimum width = 8mm, minimum height = 5mm, draw, inner sep = 0mm]
\tikzstyle{HG} = [align = center]
\tikzstyle{circledge} = [circle, minimum size = 7mm, inner sep = 0mm, color=black, draw]
\tikzstyle{HG} = [align = center]
\tikzstyle{circledge} = [circle, minimum size = 7mm, inner sep = 0mm, color=black, draw]

\pgfdeclarepatternformonly{NELines}{\pgfqpoint{-1pt}{-1pt}}{\pgfqpoint{4pt}{4pt}}{\pgfqpoint{3pt}{3pt}}%
{
	\pgfsetlinewidth{0.4pt}
	\pgfpathmoveto{\pgfqpoint{0pt}{0pt}}
	\pgfpathlineto{\pgfqpoint{1pt}{1pt}}
	\pgfpathmoveto{\pgfqpoint{2pt}{2pt}}
	\pgfpathlineto{\pgfqpoint{3.1pt}{3.1pt}}
	\pgfusepath{stroke}
}

\newcommand{\eqdef}{\mathrel{\mathop:}=}

\newcommand{\type}{\mathit{type}}
\newcommand{\rem}{\mathit{rem}}
\newcommand{\fr}{\mathit{fr}}
\newcommand{\Colors}{\mathfrak{X}}
\newcommand{\VC}{\mathrm{VC}}
\newcommand{\EC}{\mathrm{HE}}
\newcommand{\SigmaCol}{\mathfrak{S}}
\newcommand{\FCol}{\mathfrak{F}}
\newcommand{\Inc}{\mathrm{Inc}}
\newcommand{\Cond}{\mathrm{Cond}}
\newcommand{\Mem}{\mathrm{Mem}}
\newcommand{\Balance}{\mathit{Balance}}
\newcommand{\Nat}{\mathbb{N}}
\newcommand{\sqz}{\mathrm{sqz}}
\newcommand{\Dsg}{\mathrm{Dsg}}
\newcommand{\dsg}{\mathrm{dsg}}
\newcommand{\Zex}{Z^{\mathit{ex}}}
\newcommand{\FGex}{FG^{\mathit{ex}}}
\newcommand{\FGPex}{Gr^{\mathit{ex}}}
\newcommand{\Pex}{P^{\mathit{ex}}}
\newcommand{\Hex}{H^{\mathit{ex}}}
\newcommand{\Yex}{Y^{\mathit{ex}}}
\newcommand{\Cex}{C^{\mathit{ex}}}
\newcommand{\colex}{\mathit{col}^{\mathit{ex}}}

\definecolor{comment}{RGB}{0,0,0}
\definecolor{colorA}{RGB}{96,96,96}
\definecolor{colorB}{RGB}{208,208,208}
\definecolor{colorC}{RGB}{179,179,179}

\begin{document}

\begin{frontmatter}

\title{On Decidability and Expressive Power of Fusion Grammars}

\author[1]{Tikhon Pshenitsyn}
\ead{tpshenitsyn@mi-ras.ru}

\affiliation[1]{
	organization={Steklov Mathematical Institute of Russian Academy of Sciences},
	addressline={8 Gubkina Street},
	postcode={119991},
	city={Moscow},
	country={Russia}
}

\begin{abstract}
	We study algorithmic complexity and expressive power of fusion grammars, a novel formalism introduced in [Kreowski, Kuske, and Lye 2017], which extends hyperedge replacement grammars. In the first part of the work, we prove that the non-emptiness problem for fusion grammars and the membership problem for fusion grammars without markers and connectors are decidable and are in NEXPTIME. We introduce fusion grammars with bounded usage of markers and connectors and prove decidability of the membership problem for them as well. In the proofs, we develop the technique of hypergraph vertex colourings encoded in hyperedge labels and also the technique of evidence paths and their encodings. 
	\\
	In the second part of the work, we study the class of languages generated by connection-preserving fusion grammars. Namely, we prove Parikh's theorem for them, i.e.~we show that these languages are semilinear. 
\end{abstract}



\end{frontmatter}

\section{Introduction}\label{sec_introduction}

Fusion grammars are an approach to generating hypergraph languages introduced in \cite{KreowskiKL17}. They are inspired by various fusion processes occurring in nature and abstract science. The most important motivating example is fusion of DNA double strands according to the Watson-Crick complementarity; DNA fusion is a key operation used in DNA computing and, in particular, in Adleman's experiment along with certain filtering and screening procedures \cite{PaunRS98}. 

Fusion grammars were introduced quite recently and they were studied in \cite{KreowskiKL17,Lye21}. There are several extensions of fusion grammars, e.g.~splicing/fusion grammars \cite{KreowskiKL18} or context-sensitive fusion grammars \cite{Lye22}, however, we shall be concerned with the basic formalism. An extensive survey of fusion grammars is Aaron Lye's PhD thesis \cite{Lye21_thesis}. A fusion grammar is defined by a start hypergraph $Z$ which consists of several connected components $Z_1,\dotsc,Z_k$. Hypergraphs have labels on hyperedges; a label alphabet includes \emph{fusion labels}, and each fusion label $A$ is accompanied with the complementary one $\overline{A}$. Hyperedges labeled with complementary labels can be fused, which means that their corresponding attachment vertices are identified and the hyperedges are removed. See an example below (fused hyperedges are coloured grey).
$$
\vcenter{\hbox{{\tikz[baseline=.1ex]{
				\node[vertex] (V1) at (0,0) {};
				\node[vertex, colorC] (V2) at (0.8,0) {};
				\node[vertex, colorC] (V3) at (1.6,0) {};
				\draw[-latex, thick] (V1) -- node[above] {$b$} (V2);
				\draw[-latex, thick, colorC] (V2) -- node[above] {$A$} (V3);
				\node[vertex, colorC] (2V1) at (2.4,0) {};
				\node[vertex, colorC] (2V2) at (3.2,0) {};
				\node[vertex] (2V3) at (4,0) {};
				\draw[-latex, thick] (V3) -- node[above] {$c$} (2V1);
				\draw[latex-, thick, colorC] (2V1) -- node[above] {$\overline{A}$} (2V2);
				\draw[latex-, thick] (2V2) -- node[above] {$d$} (2V3);
				\node at (4.75,0) {$\Rightarrow$};
				\node[vertex] (3V1) at (5.5,0) {};
				\node[vertex] (3V2) at (5.5+0.8,0) {};
				\node[vertex] (3V3) at (5.5+1.6,0) {};
				\node[vertex] (3V4) at (5.5+2.2,0) {};
				\draw[-latex, thick] (3V1) -- node[above] {$b$} (3V2);
				\draw[latex-, thick] (3V2) -- node[above] {$d$} (3V3);
				\draw[latex-, thick] (3V4) to[out=30,in=-30,looseness=20] node[right] {$c$} (3V4);
}}}}
$$
In a fusion grammar, one can make arbitrarily many copies of hypergraphs $Z_i$ and apply fusions to them; then, a filtering procedure is applied which involves marker labels. Another feature of fusion grammars is connectors, which are added to be able to generate disconnected hypergraphs (see \cite{Lye21,Lye21_thesis} for details).

In this paper, we study decision problems for fusion grammars and the class of languages generated by them. We start with investigating decidability of the membership problem and of the non-emptiness problem for them. The current state of these problems is as follows. It is proved in \cite{Lye21_thesis} that the membership problem is decidable for \emph{substantial} and \emph{connection-preserving} fusion grammars. In \cite{Lye21}, it is proved that the non-emptiness problem is decidable for \emph{connection-preserving} fusion grammars.\footnote{The Internet has a record of the paper under the title ``Decidability and Complexity of the Membership and Emptiness Problem of Fusion Grammars'' by Aaron Lye, which was a conference paper at GCM 2018. This paper announced that the non-emptiness problem for all fusion grammars is decidable and lies in NP. However, the proof contained a substantial error, namely, it did not take disconnections into account, so, apparently, the article was revoked. The paper \cite{Lye21} fixes this error and proves decidability of the non-emptiness problem only for connection-preserving grammars.} The question of whether any of these problems is decidable or not for all fusion grammars remains open according to \cite{Lye21_thesis}. In this work, we address these open questions and establish the following results.
\begin{enumerate}
	\item The membership problem is decidable for fusion grammars without markers and connectors (Sections \ref{ssec_results_mem1}, \ref{ssec_results_mem2});
	\item The non-emptiness problem is decidable for fusion grammars (Section \ref{ssec_results_ne});
	\item Both of these problems are in NEXPTIME (Section \ref{ssec_NEXPTIME});
	\item The membership problem remains decidable if a bounded number of markers and connectors can be used (Section \ref{ssec_results_bounded}). 
\end{enumerate}

Thus we give a complete answer to the open question concerning the non-emptiness problem and a partial answer to the one concerning the membership problem. Our results generalise existing ones significantly. Decidability of the membership problem in general remains open due to subtle technical difficulties that arise when one deals with markers and connectors (see the discussion in Section \ref{ssec_results_bounded}). However, allowing only a bounded number of markers and connectors is apparently enough from the point of view of applications in DNA computing (see the discussion in Section \ref{ssec_results_bounded} too). It should also be noted that fusion grammar without markers and connectors itself is a very natural formalism based purely on the fusion operation, and the proof of its decidability contains a number of insights into how the fusion operation can be analysed.

After algorithmic complexity questions, we touch on expressive power of fusion grammars. Namely, we prove a generalisation of Parikh's theorem for connection-preserving fusion grammars. This contributes to the question from \cite{Lye21} about generative power of connection-preserving fusion grammars. 

The paper is organised as follows. 
\begin{itemize}
	\item In Section \ref{sec_preliminaries}, we define hypergraphs and fusion grammars, and we discuss basic properties of the latter. In particular, we define the new notion of an \emph{evidence path}, which shall be used in the proofs.
	\item In Section \ref{sec_results}, we present and prove the main theorems concerning algorithmic complexity (Theorem \ref{th_main_membership}, Theorem \ref{th_main_non-emptiness}, Theorem \ref{th_main_NEXPTIME}); furthermore, we define the notion of \emph{bounded fusion grammars} and prove decidability of the membership problem for them (Theorem \ref{th_main_membership_bounded}). 
	\item In Section \ref{sec_Parikh}, we prove Parikh's theorem for connection-preserving fusion grammars (Theorem \ref{th_main_parikh}).
	\item In Section \ref{sec_discussion}, we conclude by discussing the fusion grammar formalism and mentioning open research directions.
\end{itemize}

\section{Preliminaries}\label{sec_preliminaries}

Our definitions and notation is a compilation of those from \cite{DrewesKH97,KreowskiKL17,Lye21_thesis}. Some definitions are presented in a slightly different manner for the sake of convenience. We decide to present all the definitions and all the notation meticulously, organised as an enumerated list, in order to avoid misinterpretations.

\subsection{Basic Notation}
\begin{enumerate}
	\item By $w(i)$ we denote the $i$-th symbol of $w \in \Sigma^\ast$. If $v \in \Nat^k$ is a vector, then, similarly, $v(i)$ denotes its $i$-th component.
	\item The number of occurrences of a symbol $a$ in a word $w$ is denoted by $\#_a(w)$. 
	\item If $R$ is a relation, then $R^\ast$ is its transitive reflexive closure.
	\item $S_1 \uplus S_2$ denotes the union of the sets $S_1$ and $S_2$ when $S_1$ and $S_2$ are disjoint, i.e.~$S_1 \uplus S_2$ equals $S_1 \cup S_2$ if $S_1 \cap S_2 = \emptyset$ and is undefined otherwise. In contrast, $S_1 \sqcup S_2 = \{(a,i) \mid i\in \{1,2\}, a \in S_i\}$ denotes the disjoint union of $S_1$ and $S_2$. When considering the disjoint union $S_1 \sqcup S_2$, we usually identify $\{(a,i) \mid a \in S_i\}$ with $S_i$ (for $i=1,2$).
	\item Each function $f:\Sigma \to \Delta$ is extended to the function $f:\Sigma^\ast \to \Delta^\ast$ as follows: $f(a_1\dotsc a_n) = f(a_1) \dotsc f(a_n)$. 
	\item The set $\{1,2,\dotsc,n\}$ is denoted by $[n]$; in particular, $[0] \eqdef \emptyset$.
	\item A set $A \subseteq \Nat^k$ is \emph{semilinear} if it is a finite union of linear sets, i.e.~sets of the form $\{v+k_1v_1+\ldots+k_mv_m \mid k_1,\ldots,k_m \in \Nat \}$ where $m \ge 0$ and $v,v_1,\ldots,v_m \in \Nat^k$ are some fixed vectors.
	\item A \emph{typed set} $M$ is the set $M$ along with a function $\type:M \to \Nat$. 
\end{enumerate}

\subsection{Hypergraphs}
\begin{definition}
	Given a typed set $\Sigma$ called the \emph{label alphabet}, a \emph{hypergraph} $H$ over $\Sigma$ is a tuple $( V_H, E_H, att_H, lab_H )$ where $V_H$ is a finite set of \emph{vertices}, $E_H$ is a finite set of \emph{hyperedges}, $att_H: E_H\to V_H^\ast$ is an \emph{attachment} (a function that assigns a word consisting of vertices to each hyperedge; these vertices are called \emph{attachment vertices}), and $lab_H: E_H \to \Sigma$ is a \emph{labeling} such that $\type(lab_H(e))=\vert att_H(e)\vert$ for any $e\in E_H$.
\end{definition}

\begin{enumerate}
	\item The set of all hypergraphs with labels from $\Sigma$ is denoted by $\mathcal{H}(\Sigma)$. 
	
	\item In drawings of hypergraphs, small circles represent vertices, labeled rectangles represent hyperedges, and $att_G$ is represented by numbered lines. If a hyperedge has exactly two attachment vertices, then it is depicted by a labeled arrow that goes from the first attachment vertex to the second one. See Example \ref{example_FG} for illustration.
	
	\item The function $\type_H: E_H \to \Nat$ (or usually just $\type$ if $H$ is clear from the context) computes the number of attachment vertices of a hyperedge: $\type_H(e) = \vert att_H(e) \vert$. The condition from the definition of a hypergraph can then be rewritten as follows: $\type(lab_H(e)) = \type_H(e)$.
	
	\item The \emph{size} $\vert H \vert$ of a hypergraph $H$ is the total number of vertices and hyperedges in it.
	
	\item The function $\#_a(H)$ returns the number of $a$-labeled hyperedges in $H$. If $A$ is a finite set of labels, then $\#_A(H) = \sum_{a \in A} \#_a(H)$.
	
	\item If $H$ is a hypergraph and $f:E \to \Delta$ is a function such that $E \supseteq E_H$ (where $\Delta$ is a label alphabet), then $f(H) = (V_H,E_H,att_H,f \restriction_{E_H})$ is the \emph{relabeling of $H$ using $f$}.
	
	\item A hypergraph $H$ is a \emph{subhypergraph} of a hypergraph $H^\prime$ (denoted by $H \subseteq H^\prime$) if $V_H \subseteq V_{H^\prime}$, $E_H \subseteq E_{H^\prime}$ and $att_H = att_{H^\prime} \restriction_{E_H}$, $lab_H = lab_{H^\prime} \restriction_{E_H}$.
	
	\item An \emph{isomorphism} of hypergraphs $H_1$ and $H_2$ is a pair of functions $(\varphi_V,\varphi_E)$ where $\varphi_V: V_{H_1} \to V_{H_2}$, $\varphi_E: E_{H_1} \to E_{H_2}$ are bijections such that $\varphi_V(att_{H_1}(e)) = att_{H_2}(\varphi_E(e))$, $lab_{H_1}(e) = lab_{H_2}(\varphi_E(e))$ for all $e \in E_{H_1}$. If there is an isomorphism between hypergraphs $H_1$ and $H_2$, then we denote this as $H_1 \cong H_2$ and say that they are isomorphic.
	
	In the theory of hypergraph grammars, isomorphic hypergraphs are usually not distinguished \cite{DrewesKH97} because one is interested only in the structure of a hypergraph. However, in reasonings about hypergraphs, it is convenient to refer to vertices and hyperedges of a hypergraph without considering equivalence classes. The distinction between equivalence classes of the isomorphism relation and their representatives can be done rigorously by introducing the notions of \emph{abstract hypergraph} and \emph{concrete hypergraph} (\emph{concrete hypergraph} is synonymous to the term \emph{hypergraph} as defined above, while \emph{abstract hypergraph} means an equivalence class of the isomorphism relation) and by consistently using them when speaking about hypergraphs. In this work, following tradition and aiming to avoid bureaucracy, we usually do not do this explicitly but assume that this is always clear from the context.
	
	\item The \emph{empty hypergraph} is the hypergraph with no vertices and no hyperedges. 
	
	\item Let $H \in \mathcal{H}(\Sigma)$. A sequence $(i_1,e_1,o_1),\dotsc,(i_n,e_n,o_n) \in (\Nat \times E_H \times \Nat)^\ast$ is a \emph{path between $v \in V_H$ and $v^\prime \in V_{H}$} if $att_H(e_1)(i_1)=v$, $att_H(e_n)(o_n) = v^\prime$, and $att_H(e_k)(o_k) = att_H(e_{k+1})(i_{k+1})$ for $k=1,\dotsc,n-1$. 
	
	\item A hypergraph is \emph{connected} if there is a path between any two vertices. 
	
	\item A \emph{connected component of a hypergraph $H$} is a non-empty connected hypergraph $C$ such that $C \subseteq H$ and such that $C \subseteq C^\prime \subseteq H$ for some connected hypergraph $C^\prime$ implies that $C^\prime = C$. The set of connected components of $H$ is denoted by $\mathcal{C}(H)$.
	
	\item The \emph{disjoint union} $H_1+H_2$ of hypergraphs $H_1$, $H_2$ is the hypergraph $( V_{H_1}\sqcup V_{H_2}, E_{H_1}\sqcup E_{H_2}, att, lab )$ such that $att\restriction_{H_i} = att_{H_i}$, $lab\restriction_{H_i} = lab_{H_i}$ ($i=1,2$). That is, we just put these hypergraphs together without fusing any vertices or hyperedges. Note that $+$ is associative and commutative (up to isomorphism). The multiple disjoint union $H_1 + \dotsc + H_k$ is shortly denoted as $\sum_{i=1}^k H_i$, and $n \cdot H$ is a shorthand for $\underbrace{H+\ldots+H}_{n~\text{times}}$.
	
	\item Let $m: \mathcal{C}(H) \to \Nat$ be a mapping called \emph{multiplicity}. Then $m \cdot H = \sum\limits_{C \in \mathcal{C}(H)} m(C) \cdot C$ is the \emph{multiplication} of $H$.
	
	\item Let $H$ be a hypergraph and let $E \subseteq E_{H}$. Then $H - E$ is the hypergraph \\
	$(V_{H},E_{H} \setminus E, att_H\restriction_{E_{H} \setminus E},lab_H \restriction_{E_{H} \setminus E})$; i.e.~we remove all hyperedges that belong to $E$ from $E_H$.
	
	\item Let $R \subseteq \Sigma$ be the set of labels we want to remove. Then, given $H \in \mathcal{H}(\Sigma)$, the hypergraph $\rem_R(H)$ is obtained from $H$ by removing all hyperedges with labels from $R$. Formally, if $E = \{e \in E_H \mid lab_H(e) \in R\}$, then $\rem_R(H) = H - E$.
	
	\item Let $H$ be a hypergraph and let $\equiv$ be an equivalence relation on $V_H$. Then $H/{\equiv}$ is the hypergraph with the set $\{[v]_\equiv \mid v \in V_H\}$ of vertices, the set $E_H$ of hyperedges, the same labeling $lab_H$ and with the new attachment defined as follows: $att_{H/\equiv}(e) = [att_{H}(e)(1)]_\equiv[att_{H}(e)(2)]_\equiv\ldots [att_{H}(e)(\type_H(e))]_\equiv$.
\end{enumerate}

\begin{remark}\label{remark_isomorphism_isolated}
	If hypergraphs $H_1,H_2$ do not contain isolated vertices, then they are isomorphic if and only if there exists a bijection $\varphi: E_{H_1} \to E_{H_2}$ such that
	\begin{enumerate}
		\item $att_{H_1}(e_1)(i_1) = att_{H_1}(e_2)(i_2)$ if and only if $att_{H_2}(\varphi(e_1))(i_1) = att_{H_2}(\varphi(e_2))(i_2)$;
		\item $lab_{H_1}(e) = lab_{H_2}(\varphi(e))$.
	\end{enumerate}
	Indeed, given such a function $\varphi_E = \varphi$, one can define $\varphi_V: V_{H_1} \to V_{H_2}$ as follows: 
	$$
	\varphi_V(v) = att_{H_2}(\varphi(e))(i) \mbox{ whenever } v = att_{H_1}(e)(i).
	$$
	This definition is sound i.e.~it does not depend on $e$ and $i$ such that $v = att_{H_1}(e)(i)$; $\varphi_V(v)$ is defined for each $v \in V_{H_1}$ (since $H_1$ does not have isolated nodes); and $\varphi_V$ is an injection. Moreover, it is surjective because there are no isolated vertices in $H_2$, and each hyperedge in $E_{H_2}$ is an image of some hyperedge in $E_{H_1}$. Thus $(\varphi_V,\varphi_E)$ satisfies the definition of the isomorphism. In this case, we say that $\varphi$ \emph{induces an isomorphism} of $H_1$ and $H_2$.
\end{remark}

\subsection{Fusion Grammars}\label{ssec_fusion_grammars}
\begin{definition}
	Let $\Sigma$ be a finite label alphabet. Let $\Sigma = T \uplus F \uplus \overline{F} \uplus \mathcal{M} \uplus \mathcal{K}$ where
	\begin{enumerate}
		\item $T$ is an alphabet of \emph{terminal labels};
		\item $F$ is a \emph{fusion alphabet};
		\item each label $A \in F$ is accompanied by a \emph{complementary label} $\overline{A} \in \overline{F}$; it holds that $\type\left(A\right) = \type\left(\overline{A}\right)$ for all $A \in F$. The set $\overline{F}$ equals $\left\{\overline{A} \mid A \in F \right\}$;
		
		(If $B = \overline{A} \in \overline{F}$, then let $\overline{B} = A$; thus the operation $\overline{\,\cdot\,}$ becomes an involution.)
		\item $\mathcal{M}$ is the set of \emph{markers};
		\item $\mathcal{K}$ is the set of \emph{connectors}.
	\end{enumerate}
	A \emph{fusion grammar} is a tuple $FG = (Z,F,T,\mathcal{M},\mathcal{K})$ where $Z \in \mathcal{H}(\Sigma)$ is a \emph{start hypergraph}.
	\\
	A \emph{fusion grammar without markers and connectors} is a triple $FG = (Z,F,T)$ where $Z \in \mathcal{H}\left(T \cup F \cup \overline{F}\right)$ is a start hypergraph.
\end{definition}
\begin{enumerate}
	\item Let $H$ be a hypergraph with two hyperedges $e,\overline{e} \in E_H$ such that $lab_H(e) = A \in F$ and $lab_H\left(\overline{e}\right) = \overline{A}$. Let $X = H - \{e,\overline{e}\}$ be obtained from $H$ by removing these two hyperedges. Let $H^\prime = X/{\equiv}$ where $\equiv$ is the smallest equivalence relation on $V_H$ such that $att_H(e)(i) \equiv att_H(\overline{e})(i)$ for all $i =1, \dotsc, \type(A)$. Then we say that $H^\prime$ is obtained from $H$ by a \emph{fusion rule} and denote this as $H \underset{\fr}{\Rightarrow} H^\prime$.
	
	\item\label{item_parallelised_fr} Let $H$ be a hypergraph; let $P$ be a set such that each its element is of the form $\{e,\overline{e}\}$ with $e,\overline{e} \in E_H$ and $lab_H(e) = A$, $lab_H(\overline{e}) = \overline{A}$ for some $A \in F$. Let us require that $p_1,p_2 \in P$ implies that either $p_1=p_2$ or $p_1\cap p_2 = \emptyset$. Informally, $P$ is the set of hyperedge pairs we want to fuse. We define the equivalence relation $\equiv_P$ as the smallest one such that $att_H(e)(i) \equiv_P att_H(\overline{e})(i)$ for all $\{e,\overline{e}\} \in P$ and for all $i = 1,\dotsc,\type(e)$.
	
	Let $X = H- \bigcup\limits_{p \in P} p$ be obtained from $H$ by removing all the hyperedges that are contained in pairs from $P$. Let $H^\prime = X/{\equiv_P}$. Then we say that $H^\prime$ is obtained from $H$ by an application of the \emph{parallelised fusion rule} $\fr(P)$ and denote this by $H \underset{\fr(P)}{\Rightarrow} H^\prime$. 
	\\
	By the parallelisation and the sequentialisation properties \cite{KreowskiKL17,Lye21_thesis}, a parallelised fusion rule can be decomposed into a sequence of fusion rules and vice versa.
	
	\item A \emph{derivation step} $H \Rightarrow H^\prime$ is either a fusion rule application $H \underset{\fr}{\Rightarrow} H^\prime$ for some $A \in F$ or a multiplication $H \Rightarrow m \cdot H$ for some multiplicity $m$.
	
	\item The language $L(FG)$ generated by a fusion grammar $FG = (Z,F,T,\mathcal{M},\mathcal{K})$ contains a hypergraph $X$ if and only if it equals (up to isomorphism) a hypergraph of the form $\rem_{\mathcal{M} \cup \mathcal{K}}(Y)$ where
	\begin{itemize}
		\item $Y \in \mathcal{C}(H) \cap \mathcal{H}(T \cup \mathcal{M} \cup \mathcal{K})$ for some $H$ such that $Z \Rightarrow^\ast H$;
		\item $\#_{\mathcal{M}}(Y) \ge 1$.
	\end{itemize}
	That is, if one obtains $H$ from $Z$ by means of multiplications and fusion rules and if there is a connected component $Y$ of $H$ that contains a marker-labeled hyperedge, then the result of removing markers and connectors from $Y$ (up to isomorphism) belongs to $L(FG)$.
	
	\item The language $L(FG)$ generated by a fusion grammar without markers and connectors $FG = (Z,F,T)$ equals
	$
	L(FG) = \left\{X \mid \exists Y . \exists H . \left(Z \Rightarrow^\ast H \mbox{ and } Y \in \mathcal{C}(H) \cap \mathcal{H}(T) \mbox{, and } X \cong Y\right) \right\}.
	$
	
	In other words, if a hypergraph $H$ is derivable from $Z$ in $FG$ and if $Y$ is one of its connected components with terminal labels only, then any hypergraph $X$ isomorphic to $Y$ belongs to $L(FG)$.
	
	\item The fusion grammar $FG = (Z,F,T,\mathcal{M},\mathcal{K})$ (or the fusion grammar without markers and connectors $FG = (Z,F,T)$) is \emph{connection-preserving} if the following holds: \\
	\emph{Let $Z \Rightarrow^\ast H + C$ where either $C$ is connected and $e_1,e_2 \in E_C$ or $C = C_1+C_2$ where $C_1,C_2$ are connected and $e_i \in E_{C_i}$ for $i=1,2$. Let $C \underset{\fr(P)}{\Rightarrow} D$ be a fusion rule application with $P$ consisting of one element $\{e_1,e_2\}$. Then, $D$ is connected.}
	
	Informally, this definition just says that no disconnection can happen in a derivation. If $C$ is a connected hypergraph obtained during a derivation, then fusing any pair of hyperedges in it results in a connected hypergraph; likewise, if the hypergraph $C$ obtained in a derivation has two connected components and we fuse hyperedges that belong to these two components, then the resulting hypergraph is connected as well.
\end{enumerate}

\begin{example}\label{example_FG}
	Consider the hypergraph $\Zex = \sum_{i=0}^6 \Zex_i$ where $\Zex_i$ are depicted below.
	$$
	\Zex_0 = \;
	\vcenter{\hbox{{\tikz[baseline=.1ex]{
					\node[vertex] (VO) at (0,0) {};
					\node[vertex] (VL) at (-0.6,-0.6) {};
					\node[vertex] (VR) at (0.6,-0.6) {};
					\draw[-latex, thick] (VL) -- node[below] {$a$} (VR);
					\draw[-latex, thick] (VL) -- node[above left] {$A$} (VO);
					\draw[-latex, thick] (VR) -- node[above right] {$C$} (VO);
	}}}}
	\qquad\qquad
	\Zex_1 = \;
	\vcenter{\hbox{{\tikz[baseline=.1ex]{
					\node[vertex] (V1) at (0,0) {};
					\node[vertex] (V2) at (0,0.75) {};
					\node[vertex] (V3) at (0,1.5) {};
					\draw[-latex, thick] (V1) -- node[right] {$\overline{A}$} (V2);
					\draw[-latex, thick] (V2) -- node[right] {$B$} (V3);
	}}}}
	\qquad\qquad
	\Zex_2 = \;
	\vcenter{\hbox{{\tikz[baseline=.1ex]{
					\node[vertex] (V1) at (0,0) {};
					\node[vertex] (V2) at (0,0.75) {};
					\node[vertex] (V3) at (0,1.5) {};
					\draw[-latex, thick] (V1) -- node[right] {$\overline{C}$} (V2);
					\draw[-latex, thick] (V2) -- node[right] {$B$} (V3);
	}}}}
	\qquad\qquad
	\Zex_3 = \;
	\vcenter{\hbox{{\tikz[baseline=.1ex]{
					\node[vertex] (V1) at (0,0) {};
					\node[vertex] (V2) at (0,0.75) {};
					\draw[-latex, thick] (V1) to[bend left=45] node[left] {$\overline{A}$} (V2);
					\draw[-latex, thick] (V2) to[bend left=45] node[right] {$B$} (V1);
	}}}}
	$$
	$$
	\Zex_4 = \;
	\vcenter{\hbox{{\tikz[baseline=.1ex]{
					\node[vertex] (V1) at (0,0) {};
					\node[vertex] (V2) at (0,0.75) {};
					\draw[latex-, thick] (V1) to[bend left=45] node[left] {$\overline{B}$} (V2);
					\draw[latex-, thick] (V2) to[bend left=45] node[right] {$\overline{C}$} (V1);
	}}}}
	\qquad\qquad
	\Zex_5 = \;
	\vcenter{\hbox{{\tikz[baseline=.1ex]{
					\node[vertex] (VO) at (0,0) {};
					\node[hyperedge] (EL) at (-0.8,0) {$c$};
					\node[vertex] (VL) at (-0.8,-0.8) {};
					\node[vertex] (VL1) at (-1.6,0) {};
					\node[vertex] (VL2) at (-1.6,-0.8) {};
					\node[vertex] (VL3) at (-2.4,0) {};
					\draw[-latex, thick] (VL1) -- node[left] {$b$} (VL2);
					\draw[-latex, thick] (VL1) -- node[above] {$b$} (VL3);
					\draw[-latex, thick] (VL) -- node[below right] {$\overline{A}$} (VO);
					\draw[-] (EL) -- node[left] {\scriptsize 1} (VL);
					\draw[-] (EL) -- node[above] {\scriptsize 2} (VO);
					\draw[-] (EL) -- node[above] {\scriptsize 3} (VL1);
	}}}}
	\qquad\qquad
	\Zex_6 = \;
	\vcenter{\hbox{{\tikz[baseline=.1ex]{
					\node[vertex] (VO) at (0,0) {};
					\node[hyperedge] (ER) at (0.8,0) {$c$};
					\node[vertex] (VR) at (0.8,-0.8) {};
					\node[vertex] (VR1) at (1.5,0) {};
					\draw[-latex, thick] (VR1) to[out=30,in=-30,looseness=30] node[right] {$a$} (VR1);
					\draw[-latex, thick] (VR) -- node[below left] {$\overline{C}$} (VO);
					\draw[-] (ER) -- node[right] {\scriptsize 1} (VR);
					\draw[-] (ER) -- node[above] {\scriptsize 2} (VO);
					\draw[-] (ER) -- node[above] {\scriptsize 3} (VR1);		
	}}}}
	$$
	Then $\FGex = (\Zex,\{A,B,C\},\{a,b,c\})$ is a fusion grammar without markers and connectors. Let us also define $\FGPex = (\Zex,\{A,B,C\},\{a\},\{b\},\{c\})$ by making $b$ a marker and $c$ a connector. These two grammars are our running examples.
	
	Note that, for each $i=0,\dotsc,6$, there are no two hyperedges in the hypergraph $\Zex_i$ with the same label; therefore, a hyperedge of $\Zex_i$ is uniquely determined by its label. Let us refer to the $x$-labeled hyperedge of $\Zex_i$ as $e^x_i$.
\end{example}

\begin{example}\label{example_FG_der_1}
	Consider a derivation in $\FGex$:
	\begin{equation}\label{eq_FG_der_1}
		\Zex
		\;\;\Rightarrow\;\;
		\vcenter{\hbox{{\tikz[baseline=.1ex]{
						\node[vertex] (VO) at (0,0) {};
						\node[vertex] (VL) at (-0.6,-0.6) {};
						\node[vertex] (VR) at (0.6,-0.6) {};
						\draw[-latex, thick] (VL) -- node[below] {$a$} (VR);
						\draw[-latex, thick] (VL) -- node[above left] {$A$} (VO);
						\draw[-latex, thick] (VR) -- node[above right] {$C$} (VO);
		}}}}
		\,+\,
		\vcenter{\hbox{{\tikz[baseline=.1ex]{
						\node[vertex] (V1) at (0,0) {};
						\node[vertex] (V2) at (0,0.8) {};
						\draw[-latex, thick] (V1) to[bend left=30] node[left] {$\overline{A}$} (V2);
						\draw[-latex, thick] (V2) to[bend left=30] node[right] {$B$} (V1);
		}}}}
		\,+\,
		\vcenter{\hbox{{\tikz[baseline=.1ex]{
						\node[vertex] (V1) at (0,0) {};
						\node[vertex] (V2) at (0,0.8) {};
						\draw[latex-, thick] (V1) to[bend left=30] node[left] {$\overline{B}$} (V2);
						\draw[latex-, thick] (V2) to[bend left=30] node[right] {$\overline{C}$} (V1);
		}}}}
		\;\;\Rightarrow\;\;
		\vcenter{\hbox{{\tikz[baseline=.1ex]{
						\node[vertex] (VO) at (0,0) {};
						\node[vertex] (VL) at (-0.6,-0.6) {};
						\node[vertex] (VR) at (0.6,-0.6) {};
						\draw[-latex, thick] (VL) -- node[below] {$a$} (VR);
						\draw[latex-, thick] (VL) -- node[above left] {$B$} (VO);
						\draw[-latex, thick] (VR) -- node[above right] {$C$} (VO);
		}}}}
		\,+\,
		\vcenter{\hbox{{\tikz[baseline=.1ex]{
						\node[vertex] (V1) at (0,0) {};
						\node[vertex] (V2) at (0,0.8) {};
						\draw[latex-, thick] (V1) to[bend left=30] node[left] {$\overline{B}$} (V2);
						\draw[latex-, thick] (V2) to[bend left=30] node[right] {$\overline{C}$} (V1);
		}}}}
		\;\;\Rightarrow\;\;
		\vcenter{\hbox{{\tikz[baseline=.1ex]{
						\node[vertex] (VO) at (0,0) {};
						\node[vertex] (VL) at (-0.6,-0.6) {};
						\node[vertex] (VR) at (0.6,-0.6) {};
						\draw[-latex, thick] (VL) -- node[below] {$a$} (VR);
						\draw[-latex, thick] (VL) -- node[above left] {$\overline{C}$} (VO);
						\draw[-latex, thick] (VR) -- node[above right] {$C$} (VO);
		}}}}
		\;\;\Rightarrow\;\;
		\vcenter{\hbox{{\tikz[baseline=.1ex]{
						\node[vertex] (VO) at (0,0) {};
						\node[vertex] (V) at (0,-0.6) {};
						\draw[-latex, thick] (V) to[out=-30,in=30,looseness=20] (V);
						\node (V) at (0.7,-0.6) {$a$};
		}}}}
	\end{equation}
	\noindent
	The resulting hypergraph contains the connected component $\Hex_1 = \vcenter{\hbox{{\tikz[baseline=.1ex]{
					\node[vertex] (V) at (0,0) {};
					\draw[-latex, thick] (V) to[out=-30,in=30,looseness=20] (V);
					\node (V) at (0.7,0) {$a$};
	}}}}
	$, hence $\Hex_1 \in L(\FGex)$. It also contains an isolated vertex, hence the single-vertex hypergraph also belongs to $L(\FGex)$.
\end{example}

\begin{example}\label{example_FG_der_2}
	Consider another derivation in $\FGex$:
	\begin{equation}\label{eq_FG_der_2}
		\Zex 
		\quad\Rightarrow\quad
		\vcenter{\hbox{{\tikz[baseline=.1ex]{
						\node[vertex] (VO) at (0,0) {};
						\node[vertex] (VL) at (-0.6,-0.6) {};
						\node[vertex] (VR) at (0.6,-0.6) {};
						\draw[-latex, thick] (VL) -- node[below] {$a$} (VR);
						\draw[-latex, thick] (VL) -- node[above left] {$A$} (VO);
						\draw[-latex, thick] (VR) -- node[above right] {$C$} (VO);
		}}}}
		\;\;+\;\;
		\vcenter{\hbox{{\tikz[baseline=.1ex]{
						\node[vertex] (V1) at (0,0) {};
						\node[vertex] (V2) at (0,0.8) {};
						\node[vertex] (V3) at (0,1.6) {};
						\draw[-latex, thick] (V1) -- node[right] {$\overline{A}$} (V2);
						\draw[-latex, thick] (V2) -- node[right] {$B$} (V3);
		}}}}
		\;\;+\;\;
		\vcenter{\hbox{{\tikz[baseline=.1ex]{
						\node[vertex] (V1) at (0,0) {};
						\node[vertex] (V2) at (0,0.8) {};
						\node[vertex] (V3) at (0,1.6) {};
						\draw[-latex, thick] (V1) -- node[right] {$\overline{C}$} (V2);
						\draw[-latex, thick] (V2) -- node[right] {$B$} (V3);
		}}}}
		\quad\Rightarrow\quad
		\vcenter{\hbox{{\tikz[baseline=.1ex]{
						\node[vertex] (VO) at (0,0) {};
						\node[vertex] (VL) at (-0.6,-0.6) {};
						\node[vertex] (VR) at (0.6,-0.6) {};
						\node[vertex] (V1) at (0,0.8) {};
						\draw[-latex, thick] (VL) -- node[below] {$a$} (VR);
						\draw[-latex, thick] (VR) -- node[above right] {$C$} (VO);
						\draw[-latex, thick] (VO) -- node[left] {$B$} (V1);
		}}}}
		\;\;+\;\;
		\vcenter{\hbox{{\tikz[baseline=.1ex]{
						\node[vertex] (V1) at (0,0) {};
						\node[vertex] (V2) at (0,0.8) {};
						\node[vertex] (V3) at (0,1.6) {};
						\draw[-latex, thick] (V1) -- node[right] {$\overline{C}$} (V2);
						\draw[-latex, thick] (V2) -- node[right] {$B$} (V3);
		}}}}
		\quad\Rightarrow\quad
		\vcenter{\hbox{{\tikz[baseline=.1ex]{
						\node[vertex] (VO) at (0,0) {};
						\node[vertex] (VL) at (-0.6,-0.4) {};
						\node[vertex] (VR) at (0.6,-0.4) {};
						\node[vertex] (V1) at (-0.5,0.8) {};
						\node[vertex] (V2) at (0.5,0.8) {};
						\draw[-latex, thick] (VL) -- node[below] {$a$} (VR);
						\draw[-latex, thick] (VO) -- node[below left] {$B$} (V1);
						\draw[-latex, thick] (VO) -- node[below right] {$B$} (V2);
		}}}}
	\end{equation}
	\noindent
	The result of this derivation contains $\Hex_2 = \vcenter{\hbox{{\tikz[baseline=.1ex]{
					\node[vertex] (VL) at (-0.4,0) {};
					\node[vertex] (VR) at (0.4,0) {};
					\draw[-latex, thick] (VL) -- node[above] {$a$} (VR);
	}}}}$, hence $\Hex_2$ belongs to $L(\FGex)$. 
\end{example}


\begin{example}
	Note that neither $\Hex_1$ from Example \ref{example_FG_der_1} nor $\Hex_2$ from Example \ref{example_FG_der_2} belong to $L\left(\FGPex\right)$, because there is no $b$-labeled hyperedge attached to any of them.
	Now, consider the following derivation:
	$$
	\Zex 
	\quad\Rightarrow\quad
	\vcenter
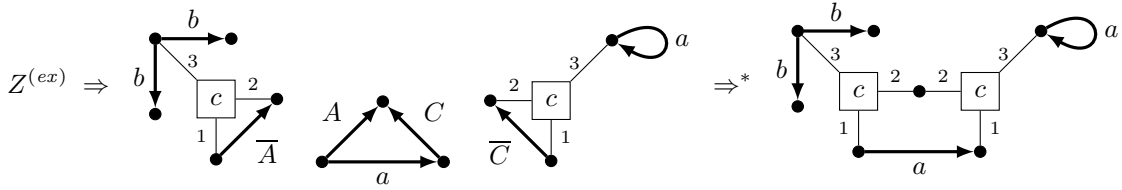
{\hbox{{\tikz[baseline=.1ex]{
					\node at (0.5,-0.4) {$+$};
					\node at (0.5+2.6,-0.4) {$+$};
					\node[vertex] (VO) at (0,0) {};
					\node[hyperedge] (EL) at (-0.8,0) {$c$};
					\node[vertex] (VL) at (-0.8,-0.8) {};
					\node[vertex] (VL1) at (-1.6,0) {};
					\node[vertex] (VL2) at (-1.6,-0.8) {};
					\node[vertex] (VL3) at (-2.4,0) {};
					\draw[-latex, thick] (VL1) -- node[left] {$b$} (VL2);
					\draw[-latex, thick] (VL1) -- node[above] {$b$} (VL3);
					\draw[-latex, thick] (VL) -- node[below right] {$\overline{A}$} (VO);
					\draw[-] (EL) -- node[left] {\scriptsize 1} (VL);
					\draw[-] (EL) -- node[above] {\scriptsize 2} (VO);
					\draw[-] (EL) -- node[above] {\scriptsize 3} (VL1);
					\node[vertex] (2VO) at (0+1.8,0) {};
					\node[vertex] (2VL) at (-0.8+1.8,-0.8) {};
					\node[vertex] (2VR) at (0.8+1.8,-0.8) {};
					\draw[-latex, thick] (2VL) -- node[below] {$a$} (2VR);
					\draw[-latex, thick] (2VL) -- node[above left] {$A$} (2VO);
					\draw[-latex, thick] (2VR) -- node[above right] {$C$} (2VO);
					\node[vertex] (3VO) at (0+3.6,0) {};
					\node[hyperedge] (3ER) at (0.8+3.6,0) {$c$};
					\node[vertex] (3VR) at (0.8+3.6,-0.8) {};
					\node[vertex] (3VR1) at (1.6+3.6,0) {};
					\draw[-latex, thick] (3VR1) to[out=-60,in=-120,looseness=20] node[below] {$a$} (3VR1);
					\draw[-latex, thick] (3VR) -- node[below left] {$\overline{C}$} (3VO);
					\draw[-] (3ER) -- node[right] {\scriptsize 1} (3VR);
					\draw[-] (3ER) -- node[above] {\scriptsize 2} (3VO);
					\draw[-] (3ER) -- node[above] {\scriptsize 3} (3VR1);		
	}}}}
	\quad\underset{\fr}{\Rightarrow}^\ast\quad
	\vcenter{\hbox{{\tikz[baseline=.1ex]{
					\node[vertex] (VO) at (0,0) {};
					\node[hyperedge] (EL) at (-0.8,0) {$c$};
					\node[hyperedge] (ER) at (0.8,0) {$c$};
					\node[vertex] (VL) at (-0.8,-0.8) {};
					\node[vertex] (VR) at (0.8,-0.8) {};
					\node[vertex] (VR1) at (1.6,0) {};
					\node[vertex] (VL1) at (-1.6,0) {};
					\node[vertex] (VL2) at (-1.6,-0.8) {};
					\node[vertex] (VL3) at (-2.4,0) {};
					\draw[-latex, thick] (VL) -- node[below] {$a$} (VR);
					\draw[-latex, thick] (VL1) -- node[left] {$b$} (VL2);
					\draw[-latex, thick] (VL1) -- node[above] {$b$} (VL3);
					\draw[-latex, thick] (VR1) to[out=-60,in=-120,looseness=20] node[below] {$a$} (VR1);
					\draw[-] (EL) -- node[left] {\scriptsize 1} (VL);
					\draw[-] (EL) -- node[above] {\scriptsize 2} (VO);
					\draw[-] (EL) -- node[above] {\scriptsize 3} (VL1);
					\draw[-] (ER) -- node[right] {\scriptsize 1} (VR);
					\draw[-] (ER) -- node[above] {\scriptsize 2} (VO);
					\draw[-] (ER) -- node[above] {\scriptsize 3} (VR1);		
	}}}}
	$$
	The resulting hypergraph (denote it by $\Yex$) contains the marker label $b$; each of its labels is either terminal, or it is a marker or a connector. Therefore, $\rem_{\{b,c\}}\left(\Yex\right)$ belongs to $L\left(\FGPex\right)$. The hypergraph $\Hex_3 = \rem_{\{b,c\}}\left(\Yex\right)$ is depicted below.
	$$
	\Hex_3 = \rem_{\{b,c\}}(\Yex) = 
	\vcenter{\hbox{{\tikz[baseline=.1ex]{
					\node[vertex] (V1) at (-0.5,-0.5) {};
					\node[vertex] (V2) at (0.5,-0.5) {};
					\node[vertex] (TR) at (0,0) {};
					\node[vertex] (TA) at (-0.5,0) {};
					\node[vertex] (TL) at (-1,0) {};
					\node[vertex] (TB) at (-0.5,-0.5) {};
					\node[vertex] (L) at (1,0) {};
					\draw[-latex, thick] (V1) -- node[below] {$a$} (V2);
					\draw[-latex, thick] (L) to[out=-60,in=-120,looseness=20] node[below] {$a$} (L);
	}}}}
	$$
\end{example}

\begin{remark}\label{rem_parallelization}
	As mentioned in \cite[Section 5]{KreowskiKL17}, fusion rules are parallel and sequentially independent. Consequently, an application of a parallelised fusion rule $H \underset{\fr(P)}{\Rightarrow} H^\prime$ can be transformed into a sequence of fusion rule applications, and vice versa. Hence, for $P = \{p_1,\dotsc,p_k\}$, the following derivation is equivalent to the application of $\fr(P)$:
	$
	H \underset{\fr(\{p_1\})}{\Rightarrow} H_1 \underset{\fr(\{p_2\})}{\Rightarrow}  \dotsc \underset{\fr(\{p_k\})}{\Rightarrow}  H_k = H^\prime.
	$
\end{remark}

\begin{example}\label{ex_parallelised}
	We can replace the sequence of fusion rule applications in (\ref{eq_FG_der_1}) by a single application of $\fr(\Pex)$ where $\Pex = \{\{e_0^A,e_3^{\overline{A}}\}, \{e_3^B,e_4^{\overline{B}}\}, \{e_4^{\overline{C}},e_0^C\}\}$. Recall that the names of hyperedges are introduced at the end of Example \ref{example_FG}.
\end{example}

\begin{remark}
	A fusion grammar without markers and connectors $FG = (Z,F,T)$ can be viewed as a fusion grammar $FG^\prime = (Z^\prime,F,T,\{\mu\},\emptyset)$ where $Z^\prime$ is obtained from $Z$ by attaching a $\mu$-labeled hyperedge to each vertex in $Z$ (with $\type(\mu)=1$). Clearly, this makes the marker label $\mu$ redundant in $Z^\prime$, hence $L(FG) = L(FG^\prime)$.
\end{remark}

Below we formulate a proposition establishing a normal form of derivations in a fusion grammar. This statement is \cite[Corollary 1]{KreowskiKL17}.
\begin{proposition}[normal form]\label{prop_nf}
	If $Z \Rightarrow^\ast H$, then there exists a \emph{most parallelised derivation} of the form $Z \Rightarrow m \cdot Z \underset{\fr(P)}{\Rightarrow} H + G$ for some multiplicity $m$, some parallelised fusion rule $\fr(P)$ and some hypergraph $G$.
\end{proposition}
\begin{proof}[Proof sketch]
	If, in a derivation, multiplication is done after fusion, then these two steps can be interchanged. To prove this we should consider several cases as in \cite[Proposition 2]{KreowskiKL17}. For example, assume that there are the following two steps: $C \underset{\fr}{\Rightarrow} C_1 + C_2 \Rightarrow k_1 \cdot C_1 + k_2 \cdot C_2$. Assume that $C,C_1,C_2$ are connected; in other words, the fusion rule application results in disconnection of $C$. We want to start with the multiplication of $C$ and proceed with a fusion. The problem is that the constants $k_1$ and $k_2$ are different in general (let, e.g., $k_1<k_2$). To overcome this, let $k = \max\{k_1,k_2\}$. Then the rules are interchanged in the following way: 
	$$
	C \Rightarrow k \cdot C \underset{\fr}{\Rightarrow}^k k \cdot (C_1 + C_2) = k \cdot C_1 + k \cdot C_2 = k_1 \cdot C_1 + k_2 \cdot C_2 + (k_2-k_1)\cdot C_1.
	$$
	An additional part $(k_2-k_1)\cdot C_1$ appears as compared to the original derivation. This is the reason why $G$ appears in a most parallelised derivation. Note that, for some reason, \cite[Corollary 1]{KreowskiKL17} is stated without this $G$, which is a mistake. In the proof of that corollary, the hypergraph $G$ appears implicitly.
	
	After the reordering, we have a sequence of multiplications followed by a sequence of fusions; the resulting hypergraph is $H+G$ for some $G$. Finally, multiplications can be combined into a single one; the same holds for fusions according to Remark \ref{rem_parallelization}. Thus we have obtained a derivation of the desired form.
\end{proof}

This normal form can be refined for connection-preserving fusion grammars.

\begin{proposition}[normal form for connection-preserving fusion grammars]\label{prop_nf_cpfg}
	If $FG = (Z,F,T,\mathcal{M},\mathcal{K})$ is a connection-preserving fusion grammar, then $H \in L(FG)$ if and only if $H = \rem_{\mathcal{M} \cup \mathcal{K}}(Y)$ such that $Y \in \mathcal{H}(T \cup \mathcal{M} \cup \mathcal{K})$ is a connected hypergraph containing at least one $\mu$-labeled hyperedge, and $Z \Rightarrow n \cdot Z \underset{\fr(Q)}{\Rightarrow} Y$ for some $n$ and $Q$.
\end{proposition}
Here no additional hypergraph $G$ appears. See the proof in \ref{appendix_proof_prop_nf_cpfg}.

\subsection{Evidence Path}

In the proofs, we shall analyse the equivalence relation $\equiv_P$ associated with the parallelised fusion rule $\fr(P)$. One useful notion we introduce in this section is that of an \emph{evidence path}. The idea behind it is quite simple: two vertices $v^\prime$ and $v^{\prime\prime}$ are identified by an application of $\fr(P)$ if and only if there exists a finite sequence of vertices $v^\prime = v_0,v_1,\dotsc,v_l = v^{\prime\prime}$ such that for each $i=0,\dotsc,l-1$ the neighbour vertices $v_i$ and $v_{i+1}$ are both $j$-th attachment vertices (for some $j$) of two hyperedges fused with each other (which causes $v_i$ and $v_{i+1}$ to be identified by $\fr(P)$).

\begin{proposition}\label{prop_evidence_path}
	Let $H \underset{\fr(P)}{\Rightarrow} H^\prime$ and let $v^\prime,v^{\prime\prime} \in V_{H}$ be two vertices such that $v^\prime \equiv_P v^{\prime\prime}$. Then there exists a sequence of the form
	\begin{equation}\label{equation_evidence_path}
		v_0, (e_0^+,k_0^+), (e_1^-,k_1^-), v_1, (e_1^+,k_1^+), (e_2^-,k_2^-), v_2, \dotsc, (e_{l-1}^+,k_{l-1}^+), (e_l^-,k_l^-), v_l
	\end{equation}
	where 
	\begin{itemize}
		\item $l \ge 0$, \qquad $v_i \in V_{H}$, \qquad $e_i^+,e_i^- \in E_{H}$, \qquad $v_0,\dotsc,v_l$ are distinct, \qquad $v_0 = v^\prime$, \qquad $v_l = v^{\prime\prime}$;
		\item $\{e_i^+,e_{i+1}^-\} \in P$ and $k_i^+ = k_{i+1}^-$ for $i=0,\dotsc,l-1$;
		\item $att_H(e_i^+)(k_i^+) = v_i$ for $i=0,\dotsc,l-1$; $att_H(e_i^-)(k_i^-) = v_i$ for $i=1,\dotsc,l$;
		\item if $v^\prime=v^{\prime\prime}$, then $l=0$ and the sequence (\ref{equation_evidence_path}) has the form $v^\prime$.
	\end{itemize}
\end{proposition}
\begin{proof}
	Let us say that $v^1 \sim_P v^2$ if there are hyperedges $e^1,e^2 \in E_{H}$ such that $\{e^1,e^2\} \in P$ and for some $i$ it holds that $att_H(e^j)(i) = v^j$ ($j=1,2$). By the definition, the relation $\equiv_P$ is the least equivalence relation containing $\sim_P$; thus, $\equiv_P$ is the reflexive transitive closure of $\sim_P$ (note that $\sim_P$ is already symmetric). Therefore, $v^\prime \equiv_P v^{\prime\prime}$ if and only if there exists a sequence of distinct vertices $v^\prime = v_0 \sim_P v_1 \sim_P \dotsc \sim_P v_l=v^{\prime\prime}$. This is the statement of the proposition.
\end{proof}

\begin{definition}
	The sequence (\ref{equation_evidence_path}) is called an \emph{evidence path} (or an \emph{evidence $\fr(P)$-path}). 
\end{definition}

\begin{definition}
	Given (\ref{equation_evidence_path}), the set 
	$\{(e_0^+,k_0^+), (e_1^-,k_1^-), (e_1^+,k_1^+), \dotsc, (e_l^-,k_l^-)\}$
	is called an \emph{evidence set} (or an \emph{evidence $\fr(P)$-set}). If $v^\prime = v^{\prime\prime}$, then the corresponding evidence set is the empty set.
\end{definition}

\begin{example}\label{example_der_evidence_path}
	Consider the parallelised fusion rule application from Example \ref{ex_parallelised} (see also (\ref{eq_FG_der_1}) from Example \ref{example_FG_der_1}). Let us draw the lines indicating which vertices are identified:
	$$
		\vcenter
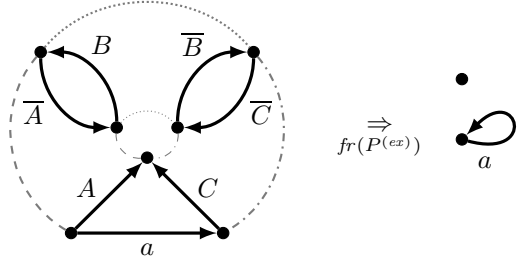
{\hbox{{\tikz[baseline=.1ex]{
				\node[vertex] (VO) at (0,0) {};
				\node[vertex] (VL) at (-0.7,-0.7) {};
				\node[vertex] (VR) at (0.7,-0.7) {};
				\draw[-latex, thick] (VL) -- node[below] {$a$} (VR);
				\draw[-latex, thick] (VL) -- (VO);
				\draw[-latex, thick] (VR) -- (VO);
				\node[] (Lb1) at (-0.6,-0.2) {$A$};
				\node[] (Lb6) at (0.6,-0.2) {$C$};
				\node[vertex] (VL1) at (-0.6,0.3) {};
				\node[vertex] (VL2) at (-1.2,0.9) {};
				\draw[latex-, thick] (VL1) to[bend left=45] (VL2);
				\draw[latex-, thick] (VL2) to[bend left=45] (VL1);
				\node[vertex] (VR1) at (0.6,0.3) {};
				\node[vertex] (VR2) at (1.2,0.9) {};
				\draw[-latex, thick] (VR1) to[bend left=45] (VR2);
				\draw[-latex, thick] (VR2) to[bend left=45] (VR1);
				\node[] (Lb2) at (-1.1,0.1) {$\overline{A}$};
				\node[] (Lb3) at (-0.45,1) {$B$};
				\node[] (Lb5) at (1.1,0.1) {$\overline{C}$};
				\node[] (Lb4) at (0.45,1) {$\overline{B}$};
				\draw[-, comment, thick, dashed] (VL) to[bend left=50] (VL2);
				\draw[-, comment, dashed] (VO) to[bend left=50] (VL1);
				\draw[-, comment, thick, densely dotted] (VL2) to[bend left=50] (VR2);
				\draw[-, comment, densely dotted] (VL1) to[bend left=50] (VR1);
				\draw[-, comment, thick, dashdotted] (VR2) to[bend left=50] (VR);
				\draw[-, comment, dashdotted] (VR1) to[bend left=50] (VO);
		}}}}
		\quad \underset{\fr(\Pex)}{\Rightarrow} \quad
		\vcenter{\hbox{{\tikz[baseline=.1ex]{
						\node[vertex] (VO) at (0,0) {};
						\node[vertex] (V) at (0,-0.7) {};
						\draw[-latex, thick] (V) to[out=-30,in=30,looseness=25] (V);
						\node (V) at (0.8,-0.7) {$a$};
		}}}}
	$$
	If two vertices are connected by a sequence of dashed or dotted arcs, then they are identified by $\fr(\Pex)$. We see that there is a path between the attachment vertices of the $a$-labeled hyperedge; this is what we call an evidence path. Formally, in this case, the evidence path is the following sequence:
	$$
	v_0, 
	(e_0^A,1),(e_3^{\overline{A}},1), v_1, 
	(e_3^B,2),(e_4^{\overline{B}},2), v_2,
	(e_4^{\overline{C}},1),(e_0^C,1), v_3.
	$$
	Here $v_0 = att_{\Zex_0}(e_0^a)(1)$ is the first attachment vertex of the $a$-labeled hyperedge; $v_1 = att_{\Zex_3}(e_3^B)(2)$; $v_2 = att_{\Zex_4}(e_4^{\overline{C}})(1)$; finally, $v_3 = att_{\Zex_0}(e_0^{a})(2)$ is the second attachment vertex of the $a$-labeled hyperedge.
\end{example}

\section{Decidability Results}\label{sec_results}

All the necessary definitions are provided, so let us proceed to defining the membership problem and the non-emptiness problem, for which the decidability results will be established.

\begin{problem}[MEM]
	\leavevmode
	\\
	\textit{Input.} 
	\begin{enumerate}
		\item A fusion grammar without markers and connectors $FG = (Z,F,T)$; 
		\item A hypergraph $H \in \mathcal{H}(T)$.
	\end{enumerate}
	\textit{Problem MEM.} Does $H$ belong to $L(FG)$?
\end{problem}

\begin{problem}[NE]
	\leavevmode
	\\
	\textit{Input.} 
	A fusion grammar $FG = (Z,F,T, \mathcal{M},\mathcal{K})$; 
	\\
	\textit{Problem NE.} Is $L(FG)$ non-empty?
\end{problem}

The question of whether any of these problems is decidable has remained open, as stated in \cite{Lye21,Lye21_thesis}. We are going to prove decidability of both problems in this paper.

\begin{theorem}\label{th_main_membership}
	The problem MEM is decidable.
\end{theorem}
\begin{theorem}\label{th_main_non-emptiness}
	The problem NE is decidable.
\end{theorem}

These results improve the following ones, which have already been proved in the literature.
\begin{enumerate}
	\item The membership problem for \emph{substantial} and connection-preserving fusion grammars is decidable and it is in NP \cite{KreowskiKL17,Lye21_thesis};
	\item The non-emptiness problem for connection-preserving fusion grammars is decidable and lies in NP \cite{Lye21}.
\end{enumerate}
The substantiality property, roughly speaking, says that, if $H = \rem_{\mathcal{M} \cup \mathcal{K}}(Y)$ is generated by a fusion grammar and the derivation is of the form $Z \Rightarrow m \cdot Z = Z_1+\dotsc+Z_k \underset{\fr(P)}{\Rightarrow} Y$, then each $Z_i$ contains some ``terminal part'' which does not disappear after the fusion rule $\fr(P)$ but remains present in $H$. Given this property, we know that the number $k$ of connected components used to derive $H$ is linearly bounded by the size of $H$. This is the main ingredient used in \cite{Lye21_thesis} to prove that the membership problem is decidable but not the only one. Indeed, the assumption that the derivation has the form $Z \Rightarrow m \cdot Z \underset{\fr(P)}{\Rightarrow} Y$ works for \emph{connection-preserving} fusion grammars, according to Proposition \ref{prop_nf_cpfg}, but not for general ones. In general, fusions may lead to disconnection, and the normal form derivation has the form $Z \Rightarrow m \cdot Z \underset{\fr(P)}{\Rightarrow} H + G$ for some additional ``garbage'' hypergraph $G$. Such a hypergraph $G$ can contain all possible kinds of labels, and it is not clear how to guess it or how to estimate its size. 

Similarly, we can prove in a straightforward way that the non-emptiness problem is decidable but only for connection-preserving fusion grammars \cite{Lye21}. It is important to understand how this proof works, so we shall present it below. Namely, let us prove that the non-emptiness problem is decidable for connection-preserving fusion grammars without markers and connectors. The proof is taken from \cite{Lye21} with certain simplifications (because we only want to prove decidability). 

In order to check whether a connection-preserving fusion grammar without markers and connectors $FG = (Z,F,T)$ (where $\mathcal{C}(Z) = \{Z_1,\dotsc,Z_k\}$) generates some hypergraph, let us transform each connected component $Z_i \in \mathcal{C}(Z)$ into the integer-valued vector $$v(Z_i) = (\#_{A_1}(Z_i)-\#_{\overline{A_1}}(Z_i),\dotsc,\#_{A_l}(Z_i)-\#_{\overline{A_l}}(Z_i))$$ where $\{A_1,\dotsc,A_l\} = F$. 
Then, $L(FG)$ is non-empty if and only if the set 
$$
S = \left\{ m_1 \cdot v(Z_1) + \dotsc + m_k \cdot v(Z_k) \mid m_i \in \Nat \mbox{ and } \exists i\; m_i > 0 \right\}
$$
contains the zero vector $\vec{0} = (0,0,\dotsc)$. Indeed, $L(FG)$ is non-empty if and only if there exists a connected hypergraph $H \in \mathcal{H}(T)$ such that $Z \Rightarrow m \cdot Z \underset{\fr}{\Rightarrow}^\ast H$ for some multiplicity $m$. If such $H$ exists, consider $v(m \cdot Z) = m(Z_1) \cdot v(Z_1) + \dotsc + m(Z_k) \cdot v(Z_k) \in S$. All the hyperedges of $m \cdot Z$ labeled by elements of $F \cup \overline{F}$ must be fused; therefore, the number of $A$-labeled hyperedges coincides with the number of $\overline{A}$-labeled ones for each $A \in F$. Thus, $v(m \cdot Z) = \vec{0}$, as desired. Note that $m(Z_i) \ne 0$ for some $i$, because otherwise $H$ is the empty hypergraph.

Conversely, assume that $\vec{0} \in S$; that is, there exist $m_1,\dotsc,m_k \in \Nat$ (one of them being greater than $0$) such that $m_1 \cdot v(Z_1) + \dotsc + m_k \cdot v(Z_k) = \vec{0}$. Then, consider the disjoint union $m_1 \cdot Z_1 + \dotsc m_k \cdot Z_k$. This hypergraph has the equal number of $A$-labeled hyperedges and of $\overline{A}$-labeled ones for each $A \in F$. Let us divide hyperedges into pairs having complementary labels and fuse them. The resulting hypergraph is from $\mathcal{H}(T)$, and it is not the empty hypergraph. Thus it has at least one connected component; the latter belongs to $L(FG)$.

This proves decidability of the non-emptiness problem for connection-preserving fusion grammars because checking whether $\vec{0} \in S$ can be reduced to checking solvability of a finite number of linear systems $Ax = b$ for $A \in \mathbb{Z}^{l\times k}$, $b \in \mathbb{Z}^l$, and $x \in \Nat^k$; this is an instance of the integer linear programming, which is in NP. In \cite{Lye21}, this proof is called the \emph{quantitative argument} since only the number of fusion labels in each connected component of $Z$ matters but not their structure (namely, how hyperedges are connected to vertices and to each other). However, for an arbitrary fusion grammar where disconnections can happen during derivations, this argument is no longer valid. More precisely, it works only one-way: if $S$ contains $\vec{0}$, then $L(FG)$ is non-empty (the same argument can be used). The counterexample to the converse statement is shown below.
\begin{example}\label{example_counter_disconnection}
	Consider the fusion grammar without markers and connectors
	$
	\left(\Zex_0+\Zex_1+\Zex_2,\{A,B,C\},\{a\}\right)
	$
	where $\Zex_i$ are defined in Example \ref{example_FG}. Let us transform $\Zex_0,\Zex_1,\Zex_2$ into vectors: $v(\Zex_0) = (1,0,1)$, $v(\Zex_1) = (-1,1,0),v(\Zex_2) = (0,1,-1)$. Clearly, $m_1\cdot(1,0,1)+m_2\cdot (-1,1,0) + m_3 \cdot (0,1,-1) = (m_1-m_2,m_2+m_3,m_1-m_3)$ equals $(0,0,0)$ iff $m_1=m_2=m_3=0$. Thus the set $S$ is empty. However, the language generated by the grammar is not empty, since, as Example \ref{example_FG_der_2} shows, $\Hex_2 = \vcenter{\hbox{{\tikz[baseline=.1ex]{
					\node[vertex] (VL) at (-0.5,0) {};
					\node[vertex] (VR) at (0.5,0) {};
					\draw[-latex, thick] (VL) -- node[above] {$a$} (VR);
	}}}}$ belongs to it.

	Note that disconnection happens at the last step of the derivation (\ref{eq_FG_der_2}). After it, there are two connected components, one of them being $\Hex_2$ and the other one being a ``garbage'' hypergraph with fusion labels.
	
	Now, suppose that we replace $\Zex_0$ by a very similar hypergraph $\vcenter{\hbox{{\tikz[baseline=.1ex]{
					\node[vertex] (VO) at (0,0) {};
					\node[vertex] (VL) at (-0.5,-0.5) {};
					\node[vertex] (VR) at (0.5,-0.5) {};
					\draw[-latex, thick] (VL) -- node[below] {$a$} (VR);
					\draw[latex-, thick] (VL) -- node[above left] {$A$} (VO);
					\draw[-latex, thick] (VR) -- node[above right] {$C$} (VO);
	}}}}$. Then the derivation (\ref{eq_FG_der_2}) changes as follows:
	\begin{equation*}
		\Zex 
		\;\;\Rightarrow\;\;
		\vcenter
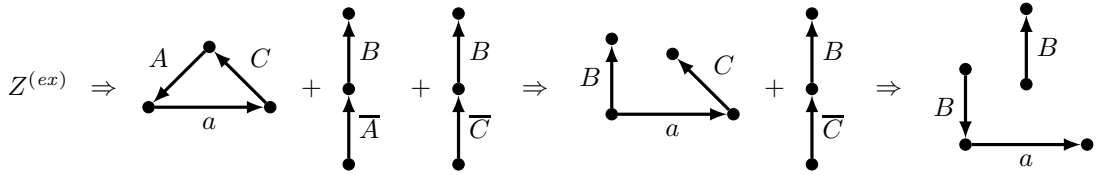
{\hbox{{\tikz[baseline=.1ex]{
						\node[vertex] (VO) at (0,0) {};
						\node[vertex] (VL) at (-0.6,-0.6) {};
						\node[vertex] (VR) at (0.6,-0.6) {};
						\draw[-latex, thick] (VL) -- node[below] {$a$} (VR);
						\draw[latex-, thick] (VL) -- node[above left] {$A$} (VO);
						\draw[-latex, thick] (VR) -- node[above right] {$C$} (VO);
		}}}}
		\;+\;
		\vcenter{\hbox{{\tikz[baseline=.1ex]{
						\node[vertex] (V1) at (0,0) {};
						\node[vertex] (V2) at (0,0.8) {};
						\node[vertex] (V3) at (0,1.6) {};
						\draw[-latex, thick] (V1) -- node[right] {$\overline{A}$} (V2);
						\draw[-latex, thick] (V2) -- node[right] {$B$} (V3);
		}}}}
		\;+\;
		\vcenter{\hbox{{\tikz[baseline=.1ex]{
						\node[vertex] (V1) at (0,0) {};
						\node[vertex] (V2) at (0,0.8) {};
						\node[vertex] (V3) at (0,1.6) {};
						\draw[-latex, thick] (V1) -- node[right] {$\overline{C}$} (V2);
						\draw[-latex, thick] (V2) -- node[right] {$B$} (V3);
		}}}}
		\;\;\Rightarrow\;\;
		\vcenter{\hbox{{\tikz[baseline=.1ex]{
						\node[vertex] (VO) at (0,0) {};
						\node[vertex] (VL) at (-0.8,-0.8) {};
						\node[vertex] (VR) at (0.8,-0.8) {};
						\node[vertex] (V1) at (-0.8,0) {};
						\draw[-latex, thick] (VL) -- node[below] {$a$} (VR);
						\draw[-latex, thick] (VR) -- node[above right] {$C$} (VO);
						\draw[-latex, thick] (VL) -- node[left] {$B$} (V1);
		}}}}
		\;+\;
		\vcenter{\hbox{{\tikz[baseline=.1ex]{
						\node[vertex] (V1) at (0,0) {};
						\node[vertex] (V2) at (0,0.8) {};
						\node[vertex] (V3) at (0,1.6) {};
						\draw[-latex, thick] (V1) -- node[right] {$\overline{C}$} (V2);
						\draw[-latex, thick] (V2) -- node[right] {$B$} (V3);
		}}}}
		\;\;\Rightarrow\;\;
		\vcenter{\hbox{{\tikz[baseline=.1ex]{
						\node[vertex] (VO) at (0,0) {};
						\node[vertex] (VL) at (-0.5,-0.5) {};
						\node[vertex] (VR) at (0.5,-0.5) {};
						\node[vertex] (V1) at (-0.5,0.3) {};
						\node[vertex] (V2) at (0,0.8) {};
						\draw[-latex, thick] (VL) -- node[below] {$a$} (VR);
						\draw[latex-, thick] (VL) -- node[left] {$B$} (V1);
						\draw[-latex, thick] (VO) -- node[right] {$B$} (V2);
		}}}}
	\end{equation*}
	The resulting hypergraph does not contain connected components without fusion labels. In fact, the grammar $\left(\Zex_0+\Zex_1+\Zex_2,\{A,B,C\},\{a\}\right)$ with new $\Zex_0$ does not generate any hypergraph (with terminal labels only). However, we only changed the direction of one hyperedge. 
\end{example}

Example \ref{example_counter_disconnection} shows us that we cannot use only the quantitative argument for solving the non-emptiness problem in general but we must also take into account the structure of the start hypergraph, i.e.~how hyperedges are connected to vertices. 

The same holds for the membership problem. Moreover, disconnections is not the only difficulty to deal with when one studies the latter problem. Assume that we have, for example, a connection-preserving fusion grammar without markers and connectors $FG$, and we want to check whether some hypergraph of interest, e.g., 
$
\vcenter{\hbox{{\tikz[baseline=.1ex]{
				\node[hyperedge] (EL) at (1.4,0) {$c$};
				\node[vertex] (VO) at (2,-0.3) {};
				\node[vertex] (VL) at (2,0.3) {};
				\node[vertex] (VL1) at (0,0) {};
				\node[vertex] (VL3) at (0.8,0) {};
				\draw[-latex, thick] (VL1) -- node[above] {$b$} (VL3);
				\draw[-latex, thick] (VL1) to[out=210,in=150,looseness=20] node[left] {$a$} (VL1);
				\draw[-] (EL) -- (VL);
				\draw[-] (EL) -- (VO);
				\draw[-] (EL) -- node[above] {\scriptsize 3} (VL3);
				\node at (1.78,0.32) {\scriptsize 1};
				\node at (1.78,-0.07) {\scriptsize 2};
}}}}
$
belongs to $L(FG)$. The quantitative argument would not solve this problem. It would only tell us whether some hypergraph with one $a$-labeled hyperedge, one $b$-labeled one and one $c$-labeled one is generated by $FG$, which is clearly not enough. So, in general, solving the membership problem requires knowing which attachment vertices of hyperedges coincide in a hypergraph obtained by fusions and which do not. This is also illustrated by Examples \ref{example_FG_der_1} and \ref{example_FG_der_2}. In order to produce the hypergraph 
$\Hex_1 = \vcenter{\hbox{{\tikz[baseline=.1ex]{
				\node[vertex] (V) at (0,0) {};
				\draw[-latex, thick] (V) to[out=-30,in=30,looseness=18] (V);
				\node (V) at (0.65,0) {$a$};
}}}}
$
we must use $\Zex_0$, $\Zex_3$ and $\Zex_4$; however, to generate
$\Hex_2 = 
{\tikz[baseline=.1ex]{
				\node[vertex] (VL) at (-0.5,0) {};
				\node[vertex] (VR) at (0.5,0) {};
				\draw[-latex, thick] (VL) -- node[above] {$a$} (VR);
}}
$
we must use $\Zex_0$, $\Zex_1$ and $\Zex_2$. The derivation (\ref{eq_FG_der_1}) of $\Hex_1$ and the derivation (\ref{eq_FG_der_2}) of $\Hex_2$ are totally different.

We see that the membership problem is non-trivial even if the hypergraph whose membership we check consists of one hyperedge. In fact, we shall show that, if we are able to solve the problem MEM for hypergraphs with one hyperedge, then we can also solve it for arbitrary hypergraphs.

The quantitative argument is too weak for arbitrary fusion grammars. Nevertheless, surprisingly, it can be used to prove Theorem \ref{th_main_membership} and Theorem \ref{th_main_non-emptiness} after modifying a given fusion grammar in an appropriate way. Namely, given a fusion grammar, we shall encode the information about the structure of its start hypergraph in labels of its hyperedges. After this encoding, we shall use the same quantitative argument as for connection-preserving grammars, and it will work successfully because, in the modified grammar, the fact that the number of $A$-labeled hyperedges equals the number of $\overline{A}$-labeled hyperedges for $A \in F$ implies that the hypergraph obtained after fusions has a certain structure. 


\subsection{Membership Problem For Hypergraphs with One Hyperedge}\label{ssec_results_mem1}

We start proving Theorem \ref{th_main_membership}. First, let us consider the variant of the membership problem where the input is a fusion grammar without markers and connectors $FG$ and a hypergraph $H_0$ that has just one hyperedge.

\begin{problem}[MEM-1]\label{problem_mem1}
	\leavevmode
	\\
	\textit{Input.} 
	\begin{enumerate}
		\item A fusion grammar without markers and connectors $FG = (Z,F,T)$; 
		\item A hypergraph $H_0$ such that $V_{H_0} = \{v_0^1,\dotsc,v_0^n\}$; $E_{H_0} = \{e_0\}$; $lab_{H_0}(e_0) = a_0 \in T$; $H_0$ is connected (equivalently, it does not contain isolated vertices).
	\end{enumerate}
	Let $N = \type(e_0)$; clearly, $N \ge n$.
	\\
	\textit{Problem MEM-1.} Does $H_0$ belong to $L(FG)$?
\end{problem}

\begin{theorem}\label{th_mem1}
	The problem MEM-1 is decidable.
\end{theorem}

To prove it, we need a number of auxiliary definitions. 
\begin{definition}
	The set $\Colors$ of \emph{colours} is $\{v_0^1,\dotsc,v_0^n,w\}$ where $w$ is a new element called the \emph{white colour}.
\end{definition}
\begin{definition}
	A \emph{coloured hypergraph} is a pair $H^\prime = (H,col)$ where $H$ is a hypergraph and $col:V_{H} \to \Colors$ is \emph{a colouring function}.
	
	The colouring function of a coloured hypergraph $H^\prime$ is also denoted by $col_{H^\prime}$.
\end{definition}
\begin{example}\label{example_colouring_Zex0}
	Suppose that $H_0 = \vcenter{\hbox{{\tikz[baseline=.1ex]{
					\node[vertex] (VL) at (-0.5,0) {};
					\node[vertex] (VR) at (0.5,0) {};
					\draw[-latex, thick] (VL) -- node[above] {$a$} (VR);
	}}}}
	$ (where $a_0 = a$) and $FG$ is the fusion grammar $\FGex$ from Example \ref{example_FG}. Then $\Colors = \{v_0^1,v_0^2,w\}$ where $v_0^i$ is the $i$-th attachment vertex of the $a$-labeled hyperedge of $H_0$. Let us call $v_0^1$ the \emph{grey colour} and depict it using grey; let us also call $v_0^2$ the \emph{light grey colour}. 
	
	Let us take $\Zex_0 = \vcenter{\hbox{{\tikz[baseline=.1ex]{
					\node[vertex] (VO) at (0,0) {};
					\node[vertex] (VL) at (-0.5,-0.5) {};
					\node[vertex] (VR) at (0.5,-0.5) {};
					\draw[-latex, thick] (VL) -- node[below] {$a$} (VR);
					\draw[-latex, thick] (VL) -- node[above left] {$A$} (VO);
					\draw[-latex, thick] (VR) -- node[above right] {$C$} (VO);
	}}}}$; let $v_1$ be its leftmost vertex, $v_2$ be its rightmost vertex, and $v_3$ its uppermost one. Let us define the colouring function $\colex_0$ as follows: $\colex_0(v_1) = v_0^1$, $\colex_0(v_2) = v_0^2$, $\colex_0(v_3) = w$. The result is depicted as follows:
	$$\left(\Zex_0,\colex_0\right) = \vcenter
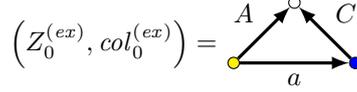
{\hbox{{\tikz[baseline=.1ex]{
					\node[vertex,fill={white}] (VO) at (0,0) {};
					\node[vertex,fill={colorA}] (VL) at (-0.6,-0.6) {};
					\node[vertex,fill={colorB}] (VR) at (0.6,-0.6) {};
					\draw[-latex, thick] (VL) -- node[below] {$a$} (VR);
					\draw[-latex, thick] (VL) -- node[above left] {$A$} (VO);
					\draw[-latex, thick] (VR) -- node[above right] {$C$} (VO);
	}}}}
	$$
\end{example}
\begin{definition}
	The set of \emph{vertex colourings} of a hypergraph $H$ is the set $\VC(H) = \{(H,col) \mid col:V_H \to \Colors\}$ of coloured hypergraphs corresponding to $H$. (Clearly, it is finite.)
\end{definition}
\begin{remark}\label{remark_meaning_of_colours}
	Informally, the purpose of using colours is to prevent attachment vertices that are different in $H_0$ from being identified by means of fusions. Namely, assume that $H_0 \in L(FG)$. There is a derivation of the form $Z \Rightarrow m \cdot Z \underset{\fr(P)}{\Rightarrow} H_0 + G$. We want to colour all the vertices in $Z^\prime = m \cdot Z$ in such a way that $att_{Z^\prime}(e_0)(i)$ has the colour $att_{H_0}(e_0)(i) = v_0^j$. A colour of a vertex $v$ in $Z^\prime$ tells one with which attachment vertex of $e_0$ this vertex is identified by $\fr(P)$. So, for example, if we colour one vertex in grey and another one in light grey, then they must not be fused. The white colour is designed to colour vertices that belong to the ``garbage'' hypergraph $G$ after the fusion.

	However, this is not enough. We must also \emph{force} attachment vertices from $Z^\prime$ that are the same in $H_0$ to be identified by fusions. Assume that $att_{H_0}(e_0)(i_1) = att_{H_0}(e_0)(i_2)$ but $v_1 = att_{Z^\prime}(e_0)(i_1) \ne att_{Z^\prime}(e_0)(i_2) = v_2$. Then there must be an evidence $\fr(P)$-path between $v_1$ and $v_2$. We are going to encode the information both about colours and evidence paths in labels of hyperedges.
\end{remark}

\begin{definition}
	A \emph{connection pair $(i_1,i_2)$} is a pair $i_1,i_2 \in \Nat$ such that $1 \le i_1 < i_2 \le N$ (recall that $N = \type(e_0)$) and $att_{H_0}(e_0)(i_1) = att_{H_0}(e_0)(i_2)$.
\end{definition}
\begin{definition}
	Let us define a new fusion alphabet $\FCol$ as follows: $A \in \FCol$ if and only if $A = (\sigma,f,\tau)$ where 
	\begin{itemize}
		\item $\sigma \in \Sigma \setminus \overline{F}$ (let $\type(\sigma) = r$);
		\item $f:[r] \to \Colors$ is a function mapping elements of $\{1,\dotsc,r\}$ to colours;
		\item $\tau$ is the set of all triples $(i_1,i_2,j)$ such that $j \in [r]$, $(i_1,i_2)$ is a connection pair, and $f(j) = att_{H_0}(e_0)(i_1) = att_{H_0}(e_0)(i_2)$.
	\end{itemize}
	The type of $A$ is defined as $\type(A) = \type(\sigma)$.
\end{definition}
{\color{comment}
	We are going to replace labels of vertex colourings of the connected components of the start hypergraph $Z$ by the new ones. These new fusion labels contain threefold information. If a hyperedge $e$ is labeled by $A = (\sigma,f,\tau)$, then this is intended to mean that
	\begin{enumerate}
		\item the old label of $e$ is $\sigma$;
		\item the colour of the $i$-th attachment vertex of $e$ is $f(i)$;
		\item $(i_1,i_2,j) \in \tau$ means that the hyperedge $e$ and its $j$-th attachment vertex are a part of the evidence path from the $i_1$-th attachment vertex of $e_0$ to the $i_2$-th one (consequently, these vertices must be of the same colour).
	\end{enumerate}
}

Clearly, the alphabet $\FCol$ is finite since the number of functions $f:[r] \to \Colors$ is finite as well as the number of triples described above. Hereinafter we assume without loss of generality that $\Sigma$ and $\FCol$ are disjoint. 

\begin{definition}
	Let $\SigmaCol = \{a_0\} \uplus \FCol \uplus \overline{\FCol}$ where $\overline{\FCol} = \{\overline{B} \mid B \in \FCol\}$ is the new label alphabet.
\end{definition}

\begin{remark}
	Hereinafter we identify $\overline{(A,f,\tau)}$ with the triple $(\overline{A},f,\tau)$ for $A \in F$, $(A,f,\tau) \in \FCol$. Note that $\FCol$ also contains triples of the form $(a,f,\tau)$ for $a \in T$. In such a case, $\overline{(a,f,\tau)}$ is just a new element, which is not equal to $(\overline{a},f,\tau)$, because there is no $\overline{a}$.
\end{remark}

\begin{definition}
	A label $(\sigma,f,\tau) \in \SigmaCol$ is called \emph{white} if $f(i)=w$ for all $i = 1,\dotsc, \type(\sigma)$ and $\tau = \emptyset$. The set of white labels is denoted by $\mathfrak{W}$.
\end{definition}
Informally, white labels are expected to be placed on hyperedges of the ``garbage'' hypergraph $G$ which appears in a derivation of $H_0$.

\begin{definition}
	Let $H$ be a hypergraph and let $v \in V_H$. Then the \emph{incidence set} $\Inc_H(v)$ equals \\ $\{(e,j) \mid e \in E_H, j \in [\type(e)], att_H(e)(j)=v \}$.
\end{definition}
\noindent
{\color{comment}
	The incidence set contains information about all the hyperedges attached to $v$.
}

\begin{definition}
	Let $H \in \mathcal{H}\left(\SigmaCol\right)$, $v \in V_H$, and let $(i_1,i_2)$ be a connection pair. Then 
	$$\Cond_H(v,i_1,i_2) = \Inc_H(v) \cap \{(e,j) \mid lab_H(e) = (\sigma,f,\tau) \mbox{ and } (i_1,i_2,j) \in \tau\}$$
	is a \emph{conduction set}.
\end{definition}
\noindent
{\color{comment}
	Roughly speaking, the conduction set $\Cond_H(v,i_1,i_2)$ is supposed to correspond to the set of hyperedges attached to $v$ that belong to an evidence path from $att_{Z^\prime}(e_0)(i_1)$ to $att_{Z^\prime}(e_0)(i_2)$ together with $v$.
}

\begin{definition}\label{def_EC1}
	Let $C \in \mathcal{C}(Z)$ and let $C^\prime \in \VC(C)$ be a coloured hypergraph. The set $\EC(C^\prime)$ of its \emph{hyperedge encodings} consists of hypergraphs $C^{\prime\prime} \in \mathcal{H}\left(\FCol \cup \overline{\FCol}\right)$ such that
	\begin{itemize}
		\item $V_{C^{\prime\prime}} = V_C$; \quad $E_{C^{\prime\prime}} = E_C$; \quad $att_{C^{\prime\prime}} = att_C$;
		\item $lab_{C^{\prime\prime}}(e) = (\sigma,f,\tau)$ with the following conditions being satisfied:
		\begin{enumerate}
			\item (labeling condition) $\sigma = lab_C(e)$;
			\item (colouring condition) $f(i) = col_{C^\prime}(att_{C}(e)(i))$;
			\item (cardinality condition) for any $v \in V_{C^{\prime\prime}}$ and for any connection pair $(i_1,i_2)$ the cardinality of $\Cond_{C^{\prime\prime}}(v,i_1,i_2)$ is either 0 or 2.
		\end{enumerate}
	\end{itemize}
\end{definition}
\noindent
{\color{comment}
	The intended meaning of the cardinality condition is that an evidence path from $att_{Z^\prime}(e_0)(i_1)$ to $att_{Z^\prime}(e_0)(i_2)$ either passes through a vertex $v \in V_{C^{\prime\prime}}$ one time (then the cardinality of $\Cond_{C^{\prime\prime}}(v,i_1,i_2)$ is 2) or it does not pass through it (then $\vert \Cond_{C^{\prime\prime}}(v,i_1,i_2) \vert = 0$).
}

\begin{example}\label{example_colouring_Zex12}
	Let us take the hypergraphs $\Zex_1$, $\Zex_2$ from Example \ref{example_FG} and colour them as follows:
	$$
	(\Zex_1,\colex_1) = \;
	\vcenter{\hbox{{\tikz[baseline=.1ex]{
					\node[vertex, fill={colorA}] (V1) at (0,0) {};
					\node[vertex, fill={white}] (V2) at (0,0.8) {};
					\node[vertex, fill={white}] (V3) at (0,1.6) {};
					\draw[-latex, thick] (V1) -- node[right] {$\overline{A}$} (V2);
					\draw[-latex, thick] (V2) -- node[right] {$B$} (V3);
	}}}}
	\qquad
	\qquad
	(\Zex_2,\colex_2) = \;
	\vcenter{\hbox{{\tikz[baseline=.1ex]{
					\node[vertex, fill={colorB}] (V1) at (0,0) {};
					\node[vertex, fill={white}] (V2) at (0,0.8) {};
					\node[vertex, fill={white}] (V3) at (0,1.6) {};
					\draw[-latex, thick] (V1) -- node[right] {$\overline{C}$} (V2);
					\draw[-latex, thick] (V2) -- node[right] {$B$} (V3);
	}}}}
	$$
	Then the hypergraphs
	$$
	\vcenter{\hbox{{\tikz[baseline=.1ex]{
					\node[vertex] (V1) at (0,0) {};
					\node[vertex] (V2) at (0,0.8) {};
					\node[vertex] (V3) at (0,1.6) {};
					\draw[-latex, thick] (V1) -- node[right] {$(\overline{A},f,\emptyset)$} (V2);
					\draw[-latex, thick] (V2) -- node[right] {$(B,g,\emptyset)$} (V3);
	}}}}
	\qquad \mbox{and} \qquad
	\vcenter{\hbox{{\tikz[baseline=.1ex]{
					\node[vertex] (V1) at (0,0) {};
					\node[vertex] (V2) at (0,0.8) {};
					\node[vertex] (V3) at (0,1.6) {};
					\draw[-latex, thick] (V1) -- node[right] {$(\overline{C},h,\emptyset)$} (V2);
					\draw[-latex, thick] (V2) -- node[right] {$(B,g,\emptyset)$} (V3);
	}}}}
	$$
	belong to $\EC(\Zex_1,\colex_1)$ and $\EC(\Zex_2,\colex_2)$ resp. where $f,g,h$ are defined as follows: $f(1) = v_0^1$ (the grey colour), $h(1) = v_0^2$ (the light grey colour), $f(2) = h(2) = g(1) = g(2) = w$ (the white colour).
\end{example}

\begin{example}
	Since our running example is $H_0 = \vcenter{\hbox{{\tikz[baseline=.1ex]{
					\node[vertex] (VL) at (-0.5,0) {};
					\node[vertex] (VR) at (0.5,0) {};
					\draw[-latex, thick] (VL) -- node[above] {$a$} (VR);
	}}}}$, there exist no connection pairs. Indeed, only $(1,2)$ might be a connection pair, but $att_{H_0}(e_0)(1) \ne att_{H_0}(e_0)(2)$. Therefore, all labels from $\FCol$ are of the form $(\sigma,f,\emptyset)$, and the third component is redundant.
\end{example}

\begin{definition}\label{def_base_hg}
	Let $C \in \mathcal{C}(Z)$ be a hypergraph with a hyperedge $e_0^\prime \in E_C$ such that
	\begin{enumerate}
		\item $lab_C(e_0^\prime) = a_0$;
		\item $att_C(e_0^\prime)(i_1) = att_C(e_0^\prime)(i_2)$ implies $att_{H_0}(e_0)(i_1) = att_{H_0}(e_0)(i_2)$.
	\end{enumerate}
	Then $(C,e_0^\prime)$ is a \emph{base hypergraph}. The set of base hypergraphs is denoted by $\mathcal{C}^b(Z)$.
\end{definition}

{\color{comment}
	Let $Z \Rightarrow m \cdot Z \underset{\fr(P)}{\Rightarrow} H_0 + G$ be a derivation of $H_0$. Clearly, there must be a connected component $C$ in $m \cdot Z$ that contains the $a_0$-labeled hyperedge which, after fusions, coincides with the hyperedge $e_0$ of $H_0$. This component $C$ (along with the distinguished hyperedge corresponding to $e_0$) is what we want to call a base hypergraph. 
	The second condition of Definition \ref{def_base_hg} says that, if some attachment vertices of $e_0^\prime$ coincide in $C$, then so do corresponding vertices of $e_0$ in $H_0$. If this is not the case, then the component $C$ with the hyperedge $e_0^\prime$ cannot be the one from which $H_0$ is obtained after the fusions because attachment vertices of $e_0^\prime$ cannot cease to coincide, i.e.~to be somehow ``unfused''.
}

\begin{definition}\label{def_coloured_base_hg}
	A \emph{coloured base hypergraph} is a triple $C^\prime = (C,e_0^\prime,col)$ where $(C,e_0^\prime)$ is a base hypergraph, $(C,col)$ is a coloured hypergraph, and $col(att_C(e_0^\prime)(i)) = att_{H_0}(e_0)(i)$.
	\\
	The function $col$ of a coloured base hypergraph $C^\prime$ will also be denoted by $col_{C^\prime}$.
\end{definition}

Informally, a coloured base hypergraph is a coloured hypergraph such that its distinguished hyperedge's colouring is consistent with the colours of the hyperedge $e_0$.

\begin{example}
	The triple $\left(\Zex_0,e_0^a,\colex_0\right)$ is a coloured base hypergraph where $\colex_0$ is defined in Example \ref{example_colouring_Zex0}. Indeed, the first attachment vertex $v_1$ of the $a$-labeled hyperedge $e_0^a$ in $\Zex_0$ has the colour $v_0^1$, which is the first attachment vertex of the $a$-labeled hyperedge in $H_0 = \vcenter{\hbox{{\tikz[baseline=.1ex]{
					\node[vertex] (VL) at (-0.5,0) {};
					\node[vertex] (VR) at (0.5,0) {};
					\draw[-latex, thick] (VL) -- node[above] {$a$} (VR);
	}}}}$; similarly, the second attachment vertex $v_2$ of $e_0^a$ has the colour $v_0^2$. 

	Note that, if we colour vertices of $\Zex_0$ according to $\colex_0$ and also colour vertices of $\Zex_1,\Zex_2$ as shown in Example \ref{example_colouring_Zex12}, then, in the derivation (\ref{eq_FG_der_2}), only vertices of the same colour are identified:
	
	\begin{equation*}
		\Zex 
		\;\;\Rightarrow\;\;
		\vcenter
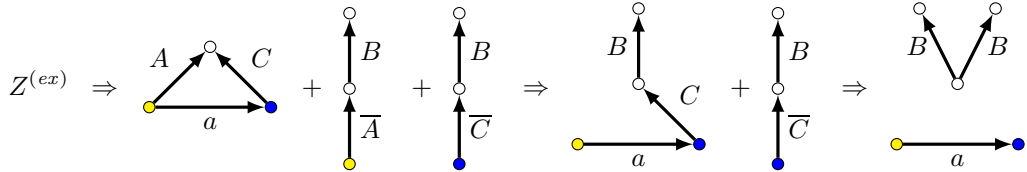
{\hbox{{\tikz[baseline=.1ex]{
						\node[vertex,fill={white}] (VO) at (0,0) {};
						\node[vertex,fill={colorA}] (VL) at (-0.6,-0.6) {};
						\node[vertex,fill={colorB}] (VR) at (0.6,-0.6) {};
						\draw[-latex, thick] (VL) -- node[below] {$a$} (VR);
						\draw[-latex, thick] (VL) -- node[above left] {$A$} (VO);
						\draw[-latex, thick] (VR) -- node[above right] {$C$} (VO);
		}}}}
		\;\;+\;\;
		\vcenter{\hbox{{\tikz[baseline=.1ex]{
						\node[vertex,fill={colorA}] (V1) at (0,0) {};
						\node[vertex,fill={white}] (V2) at (0,0.8) {};
						\node[vertex,fill={white}] (V3) at (0,1.6) {};
						\draw[-latex, thick] (V1) -- node[right] {$\overline{A}$} (V2);
						\draw[-latex, thick] (V2) -- node[right] {$B$} (V3);
		}}}}
		\;\;+\;\;
		\vcenter{\hbox{{\tikz[baseline=.1ex]{
						\node[vertex,fill={colorB}] (V1) at (0,0) {};
						\node[vertex,fill={white}] (V2) at (0,0.8) {};
						\node[vertex,fill={white}] (V3) at (0,1.6) {};
						\draw[-latex, thick] (V1) -- node[right] {$\overline{C}$} (V2);
						\draw[-latex, thick] (V2) -- node[right] {$B$} (V3);
		}}}}
		\;\;\Rightarrow\;\;
		\vcenter{\hbox{{\tikz[baseline=.1ex]{
						\node[vertex,fill={white}] (VO) at (0,0) {};
						\node[vertex,fill={colorA}] (VL) at (-0.6,-0.6) {};
						\node[vertex,fill={colorB}] (VR) at (0.6,-0.6) {};
						\node[vertex,fill={white}] (V1) at (0,0.75) {};
						\draw[-latex, thick] (VL) -- node[below] {$a$} (VR);
						\draw[-latex, thick] (VR) -- node[above right] {$C$} (VO);
						\draw[-latex, thick] (VO) -- node[left] {$B$} (V1);
		}}}}
		\;\;+\;\;
		\vcenter{\hbox{{\tikz[baseline=.1ex]{
						\node[vertex,fill={colorB}] (V1) at (0,0) {};
						\node[vertex,fill={white}] (V2) at (0,0.8) {};
						\node[vertex,fill={white}] (V3) at (0,1.6) {};
						\draw[-latex, thick] (V1) -- node[right] {$\overline{C}$} (V2);
						\draw[-latex, thick] (V2) -- node[right] {$B$} (V3);
		}}}}
		\;\;\Rightarrow\;\;
		\vcenter{\hbox{{\tikz[baseline=.1ex]{
						\node[vertex,fill={white}] (VO) at (0,0) {};
						\node[vertex,fill={colorA}] (VL) at (-0.5,-0.5) {};
						\node[vertex,fill={colorB}] (VR) at (0.5,-0.5) {};
						\node[vertex,fill={white}] (V1) at (-0.5,0.7) {};
						\node[vertex,fill={white}] (V2) at (0.5,0.7) {};
						\draw[-latex, thick] (VL) -- node[below] {$a$} (VR);
						\draw[-latex, thick] (VO) -- node[below left] {$B$} (V1);
						\draw[-latex, thick] (VO) -- node[below right] {$B$} (V2);
		}}}}
	\end{equation*}
	Moreover, in the resulting hypergraph, all vertices that are not incident to the $a$-labeled hyperedge are coloured white, as we wanted (see Remark \ref{remark_meaning_of_colours}).
\end{example}


\begin{definition}
	The set $\VC^b(H,e_0^\prime)$ of \emph{vertex colourings} of a base hypergraph $(H,e_0^\prime)$ consists of coloured base hypergraphs of the form $(H,e_0^\prime, col)$.
\end{definition}

\begin{definition}\label{def_EC2}
	Let $C^\prime = (C,e_0^\prime,col_{C^\prime})$ be a coloured base hypergraph. The set $\EC^b(C^\prime)$  of its \emph{hyperedge encodings} consists of hypergraphs $C^{\prime\prime} \in \mathcal{H}(\SigmaCol)$ where
	\begin{itemize}
		\item $V_{C^{\prime\prime}} = V_C$; \qquad $E_{C^{\prime\prime}} = E_C$; \qquad $att_{C^{\prime\prime}} = att_C$;
		\item $lab_{C^{\prime\prime}}(e_0^\prime) = a_0$;
		\item for any $e \in E_{C^{\prime\prime}}$, $e \ne e_0^\prime$ it holds that $lab_{C^{\prime\prime}}(e) = (\sigma,f,\tau)$ with the following conditions being satisfied:
		\begin{enumerate}
			\item (labeling condition) $\sigma = lab_C(e)$;
			\item (colouring condition) $f(i) = col_{C^\prime}(att_{C}(e)(i))$;
			\item (cardinality condition-1) if, given $v \in V_{C}$ and a connection pair $(i_1,i_2)$, it is the case that $att_{C}(e_0^\prime)(i_1)\ne v$ and $att_{C}(e_0^\prime)(i_2) \ne v$, then $\vert \Cond_{C^{\prime\prime}}(v,i_1,i_2) \vert$ equals either $0$ or $2$;
			\item (cardinality condition-2) if, given $v \in V_{C}$ and a connection pair $(i_1,i_2)$, it is the case that either $att_{C}(e_0^\prime)(i_1)=v$ or $att_{C}(e_0^\prime)(i_2)=v$ but not both, then $\vert \Cond_{C^{\prime\prime}}(v,i_1,i_2) \vert = 1$;
			\item (cardinality condition-3) if, given $v \in V_{C}$ and a connection pair $(i_1,i_2)$, it is the case that $att_{C}(e_0^\prime)(i_1)=att_{C}(e_0^\prime)(i_2)= v$, then $\vert \Cond_{C^{\prime\prime}}(v,i_1,i_2) \vert = 0$.
		\end{enumerate}
	\end{itemize}
\end{definition}

{\color{comment}
	The cardinality condition is different from that in Definition \ref{def_EC1}. Recall Remark \ref{remark_meaning_of_colours}: we want to say that the evidence path from $v_1$ to $v_2$ starts at $v_1$ and thus contains exactly one hyperedge attached to $v_1$; similarly, it finishes at $v_2$ and hence contains exactly one hyperedge attached to $v_2$. This motivates the cardinality condition-2. For any other vertex, the evidence path passes through it at most once, as in Definition \ref{def_EC1}, so $\vert \Cond_{C^{\prime\prime}}(v,i_1,i_2) \vert$ equals either $0$ or $2$.
}

\begin{example}
	Consider the coloured base hypergraph $\left(\Zex_0,e_0^a,\colex_0\right)$ defined in Example \ref{example_colouring_Zex0}. Then the set $\EC^b(\Zex_0,e_0^a,\colex_0)$ contains the hypergraph
	$$
	\vcenter{\hbox{{\tikz[baseline=.1ex]{
					\node[vertex] (VO) at (0,0) {};
					\node[vertex] (VL) at (-0.6,-0.6) {};
					\node[vertex] (VR) at (0.6,-0.6) {};
					\draw[-latex, thick] (VL) -- node[below] {$a$} (VR);
					\draw[-latex, thick] (VL) -- node[above left] {$(A,f,\emptyset)$} (VO);
					\draw[-latex, thick] (VR) -- node[above right] {$(C,h,\emptyset)$} (VO);
	}}}}
	$$
	The functions $f,h$ are the same as those from Example \ref{example_colouring_Zex12}: $f(1) = v_0^1$ (grey), $h(1) = v_0^2$ (light grey), $f(2) = h(2) = w$ (white). So, for example, the label $(A,f,\emptyset)$ says that the first attachment vertex of its hyperedge has the colour $f(1) = v_0^1$, i.e.~is grey. 
	
	Note that the derivation (\ref{eq_FG_der_2}) remains valid if we replace $\Zex_i$ by hyperedge encodings of their vertex colourings as follows:
	\begin{equation*}
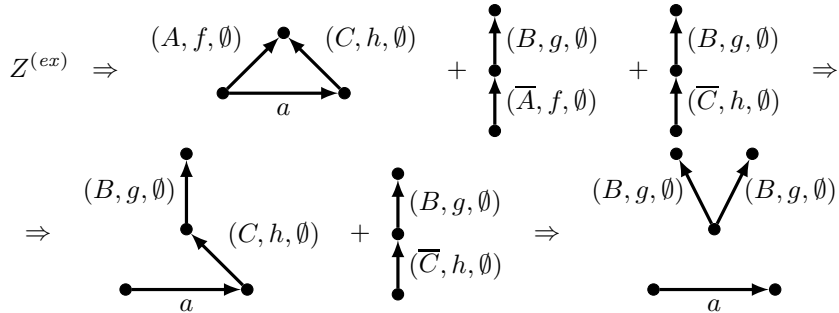
		\begin{aligned}
		\Zex 
		& \Rightarrow\quad
		\vcenter{\hbox{{\tikz[baseline=.1ex]{
						\node[vertex] (VO) at (0,0) {};
						\node[vertex] (VL) at (-0.6,-0.6) {};
						\node[vertex] (VR) at (0.6,-0.6) {};
						\draw[-latex, thick] (VL) -- node[below] {$a$} (VR);
						\draw[-latex, thick] (VL) -- node[above left] {$(A,f,\emptyset)$} (VO);
						\draw[-latex, thick] (VR) -- node[above right] {$(C,h,\emptyset)$} (VO);
		}}}}
		\quad+\quad
		\vcenter{\hbox{{\tikz[baseline=.1ex]{
						\node[vertex] (V1) at (0,0) {};
						\node[vertex] (V2) at (0,0.8) {};
						\node[vertex] (V3) at (0,1.6) {};
						\draw[-latex, thick] (V1) -- node[right] {$(\overline{A},f,\emptyset)$} (V2);
						\draw[-latex, thick] (V2) -- node[right] {$(B,g,\emptyset)$} (V3);
		}}}}
		\quad+\quad
		\vcenter{\hbox{{\tikz[baseline=.1ex]{
						\node[vertex] (V1) at (0,0) {};
						\node[vertex] (V2) at (0,0.8) {};
						\node[vertex] (V3) at (0,1.6) {};
						\draw[-latex, thick] (V1) -- node[right] {$(\overline{C},h,\emptyset)$} (V2);
						\draw[-latex, thick] (V2) -- node[right] {$(B,g,\emptyset)$} (V3);
		}}}}
		\quad\Rightarrow
		\\
		& \Rightarrow\quad
		\vcenter{\hbox{{\tikz[baseline=.1ex]{
						\node[vertex] (VO) at (0,0) {};
						\node[vertex] (VL) at (-0.6,-0.6) {};
						\node[vertex] (VR) at (0.6,-0.6) {};
						\node[vertex] (V1) at (0,1) {};
						\draw[-latex, thick] (VL) -- node[below] {$a$} (VR);
						\draw[-latex, thick] (VR) -- node[above right] {$(C,h,\emptyset)$} (VO);
						\draw[-latex, thick] (VO) -- node[left] {$(B,g,\emptyset)$} (V1);
		}}}}
		\quad+\quad
		\vcenter{\hbox{{\tikz[baseline=.1ex]{
						\node[vertex] (V1) at (0,0) {};
						\node[vertex] (V2) at (0,0.8) {};
						\node[vertex] (V3) at (0,1.6) {};
						\draw[-latex, thick] (V1) -- node[right] {$(\overline{C},h,\emptyset)$} (V2);
						\draw[-latex, thick] (V2) -- node[right] {$(B,g,\emptyset)$} (V3);
		}}}}
		\quad\Rightarrow\quad
		\vcenter{\hbox{{\tikz[baseline=.1ex]{
						\node[vertex] (VO) at (0,-0.1) {};
						\node[vertex] (VL) at (-0.6,-0.5) {};
						\node[vertex] (VR) at (0.6,-0.5) {};
						\node[vertex] (V1) at (-0.3,0.7) {};
						\node[vertex] (V2) at (0.3,0.7) {};
						\draw[-latex, thick] (VL) -- node[below] {$a$} (VR);
						\draw[-latex, thick] (VO) -- node[left] {$(B,g,\emptyset)$} (V1);
						\draw[-latex, thick] (VO) -- node[right] {$(B,g,\emptyset)$} (V2);
		}}}}
		\end{aligned}
	\end{equation*}
	The hypergraphs used in this derivation are from Example \ref{example_colouring_Zex12}. The label $(B,g,\emptyset)$ is white, because $g(1) = g(2) = w$.
\end{example}

\begin{example}
	In the above examples starting from Example \ref{example_colouring_Zex0}, we have been considering the hypergraph $\vcenter
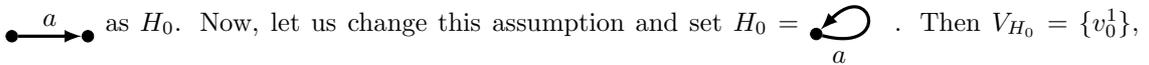
{\hbox{{\tikz[baseline=.1ex]{
					\node[vertex] (VL) at (-0.5,0) {};
					\node[vertex] (VR) at (0.5,0) {};
					\draw[-latex, thick] (VL) -- node[above] {$a$} (VR);
	}}}}$ as $H_0$. Now, let us change this assumption and take $H_0$ to be $\vcenter{\hbox{{\tikz[baseline=.1ex]{
					\node[vertex] (V) at (0,0) {};
					\draw[-latex, thick] (V) to[out=-15,in=45,looseness=30] (V);
					\node (V) at (0.3,-0.3) {$a$};
	}}}}
	$. Then, $V_{H_0} = \{v_0^1\}$, $E_{H_0} = \{e_0\}$, $att_{H_0}(e_0) = v_0^1v_0^1$; therefore $\Colors = \{v_0^1,w\}$. Let us call $v_0^1$ the \emph{green colour}.\footnote{In the black and white version of this article, green is displayed as grey in the below examples.} Given this $H_0$, the pair $(1,2)$ is a connection pair because $att_{H_0}(e_0)(1) = att_{H_0}(e_0)(2)$.

	Below we show examples of vertex colourings of $\Zex_0,\Zex_3,\Zex_4$ using new colours (white and green) and examples of their hyperedge encodings $\Cex_0,\Cex_3,\Cex_4$.
	\begin{itemize}
		\item 
		$(\Zex_0,\widetilde{\mathit{col}}^{(\mathit{ex})}_0) = 
		\vcenter{\hbox{{\tikz[baseline=.1ex]{
						\node[vertex,fill={white}] (VO) at (0,0) {};
						\node[vertex,fill={green}] (VL) at (-0.7,-0.7) {};
						\node[vertex,fill={green}] (VR) at (0.7,-0.7) {};
						\draw[-latex, thick] (VL) -- node[below] {$a$} (VR);
						\draw[-latex, thick] (VL) -- node[above left] {$A$} (VO);
						\draw[-latex, thick] (VR) -- node[above right] {$C$} (VO);
		}}}}
		\qquad\;\,
		\Cex_0 =
		\vcenter{\hbox{{\tikz[baseline=.1ex]{
						\node[vertex] (VO) at (0,0) {};
						\node[vertex] (VL) at (-0.7,-0.7) {};
						\node[vertex] (VR) at (0.7,-0.7) {};
						\draw[-latex, thick] (VL) -- node[below] {$a$} (VR);
						\draw[-latex, thick] (VL) -- node[above left] {$(A,f_1,\tau_1)$} (VO);
						\draw[-latex, thick] (VR) -- node[above right] {$(C,f_6,\tau_6)$} (VO);
		}}}}
		$
		\item 
		$(\Zex_3,\colex_3) = 
		\vcenter{\hbox{{\tikz[baseline=.1ex]{
						\node[vertex,fill={green}] (V1) at (0,0) {};
						\node[vertex,fill={white}] (V2) at (0,0.7) {};
						\draw[-latex, thick] (V1) to[bend left=45] node[left] {$\overline{A}$} (V2);
						\draw[-latex, thick] (V2) to[bend left=45] node[right] {$B$} (V1);
		}}}}
		\qquad\qquad
		\Cex_3 = 
		\vcenter{\hbox{{\tikz[baseline=.1ex]{
						\node[vertex] (V1) at (0,0) {};
						\node[vertex] (V2) at (0,0.7) {};
						\draw[-latex, thick] (V1) to[bend left=45] node[left] {$(\overline{A},f_2,\tau_2)$} (V2);
						\draw[-latex, thick] (V2) to[bend left=45] node[right] {$(B,f_3,\tau_3)$} (V1);
		}}}}
		$
		\item 
		$(\Zex_4,\colex_4) = 
		\vcenter{\hbox{{\tikz[baseline=.1ex]{
						\node[vertex,fill={green}] (V1) at (0,0) {};
						\node[vertex,fill={white}] (V2) at (0,0.7) {};
						\draw[latex-, thick] (V1) to[bend left=45] node[left] {$\overline{B}$} (V2);
						\draw[latex-, thick] (V2) to[bend left=45] node[right] {$\overline{C}$} (V1);
		}}}}
		\qquad\qquad
		\Cex_4 = 
		\vcenter{\hbox{{\tikz[baseline=.1ex]{
						\node[vertex] (V1) at (0,0) {};
						\node[vertex] (V2) at (0,0.7) {};
						\draw[latex-, thick] (V1) to[bend left=45] node[left] {$(\overline{B},f_4,\tau_4)$} (V2);
						\draw[latex-, thick] (V2) to[bend left=45] node[right] {$(\overline{C},f_5,\tau_5)$} (V1);
		}}}}
		$
	\end{itemize}
	Here, for $i=1,2,5,6$, $f_i(1) = v_0^1$ (green), $f_i(2) = w$ (white), and $\tau_i = \{(1,2,1)\}$; for $j=3,4$, $f_j(1) = w$, $f_j(2) = v_0^1$, and $\tau_j = \{(1,2,2)\}$.
	For example, $(1,2,1) \in \tau_2$ indicates that the first attachment vertex of $e_3^{\overline{A}}$ is a part of an evidence path from $att_{\Zex_0}(e_0^a)(1)$ to $att_{\Zex_0}(e_0^a)(2)$. It holds that $\Cex_3 \in \EC(\Zex_3,\colex_3)$, $\Cex_4 \in \EC(\Zex_4,\colex_4)$, $\Cex_0 \in \EC^b(\Zex_0,e_0^a,\widetilde{\mathit{col}}^{(\mathit{ex})}_0)$.
	The labeling and the colouring conditions can be checked straightforwardly. Let us prove that the cardinality condition holds e.g.~for $\Cex_3$. Let $v_1^\prime$ be the lowermost vertex of this hypergraph and let $v_2^\prime$ be the uppermost one; recall also that $e_3^{\overline{A}}$ is its leftmost hyperedge and $e_3^B$ is its rightmost one. Then 
	\begin{itemize}
		\item $\Inc_{\Cex_3}(v_1^\prime) = \{(e_3^{\overline{A}},1), (e_3^B,2)\}$ and $\Inc_{\Cex_3}(v_2^\prime) = \{(e_3^{\overline{A}},2), (e_3^B,1)\}$;
		\item Since $(1,2,1) \in \tau_2$ and $(1,2,2) \in \tau_3$, it holds that 
		$$
		\left\{(e,j) \mid lab_{\Cex_3}(e) = (\sigma,f,\tau) \mbox{ and } (1,2,j) \in \tau\right\}
		=
		\{(e_3^{\overline{A}},1), (e_3^B,2)\}
		$$
		and, consequently, $\Cond_{\Cex_3}(v_1^\prime) = \Inc_{\Cex_3}(v_1^\prime) = \{(e_3^{\overline{A}},1), (e_3^B,2)\}$. On the other hand, $\Cond_{\Cex_3}(v_2^\prime) = \emptyset$. So, $\vert \Cond_{\Cex_3}(v_1^\prime) \vert = 2$ and $\vert \Cond_{\Cex_3}(v_2^\prime) \vert = 0$, as required. 
	\end{itemize}
	Note that $\overline{(A,f_1,\tau_1)} = (\overline{A}, f_1,\tau_1) = (\overline{A}, f_2,\tau_2)$; $\overline{(B,f_3,\tau_3)} = (\overline{B}, f_4,\tau_4)$; $\overline{(\overline{C},f_5,\tau_5)} = (C, f_6,\tau_6)$. The fusion scheme from Example \ref{example_der_evidence_path} remains valid if we replace $\Zex_i$ by $\Cex_i$:
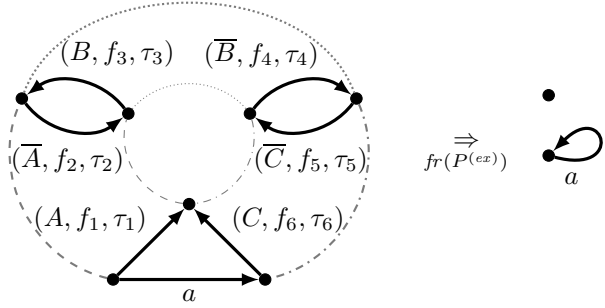
	$$
	\vcenter{\hbox{{\tikz[baseline=.1ex]{
					\node[vertex] (VO) at (0,0) {};
					\node[vertex] (VL) at (-1,-0.7) {};
					\node[vertex] (VR) at (1,-0.7) {};
					\draw[-latex, thick] (VL) -- node[below] {$a$} (VR);
					\draw[-latex, thick] (VL) -- (VO);
					\draw[-latex, thick] (VR) -- (VO);
					\node[] (Lb1) at (-1.17,-0.12) {$(A,f_1,\tau_1)$};
					\node[] (Lb6) at (1.17,-0.12) {$(C,f_6,\tau_6)$};
					\node[vertex] (VL1) at (-0.8,1.2-0.17) {};
					\node[vertex] (VL2) at (-2.2,1.4-0.17) {};
					\draw[latex-, thick] (VL1) to[bend left=30] (VL2);
					\draw[latex-, thick] (VL2) to[bend left=30] (VL1);
					\node[vertex] (VR1) at (0.8,1.2-0.17) {};
					\node[vertex] (VR2) at (2.2,1.4-0.17) {};
					\draw[-latex, thick] (VR1) to[bend left=30] (VR2);
					\draw[-latex, thick] (VR2) to[bend left=30] (VR1);
					\node[] (Lb2) at (-1.6,0.6) {$(\overline{A},f_2,\tau_2)$};
					\node[] (Lb3) at (-0.95,1.6) {$(B,f_3,\tau_3)$};
					\node[] (Lb5) at (1.6,0.6) {$(\overline{C},f_5,\tau_5)$};
					\node[] (Lb4) at (0.95,1.62) {$(\overline{B},f_4,\tau_4)$};
					\draw[-, comment, thick, dashed] (VL) to[bend left=60] (VL2);
					\draw[-, comment, dashed] (VO) to[bend left=50] (VL1);
					\draw[-, comment, thick, densely dotted] (VL2) to[bend left=55] (VR2);
					\draw[-, comment, densely dotted] (VL1) to[bend left=50] (VR1);
					\draw[-, comment, thick, dashdotted] (VR2) to[bend left=60] (VR);
					\draw[-, comment, dashdotted] (VR1) to[bend left=50] (VO);
	}}}}
	\quad \underset{\fr(\Pex)}{\Rightarrow} \quad
	\vcenter{\hbox{{\tikz[baseline=.1ex]{
					\node[vertex] (VO) at (0,0) {};
					\node[vertex] (V) at (0,-0.8) {};
					\draw[-latex, thick] (V) to[out=-15,in=45,looseness=30] (V);
					\node (V) at (0.3,-1.1) {$a$};
	}}}}
	$$
	The triples $\tau_i$ ($i=1,\dotsc,6$) contain information about the evidence path from the leftmost attachment vertex of $e_0^a$ to its rightmost one (presented in Example \ref{example_der_evidence_path}).
\end{example}

{\color{comment}
	We have presented all necessary encodings. We have introduced colours corresponding to vertices of the target hypergraph $H_0$ and also evidence paths. We have encoded information about colours and about evidence paths in hyperedge labels. A new alphabet of fusion labels has been denoted by $\FCol$, and $\SigmaCol = \FCol \cup \overline{\FCol} \cup \{a_0\}$. 
	
	Now, finally, we are ready to present the main construction. Namely, given an input $(FG,H_0)$ of the problem MEM-1, we effectively construct a semilinear set which contains $\vec{0}$ if and only if $H_0 \in L(FG)$. 
}

\begin{definition}\label{def_beta}
	Let $\FCol \setminus \mathfrak{W} = \{A_1,\dotsc,A_d\}$ be the set of all \emph{non-white} fusion labels from $\FCol$. Given a hypergraph $H \in \mathcal{H}(\SigmaCol)$, we define $\beta(H) \in \Nat^{d}$ as follows: 
	$$\beta(H) = \left(\#_{A_1}(H)-\#_{\overline{A_1}}(H),\dotsc,\#_{A_d}(H)-\#_{\overline{A_d}}(H)\right).$$
\end{definition}

\begin{definition}\label{def_mem}
	Let $S,S^b$ be the following sets:
	\begin{equation*}
		\begin{aligned}
			S & = \bigcup\limits_{C \in \mathcal{C}(Z)}\; \bigcup\limits_{C^\prime \in \VC(C)} \; \bigcup\limits_{C^{\prime\prime} \in \EC(C^\prime)} \{C^{\prime\prime}\};
			\\
			S^b &= \bigcup\limits_{C \in \mathcal{C}^b(Z)}\; \bigcup\limits_{C^\prime \in \VC^b(C)} \; \bigcup\limits_{C^{\prime\prime} \in \EC^b(C^\prime)} \{C^{\prime\prime}\}.
		\end{aligned}
	\end{equation*}
	We define the semilinear set $\Mem(FG,H_0)$ as follows:
	\begin{equation*}
		\Mem(FG,H_0) \eqdef 
		\bigcup\limits_{B \in S^b} \left\{\beta(B)+\sum_{C \in S} k(C) \beta(C) \middle\vert k(C) \in \Nat \right\}.
	\end{equation*}
\end{definition}

{\color{comment}
	The set $S$ consists of hyperedge encodings of vertex colourings of connected components from $Z$; similarly, $S^b$ consists of hyperedge encodings of vertex colourings of base hypergraphs. Then, $\Mem(FG,H_0)$ consists of the sums of $\beta$-images of hypergraphs from $S$ and $S^b$ such that there is exactly one $\beta$-image of a hypergraph from $S^b$.
}

\begin{lemma}\label{lemma_mem1_1}
	If $H_0$ belongs to $L(FG)$, then $\vec{0}$ belongs to $\Mem(FG,H_0)$. 
\end{lemma}

\begin{remark}\label{remark_decidability_linear_integer_programming}
	Checking whether $\vec{0} \in \Mem(FG,H_0)$ is decidable. Indeed, this is equivalent to solvability of at least one system of equations 
	\begin{equation}\label{eq_system_mem1}
		\sum_{C \in S} k(C) \beta(C) = -\beta(B)
	\end{equation}
	with $k(C)$ being non-negative integer variables. This system has the form $Ax = b$ where $x \in \Nat^{\vert S \vert}$ consists of the unknown numbers $k(C)$ for $C \in S$, $b \in \mathbb{Z}^d$, and $A \in \mathbb{Z}^{d \times \vert S \vert}$ is a matrix composed of $\beta(C)$ for $C \in S$. Checking solvability of the system (\ref{eq_system_mem1}) is a variant of the linear integer programming problem as noticed in \cite{Lye21}, which is in NP (consult also \cite[p.245]{GareyJ79}). The number of equations (\ref{eq_system_mem1}) equals the number of $B \in S^b$, which is finite.
\end{remark}

\begin{proof}[Proof of Lemma \ref{lemma_mem1_1}]
	There is a most parallelised derivation of the form 
	$$
	Z \underset{m}{\Rightarrow} m \cdot Z \underset{\fr(P)}{\Rightarrow} H_0 + G
	$$
	where $G$ is some hypergraph. Let $Z^\prime = m \cdot Z$. Let $V_{H_0} = \{v_0^1,\dotsc,v_0^n\}$ and $E_{H_0} = \{e_0\}$ (as in the statement of the problem MEM-1). Note that, since $H_0+G$ is the result of the fusion rule $\fr(P)$, each vertex in $H_0+G$ is an equivalence class of vertices from $Z^\prime$ (cf. Section \ref{ssec_fusion_grammars}, item \ref{item_parallelised_fr}).

	The hypergraph $Z^\prime$ must include the hyperedge $e_0 \in E_{H_0}$; let $e_0 \in E_{B_0}$ for $B_0 \in \mathcal{C}(Z^\prime)$. Then $Z^\prime = B_0 + m^\prime \cdot Z$ where $m^\prime(B_0) = m(B_0)-1$ and $m^\prime(C) = m(C)$ otherwise. 
	The pair $(B_0,e_0)$ is a base hypergraph, because $lab_{B_0}(e_0) = a_0$ and $att_{B_0}(e_0)(i_1)=att_{B_0}(e_0)(i_2)$ implies $att_{H_0}(e_0)(i_1)=[att_{B_0}(e_0)(i_1)]_{\equiv_P} = [att_{B_0}(e_0)(i_2)]_{\equiv_P} = att_{H_0}(e_0)(i_2)$.
	
	To prove the lemma, we shall colour each connected component in $Z^\prime$ and also encode information about evidence paths justifying fusions of $e_0$'s vertices. This will result in defining a relabeling function $\chi:E_{Z^\prime} \to \SigmaCol$ such that $\chi(B_0) \in S^b$ is a hyperedge encoding of a coloured base hypergraph and such that, for $C \in \mathcal{C}(Z^\prime), C\ne B_0$, it holds that $\chi(C) \in S$. Finally, we shall show that $\beta(\chi(Z^\prime)) = \vec{0}$, which shall imply that $\vec{0} \in \Mem(FG,H_0)$.
	
	Let us start constructing such a relabeling. Let $(i_1,i_2)$ be a connection pair; this means that $att_{Z^\prime}(e_0)(i_1) \equiv_P att_{Z^\prime}(e_0)(i_2)$. Let us fix an evidence $\fr(P)$-path for each such pair of vertices and the corresponding evidence $\fr(P)$-set; let us denote them by $EvP(i_1,i_2)$ and $EvS(i_1,i_2)$ respectively. 
	
	Let $col_0:V_{Z^\prime} \to \Colors$ be defined as follows: $col_0(v) = v_0^j$ if and only if $[v]_{\equiv_P} = v_0^j$ (for $j \in [n]$); otherwise, $col_0(v) = w$.
	The function $\chi: E_{Z^\prime} \to \SigmaCol$ is defined as follows:
	\begin{itemize}
		\item $\chi(e_0) = a_0$;
		\item for each $e\ne e_0$, let $\chi(e) = (lab_{Z^\prime}(e),f,\tau)$ where
		\begin{itemize}
			\item $f(i) = col_0(att_{Z^\prime}(e)(i))$;
			\item $(i_1,i_2,j) \in \tau$ if and only if $(e,j)$ belongs to the evidence set $EvS(i_1,i_2)$ (i.e.~$(e,j)$ is a member of the fixed evidence path from $att_{Z^\prime}(e_0)(i_1)$ to $att_{Z^\prime}(e_0)(i_2)$).
		\end{itemize}
	\end{itemize}
	It is the case that $(C,col_0{\restriction}_{V_{C}})$ for $C \in \mathcal{C}(Z^\prime)$, $C \ne B_0$ is a coloured hypergraph and that $(B_0,e_0,{col_0}{\restriction}_{V_{B_0}})$ is a coloured base hypergraph. The latter holds since $col_0(att_{B_0}(e_0)(i)) = [att_{B_0}(e_0)(i)]_{\equiv_P} = att_{H_0}(e_0)(i)$, as required by Definition \ref{def_coloured_base_hg}.
	
	We claim that the relabeling $\chi(B_0)$ belongs to $\EC^b(B_0,e_0,col_0{\restriction}_{V_{B_0}})$ and that $\chi(C) \in \EC(C,col_0{\restriction}_{V_{C}})$ for $C \in \mathcal{C}(Z^\prime) \setminus \{B_0\}$. This is proved by comparing the definition of $\chi$ with the conditions from Definitions \ref{def_EC1} and \ref{def_EC2}. The labeling condition and the colouring condition are checked straightforwardly from the definition of $\chi$. Let us check the cardinality condition. Note that for each $C \in \mathcal{C}(Z^\prime)$ and each $v \in V_{C}$ it holds that $\Cond_{\chi(C)}(v,i_1,i_2) = \Cond_{\chi(Z^\prime)}(v,i_1,i_2)$.
	
	The pair $(e,j)$ belongs to the set $\Cond_{\chi(Z^\prime)}(v,i_1,i_2)$ if and only if $att_{\chi(Z^\prime)}(e)(j) = att_{Z^\prime}(e)(j) = v$, $\chi(e) = (\sigma,f,\tau)$, and $(i_1,i_2,j) \in \tau$. Equivalently, it must be the case that $att_{Z^\prime}(e)(j) = v$ and $(e,j) \in EvS(i_1,i_2)$ (by the definition of $\chi$). Let us consider the evidence path $EvP(i_1,i_2)$ and use the notation (\ref{equation_evidence_path}) from Proposition \ref{prop_evidence_path} for it.
	\begin{itemize}
		\item If $v \ne att_{Z^\prime}(e_0)(i_1)$, $v \ne att_{Z^\prime}(e_0)(i_2)$, then either $v$ does not belong to the evidence path $EvP(i_1,i_2)$ (then $\vert \Cond_{\chi(Z^\prime)}(v,i_1,i_2)\vert = 0$) or it occurs there exactly once. In the latter case, $v$ cannot be the beginning or the end of $EvP(i_1,i_2)$ because this vertex is not equal to $att_{Z^\prime}(e_0)(i_1)$ or to $att_{Z^\prime}(e_0)(i_2)$, which are $v_0$ and $v_l$ of the evidence path respectively. Then $v = v_i$ for some $i \in \{1,\dotsc,l-1\}$. The conditions $att_{Z^\prime}(e)(j) = v_i$ and $(e,j) \in EvS(i_1,i_2)$ are satisfied for exactly two pairs, namely, for $(e_i^-,k_i^-)$ and $(e_i^+,k_i^+)$. This justifies that $\vert \Cond_{\chi(Z^\prime)}(v,i_1,i_2)\vert = 2$.
		\item Let $v = att_{Z^\prime}(e_0)(i_1)$ or $v = att_{Z^\prime}(e_0)(i_2)$ but not both. Then either $v = v_0$ or $v = v_l$, and $v_0 \ne v_l$. If $v=v_0$, then the conditions $att_{Z^\prime}(e)(j) = v_0$ and $(e,j) \in EvS(i_1,i_2)$ are satisfied for exactly one pair, namely, for $(e_0^+,k_0^+)$. Consequently, $\vert \Cond_{\chi(Z^\prime)}(v,i_1,i_2)\vert = 1$. The case where $v = v_l$ is dealt with similarly.
		\item If $v = att_{Z^\prime}(e_0)(i_1) = att_{Z^\prime}(e_0)(i_2)$, then $v_0 = v_l$, which, according to the definition of an evidence path (Proposition \ref{prop_evidence_path}), implies that $EvS(i_1,i_2) = \emptyset$ and, as a consequence, that $\vert \Cond_{\chi(Z^\prime)}(v,i_1,i_2)\vert = 0$.
	\end{itemize}
	
	Finally, we define the vector $u_0$ as follows:
	$$
	u_0 = 
	\beta(\chi(Z^\prime))
	=
	\beta(\chi(B_0))+\sum_{C \in \mathcal{C}(Z^\prime)\setminus \{B_0\}} \beta(\chi(C)).
	$$
	Clearly, $u_0 \in \Mem(FG,H_0)$: indeed, $\chi(B_0) \in S^b$, $\chi(C) \in S$ for $C \ne B_0$. We want to show that $u_0 = \vec{0}$. 
	Let $(h,\overline{h}) \in P$ (recall that this means that the hyperedges $h$ and $\overline{h}$ are fused by the application of $\fr(P)$). Let $\chi(h) = (lab_{Z^\prime}(h),f,\tau)$, $\chi(\overline{h}) = (lab_{Z^\prime}(\overline{h}),\overline{f},\overline{\tau})$. Observe the following.
	\begin{itemize}
		\item $\overline{lab_{Z^\prime}(h)} = lab_{Z^\prime}(\overline{h})$ because $h$ and $\overline{h}$ are fused in $Z^\prime$ by $\fr(P)$.
		\item Since $att_{Z^\prime}(h)(i) \equiv_P att_{Z^\prime}(\overline{h})(i)$ (these vertices are identified by $\fr(P)$), it holds that, if $f(i) = [att_{Z^\prime}(h)(i)]_{\equiv_P} = v_0^j$, then $\overline{f}(i) = [att_{Z^\prime}(\overline{h})(i)]_{\equiv_P} = v_0^j$ as well; consequently, if $f(i) = w$, then $\overline{f}(i)=w$ as well. Therefore, $f = \overline{f}$.
		\item $(i_1,i_2,j)$ belongs to $\tau$ if and only if $(h,j) $ belongs to $EvS(i_1,i_2)$ if and only if $(\overline{h},j)$ belongs to $EvS(i_1,i_2)$ if and only if $(i_1,i_2,j)$ belongs to $\overline{\tau}$. Thus $\tau = \overline{\tau}$. It is true that $(h,j) \in EvS(i_1,i_2)$ if and only if $(\overline{h},j) \in EvS(i_1,i_2)$ because the evidence set is the union of pairs $\{(e_i^+,k_i^+),(e_{i+1}^-,k_{i+1}^-)\}$ where $\{e_i^+,e_{i+1}^-\}$ belongs to $P$.
	\end{itemize}
	
	The above observations imply that $\overline{\chi(h)} = \chi(\overline{h})$, i.e.~that the labels of the hyperedges $h,\overline{h}$ are complementary.
	
	Let us now show that, if $h$ is not fused in the application of $\fr(P)$ (i.e.~if it does not belong to any pair from $P$) and if $h \ne e_0$, then $\chi(h) \in \mathfrak{W}$. Let $\chi(h) = (lab_{Z^\prime}(h),f,\tau)$. The hyperedge $h$ does not disappear after the application of $\fr(P)$; therefore, $h \in E_G$ must not be adjacent to $e_0$. Hence $[att_{Z^\prime}(h)(i)]_{\equiv_P} \ne v_0^j$ for any $j$ and, consequently, $f(i) = w$. This implies that $\tau = \emptyset$ and that $\chi(h) \in \mathfrak{W}$, as desired.
	
	Summing up, we have shown that, for any hyperedge $h$ in $\chi(Z^\prime)$ with a label that does not belong to $\{a_0\} \cup \mathfrak{W}$, it holds that $h \in p \in P$ for some $p = \{h,\overline{h}\}$ and that, moreover, $\overline{\chi(h)} = \chi(\overline{h})$. Therefore, we have proved that, if $A \in \FCol \setminus  \mathfrak{W}$, then the number of $A$-labeled hyperedges in $\chi(Z^\prime)$ equals the number of $\overline{A}$-labeled ones. This proves that $u_0 = \vec{0}$ and completes the proof of the lemma. 
\end{proof}

\begin{lemma}\label{lemma_mem1_2}
	If $\vec{0}$ belongs to $\Mem(FG,H_0)$, then $H_0$ belongs to $L(FG)$. 
\end{lemma}

\begin{proof}
	Let $\vec{0} \in \Mem(FG,H_0)$. Equivalently, 
	\begin{equation}\label{eq_mem1_2_zero}
		\beta(B^2) + \sum_{C^0 \in \mathcal{C}(Z)} \; \sum_{C^1 \in \VC(C^0)} \; \sum_{C^2 \in \EC(C^1)} k(C^2) \beta(C^2) = \vec{0}
	\end{equation}
	for some $B^2 \in \EC^b(B^1)$ where $B^1 = (B^0,\tilde{e}_0,col_{B^1})$ is a coloured base hypergraph and for some non-negative integers $k(C^2) \in \Nat$. 
	Consider the following hypergraphs:
	\begin{eqnarray}
		Z^2 = B^2 + \sum_{C^0 \in \mathcal{C}(Z)} \; \sum_{C^1 \in \VC(C^0)} \; \sum_{C^2 \in \EC(C^1)} k(C^2) C^2; \label{eq_Z2} \\
		Z^0 = B^0 + \sum_{C^0 \in \mathcal{C}(Z)} \; \sum_{C^1 \in \VC(C^0)} \; \sum_{C^2 \in \EC(C^1)} k(C^2) C^0. \label{eq_Z0}
	\end{eqnarray}
	In plain words, $Z^2$ is composed of the hypergraph $B^2 \in S^b$, which corresponds to a coloured base hypergraph with the distinguished hyperedge $\tilde{e}_0$, and of hypergraphs from $S$, i.e.~of all possible hypergraphs obtained from those in $\mathcal{C}(Z)$ by adding a vertex colouring and then encoding it as well as information about evidence paths in hyperedges. Namely, there are $k(C^2)$ copies of $C^2$, which is a hyperedge encoding of a vertex colouring of $C^0 \in \mathcal{C}(Z)$, in $Z^2$. The hypergraph $Z^0$ is obtained from $Z^2$ by forgetting the information about colourings and evidence paths: it holds that $lab_{Z^2}(\tilde{e}_0) = lab_{Z^0}(\tilde{e}_0) = a_0$ and that, for each $e \ne \tilde{e}_0$, $lab_{Z^2}(e) = (lab_{Z^0}(e),f_e,\tau_e)$ for some $f_e$, $\tau_e$.
	
	Let us introduce some new notation. For any $C \in \mathcal{C}(Z^2)$, $C \ne B^2$, it holds that $C \in \EC(C^1)$ for some coloured hypergraph $C^1 = (C^0,col_{C^1})$ where $C^0 \in \mathcal{C}(Z)$. Let us denote such $C^1$ by $g^1(C)$ and such $C^0$ by $g^0(C)$. Then, we can rewrite (\ref{eq_Z0}) as follows:
	$$
	Z^0 = B^0 + \sum_{C^0 \in \mathcal{C}(Z)} \; \sum_{C^1 \in \VC(C^0)} \; \sum_{C^2 \in \EC(C^1)} k(C^2) g^0(C^2).
	$$
	
	Recall that $col_{g^1(C)}$ denotes the colouring function of $g^1(C)$. Let us define a coloured hypergraph $Z^1 = (Z^0,col_{Z^1})$ by combining all vertex colourings of its connected components which are encoded in $Z^2$. Formally, let $col_{Z^1} {\restriction}_{V_{B^0}} = col_{B^1}$ and $col_{Z^1} {\restriction}_{V_{g^0(C)}} = col_{g^1(C)}$ for $C \in \mathcal{C}(Z^0)$, $C \ne B^0$. 
	
	The equation (\ref{eq_mem1_2_zero}) says that $\beta(Z^2) = \vec{0}$ and thus $\#_A(Z^2) = \#_{\overline{A}}(Z^2)$ for $A \in \FCol \setminus \mathfrak{W}$. Let $\{e_1^A,\dotsc,e_{l(A)}^A\}$ be the set of all hyperedges from $E_{Z^2}$ with the label $A$; $l(A)$ is the number of such hyperedges. Then $l(A) = l(\overline{A})$ for $A \in \FCol \setminus \mathfrak{W}$. We define $P^0$ as follows:
	$$
	P^0 = \bigcup\limits_{A \in \FCol \setminus \mathfrak{W}} \left\{\left\{e_1^A,e_1^{\overline{A}}\right\},\dotsc,\left\{e_{l(A)}^A,e_{l(\overline{A})}^{\overline{A}}\right\}\right\}.
	$$
	Consider the following derivation:
	$$
	Z \Rightarrow Z^0 \underset{\fr(P^0)}{\Rightarrow} \widetilde{H}.
	$$
	Its first step is valid since $Z^0$ is the sum of connected components of $Z$. The second step is also valid because $e_i^A$ and $e_i^{\overline{A}}$ have complementary labels in $Z^2$, which implies that these labels are also complementary in $Z^0$. Formally,
	\begin{equation}\label{eq_overline_reasonings}
		\overline{lab_{Z^2}(e_i^A)} = (\overline{lab_{Z^0}(e_i^A)},f_{e_i^A},\tau_{e_i^A}) = lab_{Z^2}(e_i^{\overline{A}}) = (lab_{Z^0}(e_i^{\overline{A}}),f_{e_i^{\overline{A}}},\tau_{e_i^{\overline{A}}}),
	\end{equation}
	hence $\overline{lab_{Z^0}(e_i^A)} =  lab_{Z^0}(e_i^{\overline{A}})$.
	
	Our main goal now is to prove that $\widetilde{H}$ contains a connected component isomorphic to $H_0$; this would prove the lemma. We claim that the hyperedge $\tilde{e}_0$ with its attachment vertices forms such a subhypergraph, i.e.~that the following two facts are true:
	\begin{enumerate}
		\item\label{fact1_lemma} for any $i \in [N]$, there is no hyperedge attached to $att_{\widetilde{H}}(\tilde{e}_0)(i)$ other than $\tilde{e}_0$ (recall that $N$ is the type of $e_0$ in the definition of the membership problem);
		\item\label{fact2_lemma} $att_{\widetilde{H}}(\tilde{e}_0)(i_1) = att_{\widetilde{H}}(\tilde{e}_0)(i_2)$ holds if and only if $att_{H_0}(e_0)(i_1) = att_{H_0}(e_0)(i_2)$ holds.
	\end{enumerate}
	Let $v^\prime \equiv_{P^0} v^{\prime\prime}$ be two vertices from $V_{Z^0}$ that become equal after the application of $\fr(P^0)$. Let us prove that then $col_{Z^1}(v^\prime) = col_{Z^1}(v^{\prime\prime})$. If $v^\prime = v^{\prime\prime}$, then there is nothing to prove. Otherwise, consider the evidence $\fr(P^0)$-path from $v^\prime$ to $v^{\prime\prime}$. Assume that it is of the form (\ref{equation_evidence_path}). Let $lab_{Z^2}(e_i^\pm) = (lab_{Z^0}(e_i^\pm), f^\pm_i,\tau^\pm_i)$ where $\pm \in \{+,-\}$. We know that:
	\begin{itemize}
		\item $f_i^+(k_i^+) = col_{Z^1}(att_{Z^0}(e_i^+)(k_i^+)) = col_{Z^1}(v_i)$ for $i=0,\dotsc,l-1$;
		
		{\color{comment}
			The first equality follows from the colouring condition, and the second one is simply a reminder that $v_i = att_{Z^0}(e_i^+)(k_i^+)$ where $v_i$ is the $i$-th vertex in the evidence path, as in (\ref{equation_evidence_path}).
		}
		\item $f_i^-(k_i^-) = col_{Z^1}(att_{Z^0}(e_i^-)(k_i^-)) = col_{Z^1}(v_i)$ for $i=1,\dotsc,l$;
		\item $\{e_i^+,e_{i+1}^-\} \in P^0$ and $k_i^+=k_{i+1}^-$ for $i=0,\dotsc,l-1$ (this is one of the conditions in the definition of an evidence path, see Proposition \ref{prop_evidence_path}). Consequently, $\overline{lab_{Z^2}(e_i^+)} = lab_{Z^2}(e_{i+1}^-)$ and hence $f_i^+ = f_{i+1}^-$ (see the argument in (\ref{eq_overline_reasonings})). Then, since $k_i^+ = k_{i+1}^-$, it holds that $f_i^+(k_i^+) = f_{i+1}^-(k_{i+1}^-)$.
	\end{itemize}
	A consequence of the above observations is that $f_i^-(k_i^-) = f_i^+(k_i^+)$ for $i=1,\ldots,l-1$. Finally, we derive the desired equality:
	\begin{equation*}
		col_{Z^1}(v^\prime) = col_{Z^1}(v_0) = f_0^+(k_0^+) = f_1^-(k_1^-) = f_1^+(k_1^+) = \dotsc = f_{l-1}^-(k_{l-1}^-) = f_{l-1}^+(k_{l-1}^+) = f_l^-(k_l^-) = col_{Z^1}(v^{\prime\prime}).
	\end{equation*}
	
	Let us prove the statement \ref{fact1_lemma}. The proof is ex falso. Let $att_{\widetilde{H}}(e^\prime)(i^\prime) = att_{\widetilde{H}}(\tilde{e}_0)(i)$ for some hyperedge $e^\prime \in E_{\widetilde{H}}$ not equal to $\tilde{e}_0$ and for some $i^\prime$. Equivalently, $att_{Z^0}(e^\prime)(i^\prime) \equiv_{P^0} att_{Z^0}(\tilde{e}_0)(i)$; this implies that
	$col_{Z^1}(att_{Z^0}(e^\prime)(i^\prime)) = col_{Z^1}(att_{Z^0}(\tilde{e}_0)(i))$, as we proved earlier.
	The colour $col_{Z^1}(att_{Z^0}(\tilde{e}_0)(i))$ equals $att_{H_0}(e_0)(i)$ according to Definition \ref{def_coloured_base_hg} (because $(B_0,\tilde{e}_0,col_{B^1})$ is a coloured base hypergraph). On the other hand, notice that $lab_{Z^2}(e^\prime)$ must belong to $\mathfrak{W}$, because otherwise $e^\prime$ must be in one of the pairs of $P^0$ and therefore it does not remain after the application of $\fr(P^0)$. Then, $lab_{Z^2}(e^\prime) = (lab_{Z^0}(e^\prime),f^\prime,\emptyset)$ where $f^\prime(j) = w$ for all $j$. In particular, $f^\prime(i^\prime) = col_{Z^1}(att_{Z^0}(e^\prime)(i^\prime)) = w \ne att_{H_0}(e_0)(i)$, which is a contradiction.
	
	Similarly, we prove the ``only if'' direction of the statement \ref{fact2_lemma}. If $att_{\widetilde{H}}(\tilde{e}_0)(i_1) = att_{\widetilde{H}}(\tilde{e}_0)(i_2)$, which is equivalent to the fact that $att_{Z^0}(\tilde{e}_0)(i_1) \equiv_{P^0} att_{Z^0}(\tilde{e}_0)(i_2)$, then 
	$$
		col_{B^1}(att_{B^0}(\tilde{e}_0)(i_1)) =
		col_{Z^1}(att_{Z^0}(\tilde{e}_0)(i_1)) =
		col_{Z^1}(att_{Z^0}(\tilde{e}_0)(i_2)) = 
		col_{B^1}(att_{B^0}(\tilde{e}_0)(i_2)) .
	$$ 
	According to Definition \ref{def_coloured_base_hg} applied to $B^1 = (B_0,\tilde{e}_0,col_{B^1})$, this is equivalent to the desired fact that $att_{H_0}(e_0)(i_1) = att_{H_0}(e_0)(i_2)$.
	
	It remains to prove the ``if'' direction of the statement \ref{fact1_lemma}. That is, given $att_{H_0}(e_0)(i_1) = att_{H_0}(e_0)(i_2)$, we aim to show that $att_{Z^0}(\tilde{e}_0)(i_1) \equiv_{P^0} att_{Z^0}(\tilde{e}_0)(i_2)$. To do this, we shall construct an evidence $\fr(P^0)$-path from $att_{Z^0}(\tilde{e}_0)(i_1)$ to $att_{Z^0}(\tilde{e}_0)(i_2)$. If $att_{Z^0}(\tilde{e}_0)(i_1) = att_{Z^0}(\tilde{e}_0)(i_2)$, then there is nothing to do. Otherwise, without loss of generality, assume that $i_1 < i_2$. Then, $(i_1,i_2)$ is a connection pair. We construct an evidence path inductively as follows.
	\begin{enumerate}
		\item Let $v_0 = att_{Z^0}(\tilde{e}_0)(i_1)$.
		\item\label{step2_inductive} Assume that we have constructed a sequence of the form 
		\begin{equation}\label{equation_evidence_path_construction}
			v_0, (e_0^+,k_0^+), (e_1^-,k_1^-), v_1, (e_1^+,k_1^+), (e_2^-,k_2^-), \dotsc, (e_{s-1}^+,k_{s-1}^+), (e_s^-,k_s^-), v_s
		\end{equation}
		for some $s \ge 0$ that satisfies the following \emph{evidence path requirements}:
		\begin{itemize}
			\item $v_i \in V_{Z^0}$; \qquad $e_i^+,e_i^- \in E_{Z^0}$; \qquad $v_0,\dotsc,v_s$ are distinct;
			\item $\{e_i^+,e_{i+1}^-\} \in P^0$ and $k_i^+ = k_{i+1}^-$;
			\item $att_{Z^0}(e_i^+)(k_i^+) = v_i$ for $i=0,\dotsc,s-1$; $att_{Z^0}(e_i^-)(k_i^-) = v_i$ for $i=1,\dotsc,s$;
			\item $(e_i^+,k_i^+) \in \Cond_{Z^2}(v_i,i_1,i_2)$ for $i=0,\dotsc,s-1$ and \\ $(e_i^-,k_i^-) \in \Cond_{Z^2}(v_i,i_1,i_2)$ for $i=1,\dotsc,s$.
		\end{itemize}
		Note that (\ref{equation_evidence_path_construction}) is an evidence path, and $v_0 \equiv_{P^0} v_s$; as a consequence, $col_{Z^1}(v_0) = col_{Z^1}(v_s)$ (recall that $col_{Z^1}(v_0) = att_{H_0}(e_0)(i_1)$).
		
		The hypergraph $Z^2$ consists of several connected components, one of which is $B^2 \in \EC^b(B^1)$, while for any other one $C^2$ it holds that $C^2 \in \EC(g^1(C^2))$. The cardinality conditions from Definitions \ref{def_EC1} and \ref{def_EC2} imply that 
		$
			\vert \Cond_{Z^2}(att_{Z^0}(\tilde{e}_0)(i_1),i_1,i_2) \vert = \vert \Cond_{Z^2}(att_{Z^0}(\tilde{e}_0)(i_2),i_1,i_2) \vert = 1
		$,
		while $\vert \Cond_{Z^2}(v,i_1,i_2) \vert$ equals 0 or 2 for any vertex $v$ other than $att_{Z^0}(\tilde{e}_0)(i_1)$ or $att_{Z^0}(\tilde{e}_0)(i_2)$.
		Let us focus on the vertex $v_s$. The first possible case is that $\vert \Cond_{Z^2}(v_s,i_1,i_2) \vert = 1$; then either $v_s = att_{Z^0}(\tilde{e}_0)(i_1) = v_0$ or $v_s = att_{Z^0}(\tilde{e}_0)(i_2)$. The second possibility is that $\vert \Cond_{Z^2}(v_s,i_1,i_2) \vert$ is either 0 or 2. Let us consider each of these cases in detail.
		\begin{enumerate}
			\item If $v_s = v_0$, then $s = 0$ (all vertices in (\ref{equation_evidence_path_construction}) are distinct). Let us denote by $(e_s^+,k_s^+) = (e_0^+,k_0^+)$ the only element of $\Cond_{Z^2}(v_0,i_1,i_2)$. Then, proceed to step \ref{item_step3_mem1_2}.
			\item\label{item_terminate} If $v_s = att_{Z^0}(\tilde{e}_0)(i_2)$, then let us terminate our inductive procedure and return (\ref{equation_evidence_path_construction}) as its output. Clearly, this is a desired evidence path.
			\item Otherwise, $s > 0$, $v_s$ equals neither $att_{Z^0}(\tilde{e}_0)(i_1)$ nor $att_{Z^0}(\tilde{e}_0)(i_2)$, and $\vert \Cond_{Z^2}(v_s,i_1,i_2) \vert$ is either 0 or 2. Since $(e_s^-,k_s^-) \in \Cond_{Z^2}(v_s,i_1,i_2)$, it must be the case that $\vert \Cond_{Z^2}(v_s,i_1,i_2) \vert = 2$. Let $(e_s^+,k_s^+) \in \Cond_{Z^2}(v_s,i_1,i_2)$ be the second element of this set. Then, proceed to step \ref{item_step3_mem1_2}.
		\end{enumerate}
		
		\item\label{item_step3_mem1_2} At step \ref{step2_inductive}, we defined $(e_s^+,k_S^+)$. Let $lab_{Z^2}(e_s^+) = (lab_{Z^0}(e_s^+),f_s^+,\tau_s^+)$. We know that:
		\begin{itemize}
			\item $att_{Z^2}(e_s^+)(k_s^+) = v_s$ (this follows from the definition of a conduction set and of an incidence set);
			\item $f_s^+(k_s^+) = col_{Z^1}(att_{Z^2}(e_s^+)(k_s^+)) = col_{Z^1}(v_s) = col_{Z^1}(v_0) = att_{H^0}(e_0)(i_1) \ne w$ (the first equality follows from the colouring condition);
			\item $(i_1,i_2,k_s^+) \in \tau_s^+$.  
		\end{itemize}
		Since $f_s^+$ is not identical to $w$ and $\tau_s^+ \ne \emptyset$, it holds that $lab_{Z^2}(e_s^+) \in \FCol \setminus \mathfrak{W}$, hence $e_s^+$ must belong to one of the pairs in $P^0$. Let $\{e_s^+,e\} \in P^0$ for some hyperedge $e$ (note that such $e$ is defined uniquely). Let us define $e_{s+1}^- = e$. Let also $k_{s+1}^- = k_s^+$. Let $v_{s+1} = att_{Z^2}(e_{s+1}^-)(k_{s+1}^-)$.
		
		We claim that $(e_{s+1}^-,k_{s+1}^-) \in \Cond_{Z^2}(v_{s+1},i_1,i_2)$. First, note that $(e_{s+1}^-,k_{s+1}^-) \in \Inc_{Z^2}(v_{s+1})$. Secondly, observe that $lab_{Z^2}(e_{s+1}^-) = \overline{lab_{Z^2}(e_{s}^+)} = (\overline{lab_{Z^0}(e_s^+)},f_s^+,\tau_s^+)$, thus $(i_1,i_2,k_{s+1}^-) = (i_1,i_2,k_{s}^+) \in \tau_s^+$. This proves our claim.
		
		Now, we want to prove that $v_{s+1} \ne v_j$ for any $j=0,\dotsc,s$. The proof is ex falso. Assume that $v_{s+1} = v_j$ for some $j=0,\dotsc,s$. Then, $(e_{s+1}^-,k_{s+1}^-) \in \Cond_{Z^2}(v_{j},i_1,i_2)$. There are several cases. 
		\begin{enumerate}
			\item\label{item_reasonings_a} $j = 0$, $s>1$. Since $\Cond_{Z^2}(v_{0},i_1,i_2)$ has the cardinality $1$, $(e_{s+1}^-,k_{s+1}^-) = (e_0^+,k_0^+)$. Then, $(e_s^+,k_s^+) = (e_1^-,k_1^-)$, because $e_s^+$ and $e_1^-$ are the hyperedges fused with $e_{s+1}^-$ and $e_0^+$ resp. Consequently, $v_1 = att_{Z^2}(e_1^-)(k_1^-) = att_{Z^2}(e_s^+)(k_s^+) = v_s$, which contradicts the fact that $v_1 \ne v_s$.
			\item\label{item_reasonings_b} $j = 0$, $s = 1$. Similarly, $(e_{s+1}^-,k_{s+1}^-) = (e_2^-,k_2^-) = (e_0^+,k_0^+)$. Then $(e_1^+,k_1^+) = (e_1^-,k_1^-)$. However, $(e_1^+,k_1^+)$ was chosen as the second element of the set $\Cond_{Z^2}(v_1,i_1,i_2)$ apart from $(e_1^-,k_1^-)$. This is a contradiction.
			\item\label{item_reasonings_c} $j=0$, $s=0$. Then, $(e_{s+1}^-,k_{s+1}^-) = (e_1^-,k_1^-) = (e_0^+,k_0^+)$. This contradicts the fact that $lab_{Z^2}(e_{1}^-) = \overline{lab_{Z^2}(e_{0}^+)}$.
			\item $j > 0$. If $v_{s+1} = v_j$, then $(e_{s+1}^-,k_{s+1}^-)$, which belongs to $\Cond_{Z^2}(v_{j},i_1,i_2)$, coincides with one of the pairs $(e_j^-,k_j^-)$, $(e_j^+,k_j^+)$ according to the cardinality condition. 
			\begin{enumerate}
				\item If $(e_{s+1}^-,k_{s+1}^-) = (e_j^-,k_j^-)$, then $(e_{s}^+,k_{s}^+) = (e_{j-1}^+,k_{j-1}^+)$, hence $v_s = att_{Z^2}(e_s^+)(k_s^+) = att_{Z^2}(e_{j-1}^+)(k_{j-1}^+) = v_{j-1}$, which contradicts the fact that $v_{j-1}$ and $v_s$ are different.
				\item If $(e_{s+1}^-,k_{s+1}^-) = (e_j^+,k_j^+)$ and $j < s - 1$, then $(e_{s}^+,k_{s}^+) = (e_{j+1}^-,k_{j+1}^-)$, hence \\ $v_s = att_{Z^2}(e_s^+)(k_s^+) = att_{Z^2}(e_{j+1}^-)(k_{j+1}^-) = v_{j+1}$, which contradicts the fact that $v_{j+1}$ and $v_s$ are different. This is similar to Case \ref{item_reasonings_a}.
				\item If $(e_{s+1}^-,k_{s+1}^-) = (e_j^+,k_j^+)$ and $j = s-1$, then the reasonings are the same to those for Case \ref{item_reasonings_b}.
				\item If $(e_{s+1}^-,k_{s+1}^-) = (e_j^+,k_j^+)$ and $j = s$, then the reasonings are the same to those for Case \ref{item_reasonings_c}.
			\end{enumerate}
		\end{enumerate}
		
		Thus we have proved that $v_{s+1}$ cannot be equal to any of $v_j$ for $j \le s$. Summing up, the following extension of (\ref{equation_evidence_path_construction}) also satisfies the evidence path requirements:
		\begin{equation}\label{equation_evidence_path_construction_step}
			v_0, (e_0^+,k_0^+), (e_1^-,k_1^-), v_1, \dotsc, (e_{s-1}^+,k_{s-1}^+), (e_s^-,k_s^-), v_s, (e_{s}^+,k_{s}^+), (e_{s+1}^-,k_{s+1}^-), v_{s+1}
		\end{equation}
		Then, we proceed to step \ref{step2_inductive} of our inductive procedure taking (\ref{equation_evidence_path_construction_step}) as a new sequence.	
	\end{enumerate}
	The described procedure is deterministic. Moreover, it must terminate because $v_0,\dotsc,v_s$ in (\ref{equation_evidence_path_construction}) are distinct, hence their number does not exceed the total number of vertices in $V_{Z^0}$. If one inspects the procedure carefully, they obvserve that the only way for it to terminate is Case \ref{item_terminate} where $v_s = att_{Z^0}(\tilde{e}_0)(i_2)$. As we noticed, $v_0 \equiv_{P^0} v_s$; thus $att_{Z^0}(\tilde{e}_0)(i_1) \equiv_{P^0} att_{Z^0}(\tilde{e}_0)(i_2)$, which concludes the proofs both of the ``if'' direction of statement \ref{fact2_lemma} and of the lemma.
\end{proof}


Clearly, the above lemmas imply Theorem \ref{th_mem1} immediately.
\begin{proof}[Proof of Theorem \ref{th_mem1}]
	Given a fusion grammar $FG$ without markers and connectors and a hypergraph $H_0$ satisfying the requirements of MEM-1, construct the set $\Mem(FG,H_0)$. Checking whether $\vec{0} \in \Mem(FG,H_0)$ is decidable (Remark \ref{remark_decidability_linear_integer_programming}). This is equivalent to checking that $H_0 \in L(FG)$ according to Lemmas \ref{lemma_mem1_1} and \ref{lemma_mem1_2}. 
\end{proof}

\subsection{Membership Problem For Arbitrary Hypergraphs}\label{ssec_results_mem2}

Now, we shall show how to reduce the general membership problem MEM to the case of one-hyperedge hypergraphs MEM-1. This is possible because the most important part in checking whether a hypergraph is generated by a fusion grammar is to be able to check which vertices are fused and which are not. Let us extend MEM-1 to MEM-2, which is almost the desired membership problem MEM.

\begin{problem}[MEM-2]
	\leavevmode
	\\
	\textit{Input.} 
	\begin{enumerate}
		\item A fusion grammar without markers and connectors $FG = (Z,F,T)$; 
		\item A connected hypergraph $H_1 \in \mathcal{H}(T)$, which is not a single vertex with no hyperedges.
	\end{enumerate}
	\textit{Problem MEM-2.} Does $H_1$ belong to $L(FG)$?
\end{problem}

\begin{theorem}\label{th_mem2}
	MEM-2 is decidable.
\end{theorem}

The following construction will be used to reduce MEM-2 to MEM-1. 

\begin{construction}\label{construction_mem2}
	Let $E_{H_1} = \{e_1,\dotsc,e_k\}$ and let $N = \type(e_1) + \dotsc + \type(e_k)$; let $\vert V_{H_1} \vert = n$. Let us define a fusion grammar $FG_1 = (Z_1,F_1,T_1)$ where
	\begin{itemize}
		\item $T_1 = \{a_0\}$ where $a_0$ is a new label such that $\type(a_0) = N$;
		\item $F_1 = F \uplus T$;
		\item $Z_1 = Z + I$ where $I$ is defined below.
		\begin{itemize}
			\item $V_I = \{(i,j) \mid 1 \le i \le k, 1 \le j \le \type(e_i)\}$. Let us fix the lexicographic order $\prec$ on $V_I$ and let us use the notation $V_I = \{\hat{1},\dotsc,\hat{N}\}$ for vertices where $\hat{1} \prec \dotsc \prec \hat{N}$. The other way round, if $\hat{l} = (i,j)$, then let $\pi_1(l) = i$, $\pi_2(l) = j$.
			\item $E_I = E_{H_1} \sqcup \{e_0\}$ where $e_0$ is a new hyperedge.
			\item $att_I(e_i)(j) = (i,j)$ for $(i,j) \in V_I$; $att_I(e_0)(l) = \hat{l}$.
			\item $lab_I(e_i) = \overline{lab_{H_1}(e_i)}$ for $i=1,\dotsc,k$; $lab_I(e_0) = a_0$.
		\end{itemize}
	\end{itemize}
	Finally, consider the hypergraph $H_0$ such that
	\begin{itemize}
		\item $V_{H_0} = V_{H_1}$;
		\item $E_{H_0} = \{e_0\}$;
		\item $att_{H_0}(e_0)(l) = att_{H_1}(e_{\pi_1(l)})(\pi_2(l))$ for $l=1,\ldots,N$;
		\item $lab_{H_0}(e_0) = a_0$.
	\end{itemize}
\end{construction}

\begin{example}
	Assume that we are given a fusion grammar $FG = (Z,\{A,B,C\},\{a,b,c\})$ and we want to know if $H_1 = \vcenter{\hbox{{\tikz[baseline=.1ex]{
					\node[hyperedge] (EL) at (1.4,0) {$c$};
					\node[vertex] (VO) at (2,-0.3) {};
					\node[vertex] (VL) at (2,0.3) {};
					\node[vertex] (VL1) at (0,0) {};
					\node[vertex] (VL3) at (0.8,0) {};
					\draw[-latex, thick] (VL1) -- node[above] {$b$} (VL3);
					\draw[-latex, thick] (VL1) to[out=210,in=150,looseness=20] node[left] {$a$} (VL1);
					\draw[-] (EL) -- (VL);
					\draw[-] (EL) -- (VO);
					\draw[-] (EL) -- node[above] {\scriptsize 3} (VL3);
					\node at (1.78,0.32) {\scriptsize 1};
					\node at (1.78,-0.07) {\scriptsize 2};
	}}}}
	$ belongs to $L(FG)$. Construction \ref{construction_mem2} turns the labels $a,b,c$ into fusion ones and introduces a new terminal label $a_0$. The following hypergraphs are constructed then.
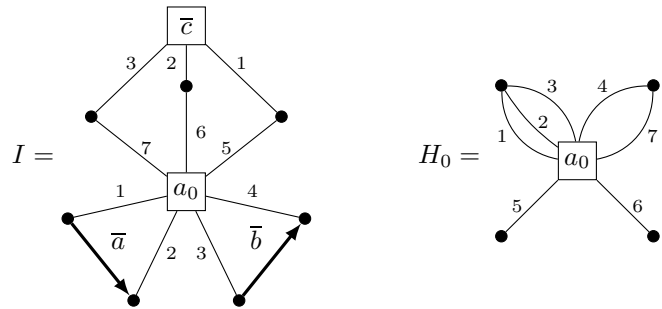
	$$
	I = 
	\vcenter{\hbox{{\tikz[baseline=.1ex]{
					\node[hyperedge] (E0) at (0/1.7,0/1.7) {$a_0$};
					\node[vertex] (V1) at (0/1.7,1.4/1.7) {};
					\node[vertex] (V2) at (-1.251/1.7,0.996/1.7) {};
					\node[vertex] (V3) at (-1.56/1.7,-0.356/1.7) {};
					\node[vertex] (V4) at (-0.694/1.7,-1.442/1.7) {};
					\node[vertex] (V5) at (0.694/1.7,-1.442/1.7) {};
					\node[vertex] (V6) at (1.56/1.7,-0.356/1.7) {};
					\node[vertex] (V7) at (1.251/1.7,0.996/1.7) {};
					\node[hyperedge] (Ec) at (1.251/1.7,2.2/1.7) {$\overline{c}$};
					\draw[-latex, thick] (V3) -- node[below left] {$\overline{a}$} (V4);
					\draw[-latex, thick] (V5) -- node[below right] {$\overline{b}$} (V6);
					\draw[-] (Ec) -- node[right] {\scriptsize 1} (V7);
					\draw[-] (Ec) -- node[above] {\scriptsize 2} (V1);
					\draw[-] (Ec) to[bend right = 40] node[left] {\scriptsize 3} (V2);
					\draw[-] (E0) -- node[right] {\scriptsize 6} (V1);
					\draw[-] (E0) -- node[right] {\scriptsize 7} (V2);
					\draw[-] (E0) -- node[above] {\scriptsize 1} (V3);
					\draw[-] (E0) -- node[left] {\scriptsize 2} (V4);
					\draw[-] (E0) -- node[left] {\scriptsize 3} (V5);
					\draw[-] (E0) -- node[below] {\scriptsize 4} (V6);
					\draw[-] (E0) -- node[below] {\scriptsize 5} (V7);
	}}}}
	\qquad
	\qquad
	H_0 = 
	\vcenter{\hbox{{\tikz[baseline=.1ex]{
					\node[hyperedge] (E0) at (-1/1.2,-1/1.2) {$a_0$};
					\node[vertex] (N1) at (0/1.2,0/1.2) {};
					\node[vertex] (N2) at (-2/1.2,0/1.2) {};
					\node[vertex] (N3) at (-2/1.2,-2/1.2) {};
					\node[vertex] (N4) at (0/1.2,-2/1.2) {};
					\draw[-] (E0) -- node[right] {\scriptsize 6} (N4);
					\draw[-] (E0) to[bend right = 40] node[right] {\scriptsize 7} (N1);
					\draw[-] (E0) to[bend left = 40] node[left] {\scriptsize 1} (N2);
					\draw[-] (E0) to[bend left = 10] node[right] {\scriptsize 2} (N2);
					\draw[-] (E0) to[bend right = 40] node[above] {\scriptsize 3} (N2);
					\draw[-] (E0) to[bend left = 40] node[above] {\scriptsize 4} (N1);
					\draw[-] (E0) -- node[left] {\scriptsize 5} (N3);
	}}}}
	$$
	Finally, $Z_1 = Z + I$. We are going to prove that $H_0$ is generated by $(Z_1,\{A,B,C,a,b,c\},\{a_0\})$ if and only if $H_1$ is generated by the initial fusion grammar $FG$. Informally, we simply make a hypergraph $I$ with one big $a_0$-labeled hyperedge, to which we attach hyperedges of the target hypergraph $H_1$. It is not hard to see that, if we fuse $H_1$ and $I$, then the result is exactly $H_0$.
\end{example}

\begin{remark}\label{remark_eI}
	According to Construction \ref{construction_mem2}, $E_I = E_{H_1} \sqcup \{e_0\}$. Hereinafter, we shall consider the disjoint union $H_1 + I$ where the hyperedges of $H_1$ and $I$ shall be made disjoint. Let us make them disjoint in advance by saying that $E_I = \{e^I \mid e \in E_{H_1}\} \sqcup \{e_0\}$ where $e^I$ are fresh copies of hyperedges from $E_{H_1}$.
\end{remark}

\begin{lemma}\label{lemma_fusion_isomorphism}
	Let $H \in \mathcal{H}(T)$. Assume that $H+I \underset{\fr(P)}{\Rightarrow} Y$ for some $P$; let us define the function $\varphi:E_{H_1} \to E_H$ as follows: $\varphi(e) = e^\prime$ for $\{e^I,e^\prime\} \in P$. Then, $Y$ is isomorphic to $H_0$ if and only if $\varphi$ induces an isomorphism of $H$ and $H_1$ in the sense of Remark \ref{remark_isomorphism_isolated}. 
\end{lemma}
\begin{proof}
	We start the proof with the following observation. The hypergraph $I$ is designed in such a way that all the attachment vertices of all hyperedges $e_1^I,\dotsc,e_k^I$ are distinct. Assume that two vertices $(i_1,j_1)$ and $(i_2,j_2)$ from $V_{I}$ are identified after the application of $\fr(P)$ and fix an evidence path between them; we shall use the same notation as in (\ref{equation_evidence_path}) from Proposition \ref{prop_evidence_path}. The vertex $v_0$ equals $(i_1,j_1)$ and $v_l$ equals $(i_2,j_2)$. Note that $H$ does not contain two complementary hyperedges as well as $I$. Therefore, the vertex $v_1$ must be from $V_H$ and the vertex $v_2$ must be from $V_I$. If $l > 2$, then $v_3$ must again be from $V_H$. However, since $v_2 = att_I(e_i^I)(j)$ holds for exactly one pair of the form $(e_i^I,j)$, it must be the case that $(e_2^-,k_2^-) = (e_2^+,k_2^+)$. As a consequence, $v_1 = v_3 = att_H(\varphi(e_i^I))(j)$. This contradicts the definition of an evidence path since $v_1$ and $v_3$ must be distinct. Thus $l = 2$, and ergo $(i_1,j_1) \equiv_P (i_2,j_2)$ holds if and only if $att_{H}(\varphi(e_{i_1}))(j_1) = att_{H}(\varphi(e_{i_2}))(j_2)$. (This follows from the definition of $\varphi$.)
	
	Let $H+I \underset{\fr(P)}{\Rightarrow} H_0$. First, note that $H$ does not contain isolated vertices because otherwise $H_0$ would, which is not the case. Then, let us show that $\varphi$ indeed induces an isomorphism of $H$ and $H_1$. It is not hard to see that $\varphi$ is a bijection because, after the application of $\fr(P)$, all the hyperedges in $E_{H+I}$ except for $e_0 \in E_I$ disappear. Furthermore, no two hyperedges within $H$ or within $I$ can be fused with each other because all labels in $H$ belong to $T$ while those in $I$ belong to $\overline{T} \cup \{a_0\}$, which is disjoint with $T$. Hence, for any pair $\{e^\prime,e^{\prime\prime}\} \in P$ one of the hyperedges is from $E_H$ while the other one is from $E_I \setminus \{e_0\} = E_{H_1}$. This proves that $\varphi$ is a bijection.
	
	Now, let us check the conditions from Remark \ref{remark_isomorphism_isolated}.
	\begin{enumerate}
		\item As we showed, $(i_1,j_1) \equiv_P (i_2,j_2)$ if and only if $att_{H}(\varphi(e_{i_1}))(j_1) = att_{H}(\varphi(e_{i_2}))(j_2)$. On the other hand, the definition of $H_0$ implies that $(i_1,j_1) \equiv_P (i_2,j_2)$ holds if and only if $att_{H_1}(e_{i_1})(j_1) = att_{H_1}(e_{i_2})(j_2)$. Thus the first condition is checked.
		\item It is the case that $lab_H(\varphi(e)) = \overline{lab_I(e^I)} = lab_{H_1}(e)$. The first equality holds because $\{e^I,\varphi(e)\} \in P$, while the second one holds according to the definition of $I$.
	\end{enumerate}
	Thus we have proved that $\varphi$ induces an isomorphism of $H$ and $H_1$.
	
	Conversely, let us prove the ``if'' direction of the lemma. Let $\varphi$ induce an isomorphism of $H$ and $H_1$. The rule $\fr(P)$ fuses $e^I \in E_I$ with $\varphi(e) \in E_H$. We know that $(i_1,j_1) \equiv_P (i_2,j_2)$ if and only if $att_{H}(\varphi(e_{i_1}))(j_1) = att_{H}(\varphi(e_{i_2}))(j_2)$. Since $\varphi$ induces an isomorphism, the latter is equivalent to $att_{H}(e_{i_1})(j_1) = att_{H}(e_{i_2})(j_2)$. Therefore, $(i_1,j_1) \equiv_P (i_2,j_2)$ if and only if $att_{H}(e_{i_1})(j_1) = att_{H}(e_{i_2})(j_2)$. This implies that the resulting hypergraph is isomorphic to $H_0$, as desired (this follows from the definition of $H_0$).
\end{proof}

\begin{lemma}\label{lemma_mem2_main}
	$H_1 \in L(FG)$ if and only if $H_0 \in L(FG_1)$. 
\end{lemma}
\begin{proof}
	Let $H_1 \in L(FG)$. Then, there is a derivation of the following form:
	$$
	Z \Rightarrow m\cdot Z \underset{\fr(P)}{\Rightarrow} H_1 + G.
	$$
	Let us transform it into the following one:
	$$
	Z_1 \Rightarrow m\cdot Z + I \underset{\fr(P)}{\Rightarrow} H_1 + I + G \underset{\fr(P^\prime)}{\Rightarrow} H_1^\prime + G.
	$$
	The set $P^\prime$ that occurs at its last step equals $\{\{e,e^I\} \mid e \in E_{H_1}\}$; in other words, we fuse hyperedges of $H_1$ and $I$ with the same names. This is a valid derivation in $FG_1$ according to the definition of $I$. Lemma \ref{lemma_fusion_isomorphism} implies that the result $H_1^\prime$ of this fusion is isomorphic to $H_0$ (because $\varphi(e) = e$ induces the trivial isomorphism of $H_1$ and $H_1$). Thus $H_0 \in L(FG_1)$.
	
	The other way round, let $H_0 \in L(FG_1)$. Then, there is a derivation of the form:
	$$
	Z_1 \Rightarrow m \cdot Z_1 \underset{\fr(P)}{\Rightarrow} H_0 + G
	$$
	for some $m$, $P$ and $G$. Clearly, $c = m(I) \ge 1$; consequently, $m \cdot Z_1 = m^\prime \cdot Z + I_1 + \dotsc + I_c$ where $I_1,\dotsc,I_c$ are disjoint isomorphic copies of $I$ and $m^\prime = m \restriction_{\mathcal{C}(Z)}$. We know that $e_0 \in E_{H_0} \subseteq E_{m \cdot Z_1}$. Without loss of generality, let $e_0$ belong to $E_{I_1}$. Let $P^{(1)} = \{p \in P \mid p \cap E_{I_j} = \emptyset \mbox{ for all } j > 1 \}$. Consider the following decomposition of the above derivation (we can do this according to Remark \ref{rem_parallelization}):
	$$
	Z_1 \Rightarrow m^\prime \cdot Z + I_1 + I_2 + \dotsc + I_c \underset{\fr(P^{(1)})}{\Rightarrow} H^\prime + I_2 + \dotsc + I_c
	\underset{\fr(P \setminus P^{(1)})}{\Rightarrow} H_0 + G.
	$$
	The hypergraph $H^\prime$ contains the hyperedge $e_0$. We claim that $e_0$ is not adjacent to any other hyperedge in $H^\prime$, i.e.~that, for all $i \in [N]$, the only hyperedge incident to $att_{H^\prime}(e_0)(i)$ is $e_0$. The proof is ex falso. Assume that $att_{H^\prime}(e^\prime)(i^\prime) = att_{H^\prime}(e_0)(i)$ for some other $e^\prime \in E_{H^\prime}$. This hyperedge must disappear after the application of $\fr(P \setminus P^{(1)})$, hence $\{e^\prime,e^{\prime\prime}\} \in P \setminus P^{(1)}$ for some hyperedge $e^{\prime\prime}$. Therefore, $e^{\prime\prime}$ belongs to $E_{I_j}$ for some $j>1$ and $att_{H^\prime}(e^\prime)(i^\prime)$ is identified with $att_{I_j}(e^{\prime\prime})(i^\prime)$. However, there is an $a_0$-labeled hyperedge in $I_j$ that is incident to $att_{I_j}(e^{\prime\prime})(i^\prime)$; thus, after the application of $\fr(P \setminus P^{(1)})$, it becomes adjacent to $e_0$. This is a contradiction since $e_0$ in $H_0 + G$ is not adjacent to any other $a_0$-labeled hyperedge.
	
	As a consequence of the above reasonings, $H^\prime = H_0^\prime + G^\prime$ where $H_0^\prime$ is a connected hypergraph such that $E_{H_0^\prime} = \{e_0\}$. We conclude from this that the fusion rule application $\fr(P \setminus P^{(1)})$ affects only $G^\prime + I_2 + \dotsc + I_c$. This implies that $H_0^\prime = H_0$. Thus we can construct the following decomposition:
	$$
	Z_1 \Rightarrow m^\prime \cdot Z + I_1 \underset{\fr(P^{(1)})}{\Rightarrow} H_0 + G^\prime.
	$$
	In other words, we have proved that, if $H_0 \in L(FG_1)$, then we can derive it using exactly one copy of $I$.
	
	Now, let $P^{(2)} = \{p \in P^{(1)} \mid p \cap E_{I_1} = \emptyset \}$. Consider the following derivation:
	$$
	Z_1 \Rightarrow m^\prime \cdot Z + I_1 \underset{\fr(P^{(2)})}{\Rightarrow} H^{\prime\prime} + I_1
	\underset{\fr(P^{(1)} \setminus P^{(2)})}{\Rightarrow} H_0 + G^\prime.
	$$
	Let $E_{I_1} = \{e_1^I, \dotsc, e_k^I, e_0\}$  (we introduced this notation in Remark \ref{remark_eI}). Since each $e_i^I$ is adjacent to $e_0$, the hyperedge $e_i^I$ must disappear after the fusion $\fr(P^{(1)} \setminus P^{(2)})$. Therefore, $P^{(1)} \setminus P^{(2)} = \{\{e_i^{\prime\prime},e_i^I\} \mid i = 1,\dotsc, k\}$ for some hyperedges $e_i^{\prime\prime} \in E_{H^{\prime\prime}}$ ($i=1,\dotsc,k$).
	
	Consider the following hypergraph $H^{\prime\prime}_1$:
	\begin{itemize}
		\item $E_{H^{\prime\prime}_1} = \{e_1^{\prime\prime},\dotsc,e_k^{\prime\prime}\}$; \qquad $V_{H^{\prime\prime}_1}$ consists of the attachment vertices of the hyperedges from $E_{H^{\prime\prime}_1}$;
		\item $att_{H^{\prime\prime}_1} = att_{H^{\prime\prime}} {\restriction}_{E_{H^{\prime\prime}_1}}$; $lab_{H^{\prime\prime}_1} = lab_{H^{\prime\prime}} {\restriction}_{E_{H^{\prime\prime}_1}}$.
	\end{itemize}
	Clearly, $H^{\prime\prime}_1$ is a subhypergraph of $H^{\prime\prime}$. We claim that $H^{\prime\prime} \cong H^{\prime\prime}_1 + H^{\prime\prime}_2$ for some hypergraph $H^{\prime\prime}_2$. To prove this, it suffices to explain why there is no hyperedge $e \in E_{H^{\prime\prime}}$ that does not belong to $E_{H^{\prime\prime}_1}$ and that is attached to one of the vertices from $V_{H^{\prime\prime}_1}$. The proof is ex falso. Assume that there is such a hyperedge $e$. It does not disappear after the application of $\fr(P^{(1)} \setminus P^{(2)})$ (because it is not equal to any of $e_i^{\prime\prime}$); then, it is adjacent to $e_0$ in $H_0+G^\prime$. However, no hyperedge is adjacent to $e_0$ in $H_0+G^\prime$, which leads to a contradiction.
	
	Thus, we have established that $H^{\prime\prime} = H^{\prime\prime}_1 + H^{\prime\prime}_2$. The hypergraph $H^{\prime\prime}_2$ is not affected by the fusion $\fr(P^{(1)} \setminus P^{(2)})$, and we have that
	$$
	H^{\prime\prime}_1 + I_1
	\underset{\fr(P^{(1)} \setminus P^{(2)})}{\Rightarrow} H_0.
	$$
	Lemma \ref{lemma_fusion_isomorphism} implies that $H^{\prime\prime}_1$ is isomorphic to $H_1$. Finally, we have that
	$$
	Z \Rightarrow m^\prime \cdot Z \underset{\fr(P^{(2)})}{\Rightarrow} H_1^{\prime\prime} + H_2^{\prime\prime}
	$$
	Therefore, $H_1 \cong H_1^{\prime\prime} \in L(FG)$, as expected.
\end{proof}

\begin{proof}[Proof of Theorem \ref{th_mem2}]
	Given $FG$ and $H_1$, we introduce the fusion grammar $FG_1$ and the hypergraph $H_0$ according to Construction \ref{construction_mem2}. Then, $H_1 \in L(FG)$ if and only if $H_0 \in L(FG_1)$ according to Lemma \ref{lemma_mem2_main}. The latter problem is decidable according to Theorem \ref{th_mem1}.
\end{proof}

This almost proves decidability of the general problem MEM. What cases have not we covered? Only the case where the input hypergraph is a single vertex without hyperedges. (Note that fusion grammars without markers and connectors generate only connected hypergraphs.) Let $O = \langle \{v\},\emptyset,\emptyset,\emptyset\rangle$ be such a hypergraph. Let us prove decidability of the membership problem for it too.

\begin{problem}[MEM-O]
	\leavevmode
	\\
	\textit{Input.} 
	A fusion grammar without markers and connectors $FG = (Z,F,T)$.
	\\
	\textit{Problem MEM-O.} Does $O = \langle \{v\},\emptyset,\emptyset,\emptyset\rangle$ belong to $L(FG)$?
\end{problem}

\begin{proposition}\label{prop_mem-o}
	MEM-O is decidable.
\end{proposition}

We are going to reduce this problem to MEM-1. Let us add a new terminal label $o$ to $T$ such that $\type(o) = 1$. Given $C \in \mathcal{C}(Z)$ and $v \in V_C$, let $o(C,v)$ denote the hypergraph obtained from $C$ by attaching an $o$-labeled hyperedge to $v$. Let us define $FG^o = (Z^o,F,T \sqcup \{o\})$ where $Z^o$ is as follows: 
$$
Z^o = Z + \sum_{C \in \mathcal{C}(Z)} \sum_{v \in V_C} o(C,v).
$$

{\color{comment}
	Informally, we take each connected component $C$ of $Z$ and attach an $o$-labeled hyperedge to one of its vertices; the disjoint union of all such hypergraphs along with connected components from $Z$ is $Z^o$.
}

Let $H^o = \vcenter{\hbox{{\tikz[baseline=.1ex]{
				\node[hyperedge] (E) at (0,0) {$o$};
				\node[vertex] (V) at (0.7,0) {};
				\draw[-] (E) -- node[above] {\scriptsize 1} (V);
}}}}
$ be the hypergraph that has one $o$-labeled hyperedge attached to the only vertex of $H^o$.
\begin{lemma}
	$O \in L(FG)$ if and only if $H^o \in L(FG^o)$.
\end{lemma}
\begin{proof}
	Let $O \in L(FG)$. Then there is a most parallelised derivation of the form 
	$$
	Z \underset{m}{\Rightarrow} Z^\prime \underset{\fr(P)}{\Rightarrow} O + G
	$$
	where $Z^\prime = m \cdot Z$. The hypergraph $O$ has one vertex (say, $v_O$) and no hyperedges. This only vertex is the equivalence class of vertices in $V_{Z^\prime}$, i.e.~$v_O = [v^\prime]_{\equiv_P}$ for some $v^\prime \in V_{Z^\prime}$. The latter vertex $v^\prime$ belongs to one of the connected components $C^\prime \in \mathcal{C}(Z^\prime)$. Let us replace this component by $o(C^\prime,v^\prime)$. The new derivation will be a valid one in the new grammar $FG^o$, and the resulting hypergraph will be $H^o+G$, hence $H^o \in L(FG^o)$.
	
	Conversely, let $H^o \in L(FG^o)$; this means that there is a derivation $Z^o \Rightarrow^\ast H^o + G$. Removing all $o$-labeled hyperedges from hypergraphs within it results in a derivation $Z \Rightarrow^\ast O + \rem_{\{o\}}(G)$, hence $O \in L(FG)$.
\end{proof}

Finally, deciding whether $H^o \in L(FG^o)$ is an instance of MEM-1, for which an algorithm has already been presented. This completes the proof of Theorem \ref{th_main_membership}.

\subsection{Non-Emptiness Problem}\label{ssec_results_ne}

Now, we are going to prove Theorem \ref{th_main_non-emptiness}. It turns out that the non-emptiness problem NE can be reduced to the membership problem for fusion grammars without markers and connectors, namely, to MEM-1.

Assume that we are given a fusion grammar $FG = (Z,F,T, \mathcal{M},\mathcal{K})$ for which we aim to answer whether $L(FG)$ is empty or not. Let us extend the terminal alphabet $T$ by a new label $a_0$ such that $\type(a_0) = 1$.\footnote{This is not the same $a_0$ as that from Construction \ref{construction_mem2}. We use the same symbol because we are going to reduce the non-emptiness problem to MEM-1, and $a_0$ plays the same role as in Section \ref{ssec_results_mem1}.}
\begin{definition}
	The \emph{squeezing function} $\sqz$ is defined on $\mathcal{H}(\Sigma)$ as follows: $\sqz(H)$ equals the hypergraph $H/{\equiv}$ where $\equiv$ is the smallest equivalence relation such that, whenever $lab_H(e)$ belongs to $T \cup \mathcal{M} \cup \mathcal{K}$ for some $e \in E_H$, it holds that $att_H(e)(i) \equiv att_H(e)(j)$ for all $i,j \in [\type(e)]$.
\end{definition}
\begin{example}
	Let us consider the fusion grammar $\FGPex = (\Zex,\{A,B,C\},\{a\},\{b\},\{c\})$ from Example \ref{example_FG} and let us apply the squeezing function to each $\Zex_i$.
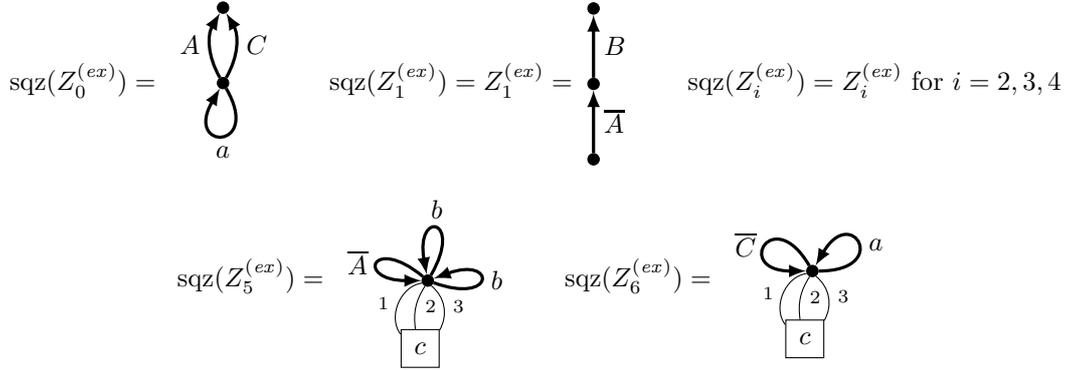
	$$
	\sqz(\Zex_0) \;\;=\;\; 
	\vcenter{\hbox{{\tikz[baseline=.1ex]{
					\node[vertex] (VO) at (0,0) {};
					\node[vertex] (VB) at (0,-0.7) {};
					\draw[-latex, thick] (VB) to[out=-60,in=-120,looseness=20] node[below] {$a$} (VB);
					\draw[-latex, thick] (VB) to[bend left = 40] node[left] {$A$} (VO);
					\draw[-latex, thick] (VB) to[bend right = 40] node[right] {$C$} (VO);
	}}}}
	\qquad\qquad
	\sqz(\Zex_1) = \Zex_1 \;\;=\;\; 
	\vcenter{\hbox{{\tikz[baseline=.1ex]{
					\node[vertex] (V1) at (0,0) {};
					\node[vertex] (V2) at (0,0.7) {};
					\node[vertex] (V3) at (0,1.4) {};
					\draw[-latex, thick] (V1) -- node[right] {$\overline{A}$} (V2);
					\draw[-latex, thick] (V2) -- node[right] {$B$} (V3);
	}}}}
	\qquad\qquad
	\sqz(\Zex_i) = \Zex_i \mbox{ for } i=2,3,4
	$$
	$$
	\sqz(\Zex_5) \;\;=\;\; 
	\vcenter{\hbox{{\tikz[baseline=.1ex]{
					\node[vertex] (V) at (0,0) {};
					\node[hyperedge] (EL) at (-0.1,-0.9) {$c$};
					\draw[-latex, thick] (V) to[out=-20,in=20,looseness=40] node[right] {$b$} (V);
					\draw[-latex, thick] (V) to[out=60,in=100,looseness=40] node[above] {$b$} (V);
					\draw[-latex, thick] (V) to[out=140,in=180,looseness=40] node[left] {$\overline{A}$} (V);
					\draw[-] (EL) to[bend left = 60] node[left] {\scriptsize 1} (V);
					\draw[-] (EL) to[bend left = 20] node[right] {\scriptsize 2} (V);
					\draw[-] (EL) to[bend right = 40] node[right] {\scriptsize 3} (V);
	}}}}
	\qquad\qquad
	\sqz(\Zex_6) \;\;=\;\; 
	\vcenter{\hbox{{\tikz[baseline=.1ex]{
					\node[vertex] (V) at (0,0) {};
					\node[hyperedge] (EL) at (-0.1,-0.9) {$c$};
					\draw[-latex, thick] (V) to[out=0,in=60,looseness=30] node[right] {$a$} (V);
					\draw[-latex, thick] (V) to[out=120,in=180,looseness=30] node[left] {$\overline{C}$} (V);
					\draw[-] (EL) to[bend left = 60] node[left] {\scriptsize 1} (V);
					\draw[-] (EL) to[bend left = 20] node[right] {\scriptsize 2} (V);
					\draw[-] (EL) to[bend right = 40] node[right] {\scriptsize 3} (V);
	}}}}
	$$
\end{example}

\begin{definition}
	Given a hypergraph $H$ along with a hyperedge $e_0 \in E_H$, the \emph{designating function} $\dsg(H,e_0)$ is defined as follows: 
	\begin{enumerate}
		\item if $att_H(e_0) = v_0\dotsc v_0$ (i.e.~all the attachment vertices of $e_0$ coincide) and $lab_H(e_0) \in \mathcal{M}$, then $\dsg(H,e_0) = H^\prime$ such that $V_{H^\prime} = V_H$, $E_{H^\prime} = E_H$, $att_{H^\prime}(e_0) = v_0$, $lab_{H^\prime}(e_0) = a_0$, and $att_{H^\prime}$ coincides with $att_H$ for all the remaining hyperedges as well as $lab_{H^\prime}$ and $lab_H$;
		\item if the condition of the previous sentence is not met, then the output is the empty hypergraph.
	\end{enumerate}
	The \emph{sum of designated hypergraphs} $\Dsg(H)$ equals $H + \sum_{e \in E_H} \dsg(H,e)$.
\end{definition}

{\color{comment}
	The designating function is designed to be applied to a hypergraph after the squeezing function. It takes a hyperedge with a marker label and replaces it with a hyperedge labeled by the new special label $a_0$ of type 1. Formally, $\dsg(H,e_0)$ is defined even if $e_0$ is not labeled by an element of $\mathcal{M}$. This is done in order to define $\Dsg(H)$ without making additional clarifications regarding the summation limits. Indeed, if $e_0$ is not labeled by a marker, then $\dsg(H,e_0)$ is the empty hypergraph and it makes no contribution to the sum of designated hypergraphs.
}

\begin{example}
	In the grammar $\FGPex$, the only marker symbol is $b$. Since, for each $\Zex_i$ such that $i\ne 5$ and for each hyperedge $e \in E_{\Zex_i}$, the label $lab_{\Zex_i}(e)$ is not the marker $b$, it holds that $\dsg(\sqz(\Zex_i),e)$ is the empty hypergraph. The hypergraph $\Zex_5$, however, contains two $b$-labeled hyperedges (denote them by $e^b_1$ and $e^b_2$). Then
	$$
	\dsg(\sqz(\Zex_5),e^b_1) = 
	\vcenter
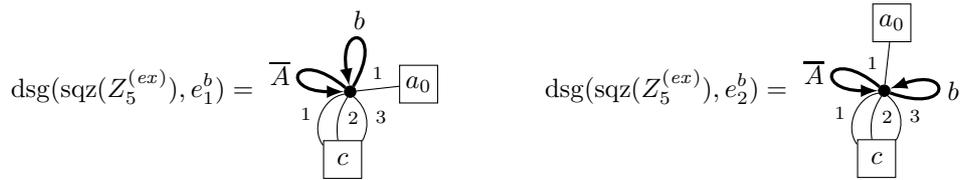
{\hbox{{\tikz[baseline=.1ex]{
					\node[vertex] (V) at (0,0) {};
					\node[hyperedge] (EL) at (-0.1,-0.9) {$c$};
					\node[hyperedge] (EB) at (0.9,0.1) {$a_0$};
					\draw[-] (EB) -- node[above] {\scriptsize 1} (V);
					\draw[-latex, thick] (V) to[out=60,in=100,looseness=40] node[above] {$b$} (V);
					\draw[-latex, thick] (V) to[out=140,in=180,looseness=40] node[left] {$\overline{A}$} (V);
					\draw[-] (EL) to[bend left = 60] node[left] {\scriptsize 1} (V);
					\draw[-] (EL) to[bend left = 20] node[right] {\scriptsize 2} (V);
					\draw[-] (EL) to[bend right = 40] node[right] {\scriptsize 3} (V);
	}}}}
	\qquad
	\qquad
	\dsg(\sqz(\Zex_5),e^b_2) = 
	\vcenter{\hbox{{\tikz[baseline=.1ex]{
					\node[vertex] (V) at (0,0) {};
					\node[hyperedge] (EL) at (-0.1,-0.9) {$c$};
					\node[hyperedge] (EB) at (0.2,0.7) {$a_0$};
					\draw[-] (EB) -- node[left] {\scriptsize 1} (V);
					\draw[-latex, thick] (V) to[out=-20,in=20,looseness=40] node[right] {$b$} (V);
					\draw[-latex, thick] (V) to[out=140,in=180,looseness=40] node[left] {$\overline{A}$} (V);
					\draw[-] (EL) to[bend left = 60] node[left] {\scriptsize 1} (V);
					\draw[-] (EL) to[bend left = 20] node[right] {\scriptsize 2} (V);
					\draw[-] (EL) to[bend right = 40] node[right] {\scriptsize 3} (V);
	}}}}
	$$
	Consequently, the sum of designated hypergraphs $\Dsg(\sqz(\Zex_i))$ equals $\sqz(\Zex_i)$ for $i \ne 5$ (because \\ $\Dsg(\sqz(\Zex_i)) = \sqz(\Zex_i) + \sum_{e \in E_{\Zex_i}} \dsg(\sqz(\Zex_i),e)$, and all the summands except for the first one are empty hypergraphs). The situation is different only for $\sqz(\Zex_5)$.
	$$
	\Dsg(\sqz(\Zex_5))
	=
	\vcenter{\hbox{{\tikz[baseline=.1ex]{
					\node[vertex] (V) at (0,0) {};
					\node[hyperedge] (EL) at (-0.1,-0.9) {$c$};
					\draw[-latex, thick] (V) to[out=-20,in=20,looseness=40] node[right] {$b$} (V);
					\draw[-latex, thick] (V) to[out=60,in=100,looseness=40] node[above] {$b$} (V);
					\draw[-latex, thick] (V) to[out=140,in=180,looseness=40] node[left] {$\overline{A}$} (V);
					\draw[-] (EL) to[bend left = 60] node[left] {\scriptsize 1} (V);
					\draw[-] (EL) to[bend left = 20] node[right] {\scriptsize 2} (V);
					\draw[-] (EL) to[bend right = 40] node[right] {\scriptsize 3} (V);
	}}}}
	\;+\;
	\vcenter{\hbox{{\tikz[baseline=.1ex]{
					\node[vertex] (V) at (0,0) {};
					\node[hyperedge] (EL) at (-0.1,-0.9) {$c$};
					\node[hyperedge] (EB) at (0.9,0.1) {$a_0$};
					\draw[-] (EB) -- node[above] {\scriptsize 1} (V);
					\draw[-latex, thick] (V) to[out=60,in=100,looseness=40] node[above] {$b$} (V);
					\draw[-latex, thick] (V) to[out=140,in=180,looseness=40] node[left] {$\overline{A}$} (V);
					\draw[-] (EL) to[bend left = 60] node[left] {\scriptsize 1} (V);
					\draw[-] (EL) to[bend left = 20] node[right] {\scriptsize 2} (V);
					\draw[-] (EL) to[bend right = 40] node[right] {\scriptsize 3} (V);
	}}}}
	\;+\;
	\vcenter{\hbox{{\tikz[baseline=.1ex]{
					\node[vertex] (V) at (0,0) {};
					\node[hyperedge] (EL) at (-0.1,-0.9) {$c$};
					\node[hyperedge] (EB) at (0.2,0.7) {$a_0$};
					\draw[-] (EB) -- node[left] {\scriptsize 1} (V);
					\draw[-latex, thick] (V) to[out=-20,in=20,looseness=40] node[right] {$b$} (V);
					\draw[-latex, thick] (V) to[out=140,in=180,looseness=40] node[left] {$\overline{A}$} (V);
					\draw[-] (EL) to[bend left = 60] node[left] {\scriptsize 1} (V);
					\draw[-] (EL) to[bend left = 20] node[right] {\scriptsize 2} (V);
					\draw[-] (EL) to[bend right = 40] node[right] {\scriptsize 3} (V);
	}}}}
	$$
\end{example}

\begin{construction}\label{construction_ne}
	First, hereinafter, let $H_0$ denote the hypergraph $\vcenter{\hbox{{\tikz[baseline=.1ex]{
					\node[hyperedge] (E) at (0,0) {$a_0$};
					\node[vertex] (V) at (0.7,0) {};
					\draw[-] (E) -- node[above] {\scriptsize 1} (V);
	}}}}$ with one $a_0$-labeled hyperedge attached to the only vertex. Secondly, let us define a fusion grammar without markers and connectors $\widehat{FG} = (\Xi, F, \{a_0\})$ where
	$$
	\Xi = \sum_{C \in \mathcal{C}(Z)} \rem_{T \cup \mathcal{M} \cup \mathcal{K}} \left(\Dsg(\sqz(C))\right).
	$$
\end{construction}

\begin{example}
	The resulting hypergraph $\Xi$ constructed from the fusion grammar $\FGPex$ is shown below.
	$$
	\Xi = 
	\vcenter
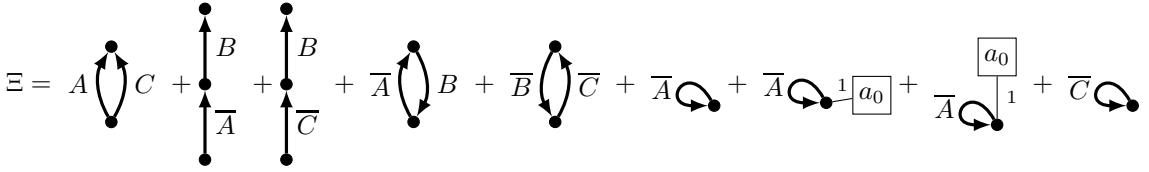
{\hbox{{\tikz[baseline=.1ex]{
					\node[vertex] (VO) at (0,0) {};
					\node[vertex] (VB) at (0,-1) {};
					\draw[-latex, thick] (VB) to[bend left = 30] node[left] {$A$} (VO);
					\draw[-latex, thick] (VB) to[bend right = 30] node[right] {$C$} (VO);
	}}}}
	\:+\:
	\vcenter{\hbox{{\tikz[baseline=.1ex]{
					\node[vertex] (V1) at (0,0) {};
					\node[vertex] (V2) at (0,0.8) {};
					\node[vertex] (V3) at (0,1.6) {};
					\draw[-latex, thick] (V1) -- node[right] {$\overline{A}$} (V2);
					\draw[-latex, thick] (V2) -- node[right] {$B$} (V3);
	}}}}
	\:+\:
	\vcenter{\hbox{{\tikz[baseline=.1ex]{
					\node[vertex] (V1) at (0,0) {};
					\node[vertex] (V2) at (0,0.8) {};
					\node[vertex] (V3) at (0,1.6) {};
					\draw[-latex, thick] (V1) -- node[right] {$\overline{C}$} (V2);
					\draw[-latex, thick] (V2) -- node[right] {$B$} (V3);
	}}}}
	\:+\:
	\vcenter{\hbox{{\tikz[baseline=.1ex]{
					\node[vertex] (V1) at (0,0) {};
					\node[vertex] (V2) at (0,1) {};
					\draw[-latex, thick] (V1) to[bend left=30] node[left] {$\overline{A}$} (V2);
					\draw[-latex, thick] (V2) to[bend left=30] node[right] {$B$} (V1);
	}}}}
	\:+\:
	\vcenter{\hbox{{\tikz[baseline=.1ex]{
					\node[vertex] (V1) at (0,0) {};
					\node[vertex] (V2) at (0,1) {};
					\draw[latex-, thick] (V1) to[bend left=30] node[left] {$\overline{B}$} (V2);
					\draw[latex-, thick] (V2) to[bend left=30] node[right] {$\overline{C}$} (V1);
	}}}}
	\:+\:
	\vcenter{\hbox{{\tikz[baseline=.1ex]{
					\node[vertex] (V) at (0,0) {};
					\draw[-latex, thick] (V) to[out=120,in=180,looseness=20] node[left] {$\overline{A}$} (V);
	}}}}
	\:+\:
	\vcenter{\hbox{{\tikz[baseline=.1ex]{
					\node[vertex] (V) at (0,0) {};
					\node[hyperedge] (EB) at (0.6,0.1) {$a_0$};
					\draw[-] (EB) -- node[above] {\scriptsize 1} (V);
					\draw[-latex, thick] (V) to[out=120,in=180,looseness=20] node[left] {$\overline{A}$} (V);
	}}}}
	\:+\:
	\vcenter{\hbox{{\tikz[baseline=.1ex]{
					\node[vertex] (V) at (0,0) {};
					\node[hyperedge] (EB) at (0,0.9) {$a_0$};
					\draw[-] (EB) -- node[right] {\scriptsize 1} (V);
					\draw[-latex, thick] (V) to[out=120,in=180,looseness=20] node[left] {$\overline{A}$} (V);
	}}}}
	\:+\:
	\vcenter{\hbox{{\tikz[baseline=.1ex]{
					\node[vertex] (V) at (0,0) {};
					\draw[-latex, thick] (V) to[out=120,in=180,looseness=20] node[left] {$\overline{C}$} (V);
	}}}}
	$$
\end{example}

Our objective is to prove that the nonemptiness problem for $FG$ is equivalent to the membership problem for $\widehat{FG}$ and $H_0$.
\begin{lemma}\label{lemma_ne_1}
	If $L(FG)$ is non-empty, then $H_0$ belongs to $L(\widehat{FG})$.
\end{lemma}
\begin{proof}
	The set $L(FG)$ is non-empty if and only if there exists a most parallelised derivation of the form 
	\begin{equation}\label{eq_der_NE1}
		Z \Rightarrow m \cdot Z \underset{\fr(P)}{\Rightarrow} Y + G
	\end{equation}
	for some $G$ and for a connected hypergraph $Y \in \mathcal{H}(T \cup \mathcal{M} \cup \mathcal{K})$ that contains at least one marker label. Fix any hyperedge $e_0 \in E_Y$ such that $lab_Y(e_0) \in \mathcal{M}$. Any hyperedge from $Y$ is also present in $m \cdot Z$. Let $e_0 \in E_{C_0}$ for $C_0 \in \mathcal{C}(m \cdot Z)$.
	
	We shall transform (\ref{eq_der_NE1}) into a derivation of $H_0$ in $\widehat{FG}$. First, let us apply the squeezing function to the hypergraphs in (\ref{eq_der_NE1}):
	\begin{equation*}
		\sqz(Z) \Rightarrow \sqz(m \cdot Z) = m \cdot \sqz(Z) \underset{\fr(P)}{\Rightarrow} \sqz(Y) + \sqz(G).
	\end{equation*}
	This is also a valid derivation. Note that the squeezing function commutes with both multiplication and fusion. It is also important to note that $\sqz(Y)$ is a hypergraph with one vertex, because $Y$ is connected and all its hyperedges are squeezed (because they are labeled by members of $T \cup \mathcal{M} \cup \mathcal{K}$).
	
	Secondly, we know that $\sqz(m \cdot Z)$ contains $\sqz(C_0)$. Let us replace this connected component within $\sqz(m \cdot Z)$ by $\dsg(\sqz(C_0),e_0)$. In other words, we simply replace $e_0$ in $C_0$ by a type-1 hyperedge labeled by $a_0$. Let us denote the resulting hypergraph by $Z^\prime$. The parallelised fusion rule application now looks as follows:
	\begin{equation*}
		Z^\prime \underset{\fr(P)}{\Rightarrow} \dsg(\sqz(Y),e_0) + \sqz(G).
	\end{equation*}

	Thirdly, let us remove all the labels from $T \cup \mathcal{M} \cup \mathcal{K}$, which results in the following derivation:
	\begin{equation*}
		\Xi \Rightarrow \rem_{T \cup \mathcal{M} \cup \mathcal{K}}(Z^\prime) \underset{\fr(P)}{\Rightarrow} \rem_{T \cup \mathcal{M} \cup \mathcal{K}}(\dsg(\sqz(Y),e_0)) + \rem_{T \cup \mathcal{M} \cup \mathcal{K}}(\sqz(G)).
	\end{equation*}
	Recall that $\Xi$ is the hypergraph defined in Construction \ref{construction_ne}.
	It remains to observe that $\rem_{T \cup \mathcal{M} \cup \mathcal{K}}(\dsg(\sqz(Y),e_0))$ is equal to $H_0$; indeed, $\dsg(\sqz(Y),e_0)$ is a hypergraph with one vertex, to which $e_0$ is attached along with hyperedges with labels from $T \cup \mathcal{M} \cup \mathcal{K}$; removing the latter yields $\vcenter{\hbox{{\tikz[baseline=.1ex]{
					\node[hyperedge] (E) at (0,0) {$a_0$};
					\node[vertex] (V) at (0.7,0) {};
					\draw[-] (E) -- node[above] {\scriptsize 1} (V);
	}}}} = H_0$.
\end{proof}

\begin{lemma}\label{lemma_ne_2}
	If $H_0$ belongs to $L(\widehat{FG})$, then $L(FG)$ is non-empty.
\end{lemma}

\begin{proof}
	There exists a most parallelised derivation of the form
	$
		\Xi \Rightarrow \Xi^\prime \underset{\fr(P)}{\Rightarrow} H_0 + \widehat{G}
	$
	where 
	$$
	\Xi^\prime = \sum_{C \in \mathcal{C}(Z)} \sum_{D \in \mathcal{C}(\Dsg(\sqz(C)))} m(D) \cdot \rem_{T \cup \mathcal{M} \cup \mathcal{K}} (D).
	$$
	for some multiplicity function $m$. Let us construct the hypergraph
	$$
	Z^\prime = \sum_{C \in \mathcal{C}(Z)} \left(\sum_{D \in \mathcal{C}(\Dsg(\sqz(C)))} m(D)\right) \cdot C.
	$$
	Clearly, $Z^\prime$ is a multiplication of $Z$. Note that $\Xi^\prime$ is obtained from $Z^\prime$ by applying $\sqz$ (i.e.~by squeezing all the hyperedges with labels from $T \cup \mathcal{M} \cup \mathcal{K}$), by replacing some hyperedges with labels from $\mathcal{M}$ by $a_0$-labeled ones and then by removing all hyperedges labeled by members of $T \cup \mathcal{M} \cup \mathcal{K}$. None of these transformations affects hyperedges with labels from $F \cup \overline{F}$. Therefore, one can apply the rule $fr(P)$ to $Z^\prime$. 
	
	Let $Z^\prime \underset{\fr(P)}{\Rightarrow} H^\prime$. We know that the same fusion rule applied to $\Xi^\prime$ results in $H_0 + \widehat{G}$. The latter hypergraph is obtained from $H^\prime$ in the same way as $\Xi^\prime$ is obtained from $Z^\prime$, namely, by squeezing all the hyperedges with labels from $T \cup \mathcal{M} \cup \mathcal{K}$, by replacing some hyperedges with labels from $\mathcal{M}$ by $a_0$-labeled ones and by removing all hyperedges with labels belonging to $T \cup \mathcal{M} \cup \mathcal{K}$. Let us look at these steps in the reverse order.
	\begin{enumerate}
		\item We know that $H_0 + \widehat{G} = \rem_{T \cup \mathcal{M} \cup \mathcal{K}}(X)$ for some hypergraph $X$. Moreover, we know that $X$ is the result of applying the squeezing function and the designating function. Therefore, each hyperedge in $X$ that is labeled by an element of $T \cup \mathcal{M} \cup \mathcal{K}$ must be attached to exactly one vertex. Therefore, $X = H_0^\prime + \widehat{G}^\prime$ where $\rem_{T \cup \mathcal{M} \cup \mathcal{K}}(H_0^\prime) = H_0$ and $\rem_{T \cup \mathcal{M} \cup \mathcal{K}}(\widehat{G}^\prime) = \widehat{G}$. The hypergraph $H_0^\prime$ can contain several $(T \cup \mathcal{M} \cup \mathcal{K})$-labeled hyperedges, which are all attached to the only vertex of $H_0^\prime$.
		\item The hypergraph $H_0^\prime$ has one vertex ($V_{H_0^\prime} = \{v_0\}$) and $att_{H_0^\prime}(e) = v_0^k = \underbrace{v_0\dotsc v_0}_{k\mbox{ times}}$ for any $e \in E_{H_0^\prime}$ where $k = \type(lab_{H_0^\prime}(e))$. Exactly one hyperedge of $H_0^\prime$ (which we denoted by $e_0$) is labeled by $a_0$. According to the definition of the designating function, this means that $H_0^\prime = \dsg(H_0^{\prime\prime},e_0)$ where $H_0^{\prime\prime}$ is obtained from $H_0^\prime$ by redefining the attachment function and the labeling function for $e_0$ as follows: $lab_{H_0^{\prime\prime}}(e_0) = \mu \in \mathcal{M}$, $att_{H_0^{\prime\prime}}(e_0) = v_0\dotsc v_0$ (where $v_0$ is repeated $\type(\mu)$ times).
		\item Finally, we know that $H_0^{\prime\prime}$ is the result of squeezing, i.e.~$H_0^{\prime\prime} = \sqz(Y)$ for some $Y$.
	\end{enumerate}
	Summing up, $\rem_{T \cup \mathcal{M} \cup \mathcal{K}}(\dsg(\sqz(Y),e_0)) \cong H_0$.
	We know that $Z^\prime \underset{\fr(P)}{\Rightarrow} Y + G$ for some hypergraph $G$ (which is obtained from $\widehat{G}$ in the same way $Y$ is obtained from $H_0$ according to the above three-step procedure). Observe that $Y$ is connected; indeed, otherwise $\sqz(Y)$ would be disconnected and hence $H_0$ would contain at least two vertices, which is not the case. Furthermore, all labels of $Y$ are from $T \cup \mathcal{M} \cup \mathcal{K}$. Moreover, one of its labels belongs to $\mathcal{M}$ (namely, that of $e_0$). All these facts imply that $Y$ belongs to $L(FG)$. The proof is concluded.
\end{proof}

Combining the above lemmas leads us to the proof of Theorem \ref{th_main_non-emptiness}.
\begin{proof}[Proof of Theorem \ref{th_main_non-emptiness}]
	Given a fusion grammar $FG$, we effectively construct $\widehat{FG}$ using Construction \ref{construction_ne}. Lemmas \ref{lemma_ne_1} and \ref{lemma_ne_2} prove that the non-emptiness problem for $FG$ is equivalent to the fact that $H_0$ belongs to $L(\widehat{FG})$. The latter is an instance of the problem MEM-1, which is decidable according to Theorem \ref{th_mem1}.
\end{proof}

\subsection{NEXPTIME Upper Bound}\label{ssec_NEXPTIME}

Now, what is the complexity of the decision problems we have studied? A straightforward inspection of the proofs of Theorem 
\ref{th_main_membership} and Theorem \ref{th_main_non-emptiness} yields the upper NEXPTIME bound.

\begin{theorem}\label{th_main_NEXPTIME}
	The problems MEM and NE are in NEXPTIME.
\end{theorem}
\begin{proof}
	Let the fusion grammar without markers and connectors $FG = (Z,F,T)$ and the hypergraph $H_0$ with one hyperedge of type $N$ and $n$ vertices be the input of MEM-1 (see Section \ref{ssec_results_mem1}). Let us denote the sum of the sizes of $FG$ and of $H_0$ by $L$. 
	The size of $\Colors$ is $n+1 = \vert H_0 \vert$. The size of $\FCol$ is not greater than 
	$$
	\vert \Sigma \vert \cdot (n+1)^{r_{\max}} \cdot 2^{N^2 \cdot r_{\max}} = \exp(L^{O(1)})
	$$
	where $r_{\max} \le L$ is the maximum type of hyperedges in the start hypergraph $Z$. The size of $\SigmaCol$ is $2 \vert \FCol \vert+1$. Each hypergraph from the set $S$ or the set $S^b$ (see Definition \ref{def_mem}) is obtained from a connected component of $Z$ by replacing each label by some label from $\SigmaCol$; therefore, 
	$$\vert S \uplus S^b\vert \le \left\vert \SigmaCol \right\vert^{\vert Z \vert} \le (2\exp(L^{O(1)})+1)^L = \exp(L^{O(1)}).$$
	Now, consider the system of equations (\ref{eq_system_mem1}). It has the form $Ax = b$ where $b = -\beta(B)$; the matrix $A$ is composed of vectors $\beta(C)$ for $C \in S$; $x$ is the vector of unknown numbers $k(C) \in \Nat$. The size of the system is defined by the size of $A$ and $b$. The absolute value of any number occurring in $A$ and $b$ does not exceed the maximum number of hyperedges in $C \in \mathcal{C}(Z)$ and thus the size of $Z$. The number of columns of $A$ equals $\vert S \vert$; the number of its rows equals $d = \vert \FCol \setminus \mathfrak{W} \vert \le \vert \FCol \vert$. Therefore, the size of $A$ is $O(\log \vert Z \vert \cdot \vert S \vert \cdot \vert \FCol \vert)$, and we have shown that each of the factors is $\exp(L^{O(1)})$. Solvability of $Ax = b$ over $x \in \Nat^{\vert S \vert}$ can be checked by a non-deterministic Turing machine using polynomial time $p(t)$ \cite{GareyJ79}, and thus solvability of (\ref{eq_system_mem1}) can be checked by a Turing machine using time $p\left(\exp(L^{O(1)})\right) = \exp(L^{O(1)})$. The number of such systems equals $\vert S^b \vert = \exp(L^{O(1)})$, and thus MEM-1 lies in NEXPTIME.
	\\
	The problem MEM-2 is reduced to MEM-1 using Construction \ref{construction_mem2} in polynomial time, because the size of $I$ and of $H_0$ is linear w.r.t. that of $H$; similarly, MEM-O is reduced to MEM-1 in polynomial time. This proves that MEM is in NEXPTIME. Finally, the non-emptiness problem NE is reduced to MEM-1 via the deterministic polynomial procedure described in Section \ref{ssec_results_ne} (which uses the squeezing function and the designating function, all polynomial-time). This proves that NE is in NEXPTIME as well.
\end{proof}

One might be interested in whether such a high upper bound as NEXPTIME can be lowered. We claim that this is rather unlikely, in view of our recent result concerning hyperedge replacement grammars (see the preprint \cite{Pshenitsyn25}). Namely, we prove there that the uniform membership problem for hyperedge replacement grammars is EXPTIME-complete. Since hyperedge replacement grammars can be translated into fusion grammars, this implies EXPTIME-hardness of the membership problem for those too (we claim that our constructions work even for fusion grammars without markers and connectors). 

Still, the non-uniform membership problem for fusion grammars remains an open problem. If one wants to prove that each language generated by a fusion grammar belongs to some (low) complexity class, they probably need to establish a normal form for fusion grammars where each hypergraph has a short derivation. 

\subsection{What Do Markers and Connectors Complicate in the Membership Problem?}\label{ssec_results_bounded}

Let us now discuss what happens if one adds markers and connectors to fusion grammars. We start with a positive result and show how to apply the decidability results we have achieved to fusion grammars where markers and connectors are allowed but their number is bounded. Such a class of grammars has not been considered in other works.
\begin{definition}
	A \emph{fusion grammar with bounded usage of markers and connectors} (or simply \emph{bounded fusion grammar}) is a tuple $BFG = (Z,F,T,\mathcal{M},\mathcal{K},f_{\mathcal{M}},f_{\mathcal{K}})$ where
	\begin{itemize}
		\item $(Z,F,T,\mathcal{M},\mathcal{K})$ is a fusion grammar,
		\item $f_{\mathcal{M}}: \Nat \to \Nat$ and $f_{\mathcal{K}}: \Nat \to \Nat$ are computable functions.
	\end{itemize}
	The \emph{language $L(BFG)$ generated by $BFG$} consists of hypergraphs $X$ isomorphic to $\rem_{\mathcal{M} \cup \mathcal{K}}(Y)$ such that
	\begin{itemize}
		\item $Z \Rightarrow^\ast H$ for some $H$ such that $Y \in \mathcal{C}(H) \cap \mathcal{H}(T \cup \mathcal{M} \cup \mathcal{K})$;
		\item $1 \le \#_{\mathcal{M}}(Y) \le f_{\mathcal{M}}(\vert X \vert)$;
		\item $\#_{\mathcal{K}}(Y) \le f_{\mathcal{K}}(\vert X \vert)$.
	\end{itemize}
\end{definition}
{\color{comment}
	That is, we say that the number of removed markers and connectors is computably bounded by the size of the resulting hypergraph $X$. So, given $X$, we can effectively guess $Y$ by exhausting all possibilities.
}

\begin{theorem}\label{th_main_membership_bounded}
	 The membership problem for bounded fusion grammars is decidable.
\end{theorem}
\begin{proof}
	Given $BFG$ as above, let $FG = (Z,F,T \cup \mathcal{M}\cup \mathcal{K})$ be a fusion grammar without markers and connectors; in other words, markers and connectors are treated as terminal labels in $FG$.
	Consider all $Y$'s such that $X = \rem_{\mathcal{M} \cup \mathcal{K}}(Y)$, $Y \in \mathcal{H}(T \cup \mathcal{M} \cup \mathcal{K})$, and $1 \le \#_{\mathcal{M}}(Y) \le f_{\mathcal{M}}(\vert X \vert)$, $\#_{\mathcal{K}}(Y) \le f_{\mathcal{K}}(\vert X \vert)$. The set of all such $Y$'s is finite and can be constructed effectively. Let us construct all such hypergraphs $Y$ and, for each $Y$, check if it belongs to $L(FG)$. The latter is decidable according to Theorem \ref{th_main_membership}.
\end{proof}
Of course, bounded fusion grammars is an ad-hoc invention where we force decidability to be. The question arises whether each fusion grammar is bounded, i.e.~whether computable functions $f_{\mathcal{M}}$ and $f_{\mathcal{K}}$ exist for it that limit the number of markers and connectors in a derivation of any hypergraph. (One of the reviewers asked this question.) Unfortunately, we do not know the answer yet. The problem is that multiple hyperedges are allowed in a hypergraph. For example, assume that the hypergraph $H = \vcenter{\hbox{{\tikz[baseline=.1ex]{
				\node[vertex] (V1) at (0,0) {};
				\node[vertex] (V2) at (1,0) {};
				\draw[-latex, thick] (V1) to[bend left = 0] node[above] {$a$} (V2);
}}}}$ belongs to $L(FG)$ for some fusion grammar $FG = (Z,F,T,\{\mu\},\{\kappa\})$. This means that, for some multiplicity $m$, $m \cdot Z \Rightarrow_{\fr(P)} Y + G$ such that $Y$ contains a $\mu$-labeled hyperedge and, after removing $\mu$-, $\kappa$-labeled hyperedges from $Y$, one obtains $H$. Such a hypergraph $Y$, for example, can be of the form
$$
Y = 
\vcenter{\hbox{{\tikz[baseline=.1ex]{
				\node[vertex] (V1) at (0,0) {};
				\node[vertex] (V2) at (1.5,0) {};
				\draw[-latex, thick] (V1) to[bend left = 50] node[above] {$a$} (V2);
				\draw[-latex, thick] (V1) to[bend left = 20] node[below] {$\mu$} (V2);
				\node at (0.75,-0.35) {\dots};
				\draw[-latex, thick] (V1) to[bend right = 70] node[below] {$\mu$} (V2);
}}}}
$$
This hypergraph has $n$ parallel $\mu$-labeled hyperedges, and $\rem_{\{\mu\}}(Y) = H$. However, it is not clear how to find an upper bound for $n$ in this case. 

One might try to apply the techniques we have already developed to fusion grammars with markers and connectors. For example, suppose that we want to prove decidability of the following problem:
\begin{problem}[MEM-$1^\prime$]
	\leavevmode
	\\
	\textit{Input.} 
	\begin{enumerate}
		\item A fusion grammar $FG = (Z,F,T,\{\mu\},\emptyset)$ where $\type(\mu) = 2$; 
		\item A hypergraph $H_0 = \vcenter{\hbox{{\tikz[baseline=.1ex]{
						\node[hyperedge] (E) at (0,0) {$a$};
						\node[vertex] (V) at (0.7,0) {};
						\draw[-] (E) -- node[above] {\scriptsize 1} (V);
		}}}}$ with one hyperedge of type $1$ and one vertex $v_0^1$.
	\end{enumerate}
	\textit{Problem MEM-$1^\prime$.} Does $H_0$ belong to $L(FG)$?
\end{problem}
This is an instance of the general membership problem. (We intentionally choose a very simple input hypergraph, namely, $H_0$, to show that there are difficulties even in this case.) $H_0$ belongs to $L(FG)$ if and only if $H_0 = \rem_{\{\mu\}}(Y)$ where $Y$ is a hypergraph with one vertex $v_0^1$, one $a$-labeled hyperedge and some positive number of $\mu$-labeled hyperedges attached to $v$ (so, they are all loops) such that
$
Z \Rightarrow m \cdot Z \underset{\fr(P)}{\Rightarrow} Y + G
$. If one tries to repeat the proof of decidability from Section \ref{ssec_results_mem1}, then the following problem arises. It is now possible that non-white labels of the form $(\mu,f,\tau)$ appear in $Y+G$, because some $\mu$-labeled hyperedges now can belong to the ``main'' hypergraph $Y$ but not to the garbage one $G$. To take this into account, it seems reasonable to change Definition \ref{def_beta} by excluding $\mathfrak{W}$ and also the set of labels of the form $(\mu,f,\tau)$ where $f(1) = f(2) = v_0^1$ from $\FCol$. Indeed, it might seem that the function $f$ tells us that both its attachment vertices are identified with $v_0^1$, so, after removing these $\mu$-labeled hyperedges, we come up with $H_0$. Unfortunately, this is not the case, because colours work well only one-way: they prevent vertices of different colours from being identified. However, it is possible that there are different vertices of the same colour. This would lead us to undesirable consequences; for example, we would mistakenly conclude that $H_0$ belongs to $L(FG)$ in the case where the following hypergraph $Y$ is derivable from $Z$:
$$
Y = 
\vcenter{\hbox{{\tikz[baseline=.1ex]{
				\node[hyperedge] (E) at (0,0) {$a$};
				\node[vertex] (V1) at (0.7,0) {};
				\node[vertex] (V2) at (1.5,0) {};
				\node[vertex] (V3) at (2.3,0) {};
				\node[vertex] (V4) at (3.1,0) {};
				\draw[-latex, thick] (V1) -- node[above] {$\mu$} (V2);
				\draw[-latex, thick] (V2) -- node[above] {$\mu$} (V3);
				\draw[-latex, thick] (V3) -- node[above] {$\mu$} (V4);
				\draw[-] (E) -- node[above] {\scriptsize 1} (V1);
}}}}
$$
Indeed, we could simply colour all $Y$'s vertices in one colour, which would make it indistinguishable from
$$
Y^\prime = 
\vcenter{\hbox{{\tikz[baseline=.1ex]{
				\node[hyperedge] (E) at (0,0) {$a$};
				\node[vertex] (V) at (0.7,0) {};
				\draw[-latex, thick] (V) to[in=-120,out=-60,looseness=30] node[below] {$\mu$} (V);
				\draw[-latex, thick] (V) to[in=-30,out=30,looseness=30] node[right] {$\mu$} (V);
				\draw[-latex, thick] (V) to[in=60,out=120,looseness=30] node[above] {$\mu$} (V);
				\draw[-] (E) -- node[above] {\scriptsize 1} (V1);
}}}}.
$$
How to overcome this issue? Recall that, in order to force vertices to be indentified, we developed the technique involving evidence paths and their encodings. However, we needed to encode the information about \emph{each} pair of vertices that must be identified. Now, since we do not know the number of $\mu$-labeled hyperedges in advance, this would lead us to an infinite number of labels, which is bad. So, the straightforward way of generalising the proof of Theorem \ref{th_mem1} fails, and decidability of the membership problem in general remains open. Probably, to prove it we need to modify the method of evidence paths for the case of arbitrarily many markers and connectors in order not to introduce infinitely many labels. To do this, we were trying to develop and use the concept of \emph{evidence trees} but it did not work as we wanted.

Nevertheless, even bounded fusion grammars are general enough for many purposes. First, we claim that it is enough for applications in DNA computing to assume that the resulting hypergraph has one marker and no connectors. For example, this is enough to model Adleman's experiment using fusion grammars; namely, in \cite{KreowskiKL19}, hypergraphs only of the following form are produced:
$$
\vcenter{\hbox{{\tikz[baseline=.1ex]{
				\node[vertex] (V1) at (0,0) {};
				\node[vertex] (V2) at (0.9,0) {};
				\node[vertex] (V3) at (1.8,0) {};
				\node[vertex] (V4) at (2.7,0) {};
				\draw[-latex, thick] (V1) to[out=210,in=150,looseness=20] node[left] {$\mu$} (V1);
				\draw[-latex, thick] (V1) to[out=120,in=60,looseness=20] node[above] {$0$} (V1);
				\draw[-latex, thick] (V2) to[out=120,in=60,looseness=20] node[above] {$1$} (V2);
				\draw[-latex, thick] (V3) to[out=120,in=60,looseness=20] node[above] {$n-1$} (V3);
				\draw[-latex, thick] (V4) to[out=120,in=60,looseness=20] node[above] {$n$} (V4);
				\draw[-latex, thick] (V1) -- node[below] {} (V2);
				\node at (1.38,0) {$\dotsc$};
				\draw[-latex, thick] (V3) -- node[below] {} (V4);
}}}}
$$
Secondly, bounded fusion grammars are expressive enough to generate languages generated by hyperedge replacement grammars. Indeed, there is Construction 2 from \cite{Lye21_thesis}, which transforms a hyperedge replacement grammar $\mathit{HRG}$ into a fusion grammar $FG(\mathit{HRG})$. Recall the fact that, if $H$ is generated by a hyperedge replacement grammar $\mathit{HRG}$, then there exists its derivation of size $\le f(\vert H \vert)$ for some linear function $f$ \cite[p. 141]{DrewesKH97}. Analysis of \cite[Construction 2]{Lye21_thesis} shows that, if $H$ belongs to the language of the corresponding fusion grammar $FG(HRG)$, then $H = \rem_{\{\mu,\kappa\}}(Y)$ for $Y$ such that $Y$ is derivable by $FG(HRG)$ and such that it contains one marker and at most $f(\vert H \vert) \cdot \vert HRG \vert$ connectors. Hence, the number of the removed markers and connectors is computable, and thus $FG(HRG)$ can be defined as a bounded fusion grammar. 

Summing up, bounded fusion grammars is an expressive enough formalism, for which we know that the membership problem is decidable. I may suppose that the creators of fusion grammars also intended to use a bounded number of markers and connectors. Indeed, usually it suffices to have just one marker, and, since connectors are needed to link connected components, it is sufficient if there are no more connectors than connected components. Now, we are aware that unboundedly many markers and connectors can be an obstacle on the way to decidability.

\section{Parikh's Theorem for Connection-Preserving Fusion Grammars}\label{sec_Parikh}

In this section, we move to investigating the expressive power of fusion grammars. Namely, our goal is to generalise Parikh's well-known theorem. Parikh's theorem for context-free grammars says that the Parikh image of any context-free language is a semilinear set. We are going to prove the same for \emph{connection-preserving} fusion grammars.

\begin{definition}
	Let $T = \{a_1,\dotsc,a_t\}$. Given $H \in \mathcal{H}(\Sigma)$, we define its \emph{Parikh image} $\psi(H) \in \Nat^{t}$ as follows: 
	$$\psi(H) = \left(\#_{a_1}(H),\dotsc,\#_{a_t}(H)\right).$$
\end{definition}

\begin{theorem}\label{th_main_parikh}
	The set $\psi(L(FG)) = \{\psi(H) \mid H \in L(FG)\}$ is semilinear for any connection-preserving fusion grammar $FG = (Z,F,T,\mathcal{M},\mathcal{K})$.
\end{theorem}
Let $F = \{A_1,\dotsc,A_f\}$. Without loss of generality, we assume that $\mathcal{M} = \{\mu\}$ and $\mathcal{K} = \{\kappa\}$. Indeed, if there are several markers (connectors), then we can simulate each of them using the one with the maximal type.

The proof we are going to present is inspired by the quantitative argument discussed before and by Construction 3 from \cite{Lye21_thesis}. The latter construction shows how to transform connection-preserving and \emph{strictly joining} fusion grammars into hyperedge replacement grammars. \emph{Strictly joining} means that, in any derivation of a hypergraph from $L(FG)$, fusion never happens within a single connected component, i.e.~no foldings are possible. This restriction, of course, it too strong, because folding is one of the distinguishing features of fusion grammars, making them more complicated than hyperedge replacement grammars.

Before proving Theorem \ref{th_main_parikh}, let us show a na\"{i}ve way of doing this and discuss why it is wrong. Let us define the function $\Phi:\mathcal{H}(\Sigma) \to \Nat^{t+2f+2}$ as follows:
$$
\Phi(H) = (\#_{a_1}(H),\dotsc,\#_{a_t}(H),\#_{A_1}(H),\dotsc,\#_{A_f}(H),\#_{\overline{A_1}}(H),\dotsc,\#_{\overline{A_f}}(H), \#_{\mu}(H),\#_{\kappa}(H)).
$$
Let us also define the following set.
\begin{definition}\label{def_balance}
	The set $\Balance$ consists of vectors $v \in \Nat^{t+2f+2}$ such that $v(t+i) = v(t+f+i)$ for all $i=1,\dotsc,f$ and such that $v(t+2f+1) > 0$.
\end{definition}
The set $\Balance$ defined above is semilinear. Now, consider the set 
$$
W = \Balance \cap \left\{\sum_{C \in \mathcal{C}(Z)} m(C)\cdot \Phi(C) \middle| m(C) \in \Nat \right\}.
$$
The second set in this intersection is also semilinear (this follows from its definition; it is even linear). Since semilinear sets are closed under intersection \cite{LiuW70}, the set $W$ is semilinear. Each element of $W$ is of the form $\Phi(m \cdot Z)$ for some multiplicity $m$; the number of $A$-labeled hyperedges in $m \cdot Z$ equals the number of $\overline{A}$-labeled ones (for $A \in F$); furthermore, $m \cdot Z$ contains at least one marker. Then, clearly, one can fuse all the hyperedges in $m\cdot Z$ with fusion labels (in any order). Let $m \cdot Z \underset{\fr}{\Rightarrow}^\ast H$ be the result of such fusion. From this, we would like to conclude that $H \in L(FG)$ and that thus $\psi(H)$ is the projection of $\Phi(m \cdot Z)$ onto the first $t$ coordinates. Finally, we would like to infer that $\psi(L(FG))$ equals the projection of $W$ onto the first $t$ coordinates and thus it is semilinear.

The problem of the above reasonings is that the resulting hypergraph $H$ is not necessarily connected. For example, assume that $Z = \vcenter{\hbox{{\tikz[baseline=.1ex]{
				\node[hyperedge] (E1) at (0,0) {$a$};
				\node[hyperedge] (E2) at (1.4,0) {$\mu$};
				\node[vertex] (V) at (0.7,0) {};
				\draw[-] (E1) -- node[above] {\scriptsize 1} (V);
				\draw[-] (E2) -- node[above] {\scriptsize 1} (V);
}}}}$; then $L(FG) = \{\rem_{\{\mu\}}(Z)\}$ and $\psi(L(FG)) = \{1\}$. However, clearly, $W = \{m \cdot \Phi(Z) \mid m \in \Nat,m>0\} = \{(m,m) \mid m \in \Nat,m>0\}$, and its projection onto the first coordinate equals $\Nat \setminus \{0\}$. To avoid this problem we must somehow guarantee that the resulting hypergraph $H$ is connected. We shall do this using context-free grammars and a new notion called \emph{fusion net}.

The definition of a fusion net uses the notion of \emph{undirected unlabeled graph}, which is a triple $U = (V_U,E_U,att_U)$ such that the attachment $att_U:E_U \to \mathcal{P}(V_U)$ assigns a set $\{v_1,v_2\}$ consisting of one or two vertices to each edge from $E_U$. That is, each edge is attached to two vertices, and loops are allowed; the order of the attachment is disregarded. This is a usual notion of a (multi)graph. We are going to use other standard notions of graph theory; they are introduced in numerous textbooks including, e.g., \cite{GrossYA18}. In particular, we shall use the following well-known property of graphs \cite[p. 76]{GrossYA18}: if an undirected unlabeled graph $U$ is connected, then it contains a spanning tree, i.e.~a subgraph which has the same set of vertices $V_U$ and which is a tree.
\begin{definition}
	Given a fusion rule application $H^\prime \underset{\fr(P)}{\Rightarrow} H^{\prime\prime}$, its \emph{fusion net} is the undirected unlabeled graph $FN = (V_{FN},E_{FN},att_{FN})$ such that $V_{FN} = \mathcal{C}(H^\prime)$, $E_{FN} = P$, and such that the following holds. Let $p = \{e_1,e_2\} \in P$ be a pair of hyperedges fused by the application of $\fr(P)$. Assume that $e_i \in E_{C_i}$ ($i=1,2$) for some $C_1,C_2 \in \mathcal{C}(H^\prime)$. Then $att_{FN}(p) = \{C_1,C_2\}$.
\end{definition}

\begin{example}
	The fusion net of the fusion rule application presented in Example \ref{example_der_evidence_path} can be depicted as follows:
	$$
	\vcenter{\hbox{{\tikz[baseline=.1ex]{
					\node[circle,minimum size=32, draw={black}] (C1) at (0,-0.25) {};
					\node[circle,minimum size=32, draw={black}] (C2) at (-0.825,0.825) {};
					\node[circle,minimum size=32, draw={black}] (C2) at (0.825,0.825) {};
					\node[vertex,color={colorB}] (VO) at (0,0) {};
					\node[vertex,color={colorB}] (VL) at (-0.4,-0.4) {};
					\node[vertex,color={colorB}] (VR) at (0.4,-0.4) {};
					\draw[color={colorB},-latex, thick] (VL) -- node[color={colorB},below] {$\scriptstyle a$} (VR);
					\draw[color={colorB},-latex, thick] (VL) -- (VO);
					\draw[color={colorB},-latex, thick] (VR) -- (VO);
					\node[color={colorB}] (Lb1) at (-0.4,-0.11) {$\scriptstyle A$};
					\node[color={colorB}] (Lb6) at (0.4,-0.11) {$\scriptstyle C$};
					\node[vertex,color={colorB}] (VL1) at (-0.6,0.6) {};
					\node[vertex,color={colorB}] (VL2) at (-1.05,1.05) {};
					\draw[color={colorB},latex-, thick] (VL1) to[bend left=45] (VL2);
					\draw[color={colorB},latex-, thick] (VL2) to[bend left=45] (VL1);
					\node[vertex,color={colorB}] (VR1) at (0.6,0.6) {};
					\node[vertex,color={colorB}] (VR2) at (1.05,1.05) {};
					\draw[color={colorB},-latex, thick] (VR1) to[bend left=45] (VR2);
					\draw[color={colorB},-latex, thick] (VR2) to[bend left=45] (VR1);
					\node[color={colorB}] (Lb2) at (-1.2,0.65) {$\scriptstyle \overline{A}$};
					\node[color={colorB}] (Lb3) at (-0.65,1.2) {$\scriptstyle B$};
					\node[color={colorB}] (Lb5) at (1.2,0.65) {$\scriptstyle \overline{C}$};
					\node[color={colorB}] (Lb4) at (0.65,1.2) {$\scriptstyle \overline{B}$};
					\draw[thick,-] (-0.5,-0.12) -- (-1.1,0.55);
					\draw[thick,-] (-0.55,1.2) -- (0.55,1.2);
					\draw[thick,-] (0.5,-0.12) -- (1.1,0.55);
	}}}}
	$$
	Here hypergraphs in circles are vertices of the fusion net, and the lines are edges; for the sake of better understanding, we draw a line going from a hyperedge to the one which it is fused with.
\end{example}

\begin{proposition}\label{prop_connectedness_fn}
	Let $H^\prime \underset{\fr(P)}{\Rightarrow} H^{\prime\prime}$ and let $H^{\prime\prime}$ be connected. Then, the fusion net of this fusion rule application is a connected graph.
\end{proposition}
\begin{proof}
	The proof is ex falso. Assume that the fusion net $FN$ of $\fr(P)$ is disconnected. Then, we can divide the set of vertices $V_{FN} = \mathcal{C}(H^\prime)$ of the fusion net into two non-empty parts such that vertices of different parts are not connected by edges. Let these parts be $\{C_1^1,\dotsc,C_1^{l^1}\}$ and $\{C_2^1,\dotsc,C_2^{l^2}\}$. The hypergraph $H^\prime$ equals $C_1^1 + \dotsc + C_1^{l^1} + C_2^1 + \dotsc + C_2^{l^2} = H^\prime_1 + H^\prime_2$ where $H^\prime_i = C_i^1 + \dotsc + C_i^{l^i}$ (for $i=1,2$). Since $C_1^i$ and $C_2^j$ are not connected by an edge of $FN$ for all $i,j$, this implies that fusions happen only within $H^\prime_1$ and within $H^\prime_2$ but not between them. Therefore, $H^\prime_1+H^\prime_2 \underset{\fr(P)}{\Rightarrow} H^{\prime\prime}_1+H^{\prime\prime}_2 = H^{\prime\prime}$. Thus $H^{\prime\prime}$ is not connected, which is a contradiction.
\end{proof}

The auxiliary preparations are completed, and we can start defining the main construction used to prove Theorem \ref{th_main_parikh}. 
\begin{definition}
	The \emph{linearisation function} is defined as
	$$
	\lambda(H) = {a_1}^{\#_{a_1}(H)} \dotsc {a_t}^{\#_{a_t}(H)} {A_1}^{\#_{A_1}(H)} \dotsc {A_f}^{\#_{A_f}(H)} \overline{A_1}^{\#_{\overline{A_1}}(H)}\dotsc {\overline{A_f}}^{\#_{\overline{A_f}}(H)}\mu^{\#_{\mu}(H)}\kappa^{\#_{\kappa}(H)}.
	$$
	Recall that $a^k$ is the string $a\dotsc a$ where $a$ is repeated $k$ times.
\end{definition}
\begin{remark}
	To be complete, let us introduce the notion of a context-free grammar. A context-free grammar is a quadruple $(N,T,Pr,S)$ where $N$ is a finite set of nonterminal symbols, $T$ is a finite set of terminal symbols, $S \in N$ is a distinguished start symbol, and $Pr$ is a finite set of productions, which are of the form $A \to \alpha$ for $A \in N$ and $\alpha \in (N \cup T)^\ast$. 
	An application of a production $\pi = A \to \alpha$ is simply a replacement of the symbol $A$ in a string of the form $\eta A \theta$ by $\alpha$ resulting in $\eta \alpha \theta$. This is denoted by $\eta A \theta \Rightarrow \eta \alpha \theta$ or $\eta A \theta \underset{\pi}{\Rightarrow} \eta \alpha \theta$. The language generated by a context-free grammar consists of words $w \in T^\ast$ such that $S \Rightarrow^\ast w$ using productions from $Pr$. Such a language is called \emph{context free}.
\end{remark}

\begin{construction}\label{construction_parikh}
	We construct a context-free grammar $CFG(FG)$, which we call \emph{the linearisation of $FG$}. The grammar $CFG(FG)$ is the quadruple $(F \cup \overline{F} \cup \{S\},T \cup \{\mu,\kappa\},Pr,S)$ where $S$ is a new nonterminal symbol and where $Pr$ consists of the productions defined below.
	\begin{enumerate}
		\item Let $C \in \mathcal{C}(Z)$ and let $lab_C(e) \in F \cup \overline{F}$ for some $e \in E_C$. Then $\overline{lab_C(e)} \to \lambda(C-\{e\})$ is in $Pr$. We denote this production by $\pi(C,e)$. {\color{comment} Recall that $\overline{\overline{A}} = A$.}
		\item Let $C \in \mathcal{C}(Z)$ and let $lab_C(e) = \mu$ for some $e \in E_C$. Then $S \to \lambda(C)$ is in $Pr$. We denote this production by $\pi^\mu(C,e)$.
	\end{enumerate}
\end{construction}

\begin{definition}
	$Lin(FG)$ is the set of all words $w \in \Sigma^\ast$ such that $S \Rightarrow^+ w$ using productions of $CFG(FG)$; $w$ may contain both terminal and non-terminal symbols of $CFG(FG)$.
\end{definition}

\begin{definition}
	Let us define the Parikh image for words from $\Sigma^\ast$ as follows:
	$$
	\Psi(w) = (\#_{a_1}(w),\dotsc,\#_{a_t}(w),\#_{A_1}(w),\dotsc,\#_{A_f}(w),\#_{\overline{A_1}}(w),\dotsc,\#_{\overline{A_f}}(w), \#_{\mu}(w),\#_{\kappa}(w)).
	$$
\end{definition}

\begin{definition}
	The \emph{projection function} $\mathit{pr}_t:\Nat^{t+2f+2} \to \Nat^t$ takes a vector and returns its prefix consisting of the first $t$ coordinates: $\mathit{pr}_t(v) = (v(1),\dotsc,v(t))$.
\end{definition}
\begin{remark}
	Clearly, $\mathit{pr}_t(\Psi(\lambda(H))) = \psi(H)$.
\end{remark}
\begin{remark}
	Of course, $Lin(FG) \subseteq \Sigma^\ast$ is context free, so $\Psi(Lin(FG))$ is semilinear.
\end{remark}
\begin{lemma}\label{lemma_parikh_1}
	$\psi(L(FG)) \subseteq \mathit{pr}_t\left(\Psi(Lin(FG))\cap \Balance\right)$.
\end{lemma}
\begin{proof}
	Suppose that $H \in L(FG)$. Let us show that $\psi(H) \in \mathit{pr}_t\left(\Psi(Lin(FG))\cap \Balance\right)$. As $H$ belongs to $L(FG)$, Proposition \ref{prop_nf_cpfg} implies that $H = \rem_{\mathcal{M} \cup \mathcal{K}}(Y)$ such that $Y \in \mathcal{H}(T \cup \mathcal{M} \cup \mathcal{K})$ is a connected hypergraph containing at least one $\mu$-labeled hyperedge, and $Z \Rightarrow n \cdot Z \underset{\fr(Q)}{\Rightarrow} Y$ for some multiplicity $n$ and for some parallelised fusion rule $\fr(Q)$. Take any $C^0 \in \mathcal{C}(n \cdot Z)$ that contains a $\mu$-labeled hyperedge (since there exists at least one $\mu$-labeled hyperedge in $Y$); let $lab_{C^0}(e^0) = \mu$ for a hyperedge $e^0 \in E_{C^0}$ (which we fix). 
	
	Consider the fusion net $FN$ of $n \cdot Z \underset{\fr(Q)}{\Rightarrow} Y$. Since $Y$ is connected, so must be $FN$ (according to Proposition \ref{prop_connectedness_fn}). Let us fix a spanning tree contained in $FN$; let $Q^\prime \subseteq Q$ be the set of its edges. We make this tree rooted by letting $C^0$ be its root. Let us denote its children as $C_{1},C_{2},\dotsc$; let us denote the children of $C_{ i}$ as $C_{ i 1},C_{ i 2},\dotsc$; in general, if we denote a vertex of $FN$ as $C_\alpha$, then let $C_{\alpha 1},C_{\alpha 2},\dotsc$ be the notation for all its children in the spanning tree of $FN$. Thus each vertex of $FN$ except for $C^0$ now is of the form $C_\alpha$ for some word $\alpha \in (\Nat \setminus \{0\})^\ast$ with symbols being positive integers. Finally, let us order the vertices of $FN$ by saying that $C_\alpha \prec C_\gamma$ iff $\alpha$ is less than $\gamma$ according to the lexicographic order; let also $C^0 \prec C_\alpha$ for all $\alpha$. Let $C^0 \prec C^1 \prec C^2 \prec \dotsc \prec C^N$ be the list of all vertices of $FN$ ordered according to $\prec$. 
	
	For $i>0$, consider $C^i$. Take the parent $C^j$ of $C^i$; note that $j<i$. Since $\{C^i,C^j\} \in E_{FN}$, these connected components contain hyperedges that are fused. Let $e^i \in E_{C^i}$, $\overline{e}^i \in E_{C^j}$ be such hyperedges; i.e.~$q^i = \{e^i,\overline{e}^i\} \in Q^\prime$. Consider the following sequence of fusion rule applications:
	\begin{multline*}
		n \cdot Z = C^0+C^1+C^2+C^3+\dotsc+C^N 
		\underset{\fr(\{q^1\})}{\Rightarrow}
		D^1+C^2+C^3+\dotsc+C^N 
		\underset{\fr(\{q^2\})}{\Rightarrow}
		\\
		D^2+C^3+\dotsc+C^N 
		\underset{\fr(\{q^3\})}{\Rightarrow}
		\dotsc
		\underset{\fr(\{q^{N-1}\})}{\Rightarrow}
		D^{N-1}+C^N 
		\underset{\fr(\{q^N\})}{\Rightarrow}
		D^N.		
	\end{multline*}
	This derivation can be viewed as follows: one starts with $C^0$ and joins it step-by-step with the hypergraphs $C^1,\dotsc,C^N$. The hypergraph $D^i$ is the intermediate hypergraph (let $D^0 = C^0$); at the $i$-th step, the fusion joins $D^i$ and $C^i$ resulting in $D^{i+1}$.
	
	This derivation can be translated into a derivation in $CFG(FG)$. Consider the following sequence of productions of $CFG(FG)$:
	\begin{equation}\label{eq_sequence_productions}
		\pi^\mu({C^0,e^0}), \pi({C^1,e^1}), \pi({C^2,e^2}),\dotsc, \pi({C^N,e^N}).
	\end{equation}
	Let one start with $S$ and apply these productions in the same order as they show in (\ref{eq_sequence_productions}). A rule $A \to \alpha$ is always applied to the leftmost occurrence of $A$. The resulting derivation is as follows:
	\begin{equation*}
		S\underset{\pi^\mu({C^0,e^0})}{\Rightarrow}
		w^0 
		\underset{\pi({C^1,e^1})}{\Rightarrow}
		w^1 
		\underset{\pi({C^2,e^2})}{\Rightarrow}
		\dotsc
		\underset{\pi({C^N,e^N})}{\Rightarrow}
		w^N.		
	\end{equation*}
	We claim that $\Psi(w^i) = \Psi(\lambda(D^i))$ for all $i=0,\dotsc,N$. The proof is by induction on $i$. The base case is true because $w^0 = \lambda(C^0) = \lambda(D^0)$. The induction step is proved by observing that the application of the production $\pi({C^{i+1},e^{i+1}})$ transforms $\Psi(\lambda(D^i))$ into $\Psi(\lambda(D^{i+1}))$. As a consequence, $\Psi(w^N) = \Psi(\lambda(D^N))$. The word $w^N$ belongs to $Lin(FG)$.
	
	We know that $D^N \underset{\fr(Q\setminus Q^\prime)}{\Rightarrow} Y$. The latter hypergraph does not contain fusion labels, which means that they all are fused by $\fr(Q \setminus Q^\prime)$. Therefore, for each $A \in F$, it holds that $\#_A(D^N) = \#_{\overline{A}}(D^N)$. This justifies that $\Psi(\lambda(D^N)) \in \Balance$ (along with the fact that $D^N$ must contain at least one $\mu$-labeled hyperedge, because so does $Y$). Concluding, $\psi(H) = \psi(Y) = \psi(D^N) = \mathit{pr}_t(\Psi(\lambda(D^N))) = \mathit{pr}_t(\Psi(w^N))$, and $\Psi(w^N)$ belongs to $\Psi(Lin(FG))\cap \Balance$. This completes the proof.
\end{proof}

\begin{lemma}\label{lemma_parikh_2}
	$\mathit{pr}_t\left(\Psi(Lin(FG))\cap \Balance\right) \subseteq \psi(L(FG))$.
\end{lemma}
\begin{proof}
	The proof is conceptually similar to the previous one turned upside down. Let $w \in Lin(FG)$ and let $\Psi(w) \in \Balance$; we need to show that $\mathit{pr}_t(\Psi(w)) \in \psi(L(FG))$. Consider the derivation of $w$ in $CFG(FG)$. It must start with an application of a production of the form $\pi^\mu(C,e)$ and proceed with applications of productions of the form $\pi(C,e)$. Thus it is of the form
	\begin{equation}\label{eq_der_cfg_fg}
		S\underset{\pi^\mu({C^0,e^0})}{\Rightarrow}
		w^0 
		\underset{\pi({C^1,e^1})}{\Rightarrow}
		w^1 
		\underset{\pi({C^2,e^2})}{\Rightarrow}
		\dotsc
		\underset{\pi({C^N,e^N})}{\Rightarrow}
		w^N = w.		
	\end{equation}
	We want to translate this derivation into a sequence of fusion rule applications. Let us explain this formally (this might be tedious and hence unnecessary). Let us work with \emph{superscribed words}, i.e.~with words of the form $\overset{i_1}{a_1}\dotsc \overset{i_k}{a_k}$ where $a_1,\dotsc,a_k$ are symbols and $i_1,\dotsc,i_k$ are any objects called \emph{superscripts}. Let us use the following notation: if $\alpha = a_1\dotsc a_k$ and $\iota = i_1\dotsc i_k$, then $\overset{\iota}{\alpha} = \overset{i_1}{a_1}\dotsc \overset{i_k}{a_k}$.
	\\
	Let $\pi = (A \to \alpha)$ be a production (here $\alpha$ is an ordinary word); let $\iota$ be another word such that $\vert \alpha \vert = \vert \iota \vert$. Then, an \emph{application of $\pi$ superscribed by $\iota$} is the following rewriting procedure on superscribed words: $\eta \overset{j}{A} \theta \overset{\iota}{\underset{\pi}{\Rightarrow}} \eta \overset{\iota}{\alpha} \theta$. Informally speaking, $A$ is erased along with its superscript $j$, and the word $\alpha$ superscribed by $\iota$ is inserted.
	
	Now, let us return to (\ref{eq_der_cfg_fg}). Let $\widetilde{Z} = C^0 + \dotsc + C^N$. Clearly, $\widetilde{Z}$ is the multiplication of $Z$ since all $C^i$ belong to $\mathcal{C}(Z)$ according to Construction \ref{construction_parikh}. From now on, let us regard $C^i$ as a subhypergraph of $\widetilde{Z}$.
	
	First, given $C^i$ for some $i=1,\dotsc,N$, let $\varepsilon^i = e^i_1\dotsc e^i_{\vert E_{C^i}\vert-1}$ be any word such that $\left\{e^i_1,\dotsc, e^i_{\vert E_{C^i} \vert - 1}\right\} = E_{C^i} \setminus \{e^i\}$ and such that $lab_{\widetilde{Z}}\left(\varepsilon^i\right) = lab_{\widetilde{Z}}\left(e^i_1\right)\dotsc lab_{\widetilde{Z}}\left(e^i_{\vert E_{C^i}\vert-1}\right)$ coincides with the right-hand side of the production $\pi({C^i,e^i})$. Similarly, let $\varepsilon^0 = e^0_1\dotsc e^0_{\vert E_{C^0}\vert}$ be any word such that $\left\{e^0_1,\dotsc, e^0_{\vert E_{C^0} \vert}\right\} = E_{C^0}$ and such that $lab_{\widetilde{Z}}\left(\varepsilon^0\right)$ coincides with the right-hand side of the production $\pi^\mu({C^0,e^0})$. 
	
	With all these new notions, we change (\ref{eq_der_cfg_fg}) as follows: we superscribe $S$ by a blank symbol (its superscript is not important); we superscribe the application of $\pi^\mu({C^0,e^0})$ by $\varepsilon^0$; we superscribe the application of $\pi({C^i,e^i})$ by $\varepsilon^i$ for $i=1,\dotsc,N$. Then, (\ref{eq_der_cfg_fg}) turns into the following derivation:
	\begin{equation}\label{eq_der_cfg_fg_super}
		S\overset{\varepsilon^0}{\underset{\pi^\mu({C^0,e^0})}{\Rightarrow}}
		\tilde{w}^0 
		\overset{\varepsilon^1}{\underset{\pi({C^1,e^1})}{\Rightarrow}}
		\tilde{w}^1 
		\overset{\varepsilon^2}{\underset{\pi({C^2,e^2})}{\Rightarrow}}
		\dotsc
		\overset{\varepsilon^N}{\underset{\pi({C^N,e^N})}{\Rightarrow}}
		\tilde{w}^N.		
	\end{equation}
	{\color{comment}
		Informally, we superscribe each symbol by a hyperedge of a corresponding hypergraph.
	}
	
	Now, using this notation, we can say the following. In (\ref{eq_der_cfg_fg_super}), look at the rule application number $(i+1)$, which is that of $\pi({C^i,e^i})$ (for $i=1,\dotsc,N$); let the symbol, to which it is applied, be superscribed by $\overline{e}^i$ (note that each superscript is a hyperedge from $\widetilde{Z}$). Let $q^i = \{e^i,\overline{e}^i\}$. Finally, we construct a sequence of fusion rule applications in $FG$:
	\begin{equation}\label{eq_der_Parikh_2}
		\widetilde{Z} = C^0 + \dotsc + C^N
		\underset{\fr(\{q^1\})}{\Rightarrow}
		D^1 + C^2 + \dotsc + C^N
		\underset{\fr(\{q^2\})}{\Rightarrow}
		D^2 + C^3 + \dotsc + C^N
		\underset{\fr(\{q^3\})}{\Rightarrow}
		\dotsc
		\underset{\fr(\{q^N\})}{\Rightarrow}
		D^N.		
	\end{equation}
	Here, $D^{i}$ is the result of applying $\fr(\{q^i\})$ to $D^{i-1}+C^i$. Note that $D^N$ is connected because the grammar is connection-preservin, so no disconnection happens.
	
	We know that the Parikh image of $D^N$ is the same as that of $w^N = w$: $\Psi(\lambda(D^N)) = \Psi(w)$. Furthermore, the number of occurrences of $A$ and that of $\overline{A}$ in $w$ are the same for any fusion label $A$ (because $\Psi(w) \in \Balance$). Therefore, we can divide hyperedges of $D^N$ labeled by members of $F \cup \overline{F}$ into pairs of complementary ones. Let us fuse all hyperedges in these pairs. The resulting hypergraph $Y$ obtained from $D^N$ is connected as well since $FG$ is connection-preserving, and it does not contain fusion labels. Finally, since $\Psi(w) \in \Balance$, there must be at least one $\mu$-labeled hyperedge in $Y$. This means that $\rem_{\{\mu,\kappa\}}(Y) \in L(FG)$. It remains to observe that
	$$
	\mathit{pr}_t(\Psi(w))
	=\mathit{pr}_t(\Psi(\lambda(D^N)))
	=\psi(D^N)
	=\psi(Y)
	=\psi(\rem_{\{\mu,\kappa\}}(Y)) \in \psi(L(FG)),
	$$
	as we desired to prove.
\end{proof}

Parikh's theorem for connection-preserving fusion grammars follows from the above lemmas.
\begin{proof}[Proof of Theorem \ref{th_main_parikh}]
	Lemmas \ref{lemma_parikh_1} and \ref{lemma_parikh_2} imply that 
	$
	\psi(L(FG)) = \mathit{pr}_t\left(\Psi(Lin(FG))\cap \Balance\right)
	$. The set $\Psi(Lin(FG))$ is semilinear and so is $\Balance$. Furthermore, semilinear sets are closed under intersection \cite{LiuW70} and projection. Therefore, $\psi(L(FG)) = \mathit{pr}_t\left(\Psi(Lin(FG))\cap \Balance\right)$ is semilinear.
\end{proof}

Using the context-free grammar $CFG(FG)$ is necessary in order to ensure that all connected components $C^0,\ldots,C^N$ in the proof of Lemma \ref{lemma_parikh_2} are combined into a single connected hypergraph $D^N$.

\section{Discussion and Conclusion}\label{sec_discussion}

Fusion grammars are an attractive formalism because they are quite simple as well as hyperedge replacement grammars. Indeed, fusion rules are applied in a context-free manner, unconditionally. They can be freely interchanged, several rules can be combined into a single parallelised one and, conversely, a parallelised rule can be sequentialised. Fusion grammars have no mechanisms to control rule applications and thus no source of undecidability. This is justified by the decidability results of this article. 

On the other hand, fusion grammars are not as simple as hyperedge replacement grammars due to the folding feature. If fusion happens between two connected hypergraphs $C_1$ and $C_2$ (where $e_1 \in E_{C_1}$ and $e_2 \in E_{C_2}$ are fused), then one can view this as a hyperedge replacement rule, namely, the one where the hyperedge $e_1$ in $C_1$ is replaced by $C_2-\{e_2\}$ (the hyperedge $e_2$ points at the ``gluing points'' for the hyperedge replacement). So, fusion between two hypergraphs is essentially hyperedge replacement. However, folding, i.e.~fusion within a connected hypergraph, cannot be simulated by hyperedge replacement. Thanks to folding, one can generate the language of all graphs or the language of pseudotori by means of fusion grammars \cite{KreowskiKL17}, which cannot be generated by hyperedge replacement grammars. 

We know that languages generated by (bounded) fusion grammars are decidable but many interesting questions concerning the class of languages generated by fusion grammars still remain open. Does Parikh's theorem hold for fusion grammars (with disconnections)? Do fusion grammars generate languages ``presumably harder than NP'' (e.g.~a PSPACE-complete one)? Can the class of fusion grammars' languages be characterised as the class of languages definable in some logic?

Also, it is interesting to explore the difference between fusion grammars and connection-preserving fusion grammars. Do they generate the same class of languages? The property of being connection-preserving is of interest itself. It is mentioned in \cite{Lye21_thesis} briefly that ``in general, the preservation of connectedness is difficult to check.'' However, this statement is not elaborated on there. Is it decidable to check whether a fusion grammar is connection-preserving?

\section*{Acknowledgments}

I am grateful to Prof. Hans-J\"{o}rg Kreowski and to Aaron Lye for discussing this work with me. I also thank the reviewers for their valuable remarks and suggestions.

\bibliographystyle{plain}
\bibliography{Decidability_Expressive_Power_Fusion_Grammars}

\newpage
\appendix

\section{Proof of Proposition \ref{prop_nf_cpfg}}\label{appendix_proof_prop_nf_cpfg}


\begin{lemma}
	Let $H^1 \underset{\fr}{\Rightarrow} H^2$ be a fusion rule application in a connection-preserving fusion grammar. Let $H^2 = C_1^2 + \dotsc + C_{l}^2$. Then $H^1$ can be represented as the disjoint union $H^1 = C_1^1+\dotsc + C_{l}^1$ such that for each $i=1,\dotsc,l$ either $C_i^1 = C_i^2$ or $C_i^1 \underset{\fr}{\Rightarrow} C_i^2$.
\end{lemma}
\begin{proof}
	The fusion rule application of interest fuses two hyperedges, which we denote as $e,\overline{e}$. There are two possible cases. The first one is where both $e$ and $\overline{e}$ belong to the same connected component $C$ in $H^1$; hence $H^1 = X + C$. Then $X + C \underset{\fr}{\Rightarrow} X + D = H^2$ where $D$ is connected by the definition of connection-preserving fusion grammars. The second one is where $e$ and $\overline{e}$ belong to different connected components in $H^1$, namely, to $C^\prime$ and $C^{\prime\prime}$ resp. Hence $H^1 = X + C^\prime+C^{\prime\prime}$ and $X + C^\prime+C^{\prime\prime} \underset{\fr}{\Rightarrow} X + D = H^2$ where $D$ is connected, again by the definition of connection-preserving fusion grammars.
	
	In both cases, we have that $H^1 = X + \widetilde{C}$ for $\widetilde{C}$ consisting of either one or two connected components and $X+\widetilde{C} \underset{\fr}{\Rightarrow} X+D = H^2$ for $D$ being connected. The latter implies that $D$ is contained in some $C^2_i$ for some $i$; that is, $C^2_i = Y+D$. Note that, if $D$ had not been connected, it could have been the case that its parts belong to different connected components $C^2_i$ but not to any single one. Without loss of generality, let $i=1$: $C^2_1 = Y+D$. In such a case, the statement of the proposition holds for $C^1_1 = Y + \widetilde{C}$ and $C^1_i = C^2_i$ for $i>1$. 
\end{proof}

\begin{corollary}\label{corollary_nf_cpfg_main}
	Let $H^1 \underset{\fr}{\Rightarrow}^\ast H^2$ be a sequence of fusion rule applications in a connection-preserving fusion grammar. Let $H^2 = C_1^2 + \dotsc + C_{l}^2$. Then $H^1$ can be represented as the disjoint union $H^1 = C_1^1+\dotsc + C_{l}^1$ such that for each $i=1,\dotsc,l$ it holds that $C_i^1 \underset{\fr}{\Rightarrow}^\ast C_i^2$.
\end{corollary}
Now, we can prove Proposition \ref{prop_nf_cpfg}.

\begin{proof}[Proof of Proposition \ref{prop_nf_cpfg}]
	The ``if'' direction is obvious; let us prove the ``only if'' one. Let $H \in L(FG)$. Then $H = \rem_{\mathcal{M} \cup \mathcal{K}}(Y)$ for some connected hypergraph $Y \in \mathcal{H}(T \cup \mathcal{M} \cup \mathcal{K})$ with at least one $\mathcal{M}$-labeled hyperedge such that $Z \Rightarrow^\ast Y + X$ for some $X$. Proposition \ref{prop_nf} entails that there exists a most parallelised derivation of the form $Z \Rightarrow m \cdot Z \underset{\fr}{\Rightarrow}^\ast Y + X + G$ for some $G$. What we want is to transform this derivation in such a way that the resulting hypergraph is $Y$ without the additional summands $X + G$.
	
	Let us apply Corollary \ref{corollary_nf_cpfg_main} to the fusion rule applications $m \cdot Z \underset{\fr}{\Rightarrow}^\ast Y + X + G$ with $C^2_1 = Y$, $C^2_2 = X+G$. As a consequence, $m \cdot Z = C^1_1+C^1_2$ where $C^1_1 \underset{\fr}{\Rightarrow}^\ast Y$. The subhypergraph $C^1_i$ ($i=1,2$) of $m \cdot Z$ must be the disjoint union of some connected components of $m \cdot Z$; therefore, $C^1_i = m_i \cdot Z$ such that $m_1(C)+m_2(C)=m(C)$ for all $C \in \mathcal{C}(Z)$. Using $n = m_1$ we construct the following derivation:
	$$
	Z \Rightarrow n \cdot Z \underset{\fr}{\Rightarrow}^\ast Y 
	$$
	It remains to replace the sequence of fusion rules by the parallelised fusion rule $\fr(Q)$ application (see Remark \ref{rem_parallelization}).
\end{proof}

\end{document}